\documentclass[letterpaper,11pt]{article}
\usepackage{tikz}
\usepackage[T1]{fontenc}
\usepackage[utf8]{inputenc}
\usepackage[margin=1in]{geometry}

\usepackage{amsmath,amssymb,amsthm}
\usepackage{thmtools}
\usepackage{mathrsfs}
\usepackage{xcolor}
\usepackage{forest}
\usepackage{subcaption}
\usepackage{graphicx} 
\usepackage{authblk}
\usetikzlibrary{arrows.meta, calc, decorations.pathreplacing}

\declaretheorem[numberwithin=section]{theorem}

\declaretheorem[sibling=theorem]{corollary}
\declaretheorem[sibling=theorem]{lemma}
\declaretheorem[sibling=theorem]{definition}

\declaretheorem[sibling=theorem]{observation}

\declaretheorem[sibling=theorem]{claim}

\usepackage[hypertexnames=false]{hyperref}
\hypersetup{
    colorlinks=false,
    linkcolor=blue,
    urlcolor=cyan,
}

\usepackage[capitalize]{cleveref}

\crefname{observation}{observation}{observations}
\Crefname{observation}{Observation}{Observations}
\crefname{claim}{claim}{claims}
\Crefname{claim}{Claim}{Claims}

\newcommand{\Lefts}{\mathbb{L}}
\newcommand{\LZend}{\mathcal{Z}}

\newcommand{\nil}{\mathsf{nil}}

\newcommand{\pre}{\mathsf{pre}}
\newcommand{\delpre}{\mathsf{\delta}\text{-}\pre{}}

\DeclareMathOperator{\polylog}{polylog}
\DeclareMathOperator{\poly}{poly}
\newcommand{\val}{\mathsf{val}}
\newcommand{\K}{\mathbb{K}}

\newcommand{\Split}{\mathsf{Split}}
\newcommand{\Merge}{\mathsf{Merge}}
\newcommand{\Shift}{\mathsf{Shift}}
\newcommand{\Parent}[3]{
  \mathsf{\pi}^{(#2)}_{#3}(#1)
}

\newcommand{\floor}[1]{\left\lfloor #1 \right\rfloor}
\newcommand{\ceil}[1]{\left\lceil #1 \right\rceil}
\newcommand{\Bad}{\mathsf{bad}}

\title{Random Access to LZ-End: Faster and Deterministic}

\author[1]{Itai Boneh}
\author[1]{Paweł Gawrychowski}

\affil[1]{University of Wrocław}

\date{}

\begin{document}

\maketitle
\thispagestyle{empty}

\begin{abstract}
The LZ-End parsing of a length-$n$ string is a variation of Lempel-Ziv compression introduced by Kreft and Navarro [DCC 2010],
motivated by the lack of a linear-size structure with $O(\log n)$ access time for the classical variant.
While the original paper was only able to provide efficient extraction from the
phrase boundaries, recently Kempa and Saha [SODA 2022] established that, for a string $S$ whose LZ-End parsing consists of $z$ phrases, there exists a random access data structure that uses $O(z)$ space and guarantees $O(\log^{4}n \cdot \log\log n)$ query time.
However, their proof does not yield an efficient construction algorithm, and their data structure is inherently randomized.

We resolve both limitations by providing a deterministic, $O(z)$-space data structure that supports random access queries in polylogarithmic time and can be constructed in $O(z\log^{2}(n/z))$ time directly from the LZ-End parsing.
In addition to eliminating randomness and providing an efficient construction algorithm, the query time of our data structure is $O(\log^{2}(n/z))$, significantly improving upon the query time of Kempa and Saha.

We also show that our techniques can be used to support the more general substring-extraction.
Namely, we present a data structure with the same space and the same construction time that given two indices $i$ and $j$, outputs $S[i..j]$ in $O(j-i+\log^2\frac{n}{z})$ time.
\end{abstract}

\newpage

\section{Introduction}
The natural approach to compressing data is to identify its repeating fragments. Then,
for every subsequent occurrence of such a fragment, instead of spelling it out explicitly again, we
can write down a pointer to the previous occurrence. This is the underlying high-level idea of the
well-known Lempel-Ziv compression scheme~\cite{DBLP:journals/tit/ZivL77} used in popular
compressors such as \texttt{ZIP}, \texttt{gzip}, \texttt{RAR}, \texttt{PKZIP}, and \texttt{PNG}.

While being able to compress the input to obtain its hopefully significantly smaller compressed
representation that can be still decompressed to obtain the original input is already interesting,
in most applications we would actually prefer to avoid the decompression and instead directly
operate on the compressed representation. Ideally, we would like to augment the compressed
representation with a data structure of roughly the same size that allows querying the original input.
We refer to Navarro's survey~\cite{DBLP:journals/csur/Navarro21} for an overview of such results
for different compression methods. While the ultimate goal is designing a \emph{compressed index},
which is a data structure that allows for efficient pattern matching queries~
\cite{DBLP:conf/dcc/BelazzouguiGGKO15,DBLP:journals/tcs/BilleEGV18,DBLP:journals/talg/ChristiansenEKN21,DBLP:journals/tcs/KreftN13,DBLP:journals/dam/NishimotoIIBT20,DBLP:conf/latin/GagieGKNP14,DBLP:conf/spire/TakabatakeTS15,DBLP:journals/jcss/ClaudeNP21},
the first step (and in fact an important building block in many of the compressed indexes) is
a compressed random access structure.

A compressed random access structure for a string $S[1..n]$ should use space roughly proportional
to the size of the compressed representation of the input and allow for extracting any character $S[i]$ efficiently.
Ideally, we would like such a query to take logarithmic time in the length of the string.
Whether this is known to be possible depends on the chosen compression method.

A Lempel-Ziv (LZ77) representation of a string $S[1..n]$ is a factorization
of the string into $z$ blocks called phrases: $S[1..n] = P_{1} P_{2} \dots P_{z}$. Each phrase $P_{i}$
is either a single character, that is, $|P_{i}|=1$, or has an earlier occurrence in $S$. This allows
us to encode each phrase is constant space by either specifying its only character, or the position
and the length of the previous occurrence. It is known \cite{DBLP:journals/tit/LempelZ76} that
greedily choosing each phrase to be as long as possible results in the smallest number of phrases.
A related (provably weaker, but easier to operate on) compression method is a straight-line
program (SLP) representation, which is simply a context-free grammar with exactly one production
for each non-terminal that derives exactly one string $S$. It can be seen that an SLP consisting
of $g$ productions can be converted into a LZ77 parse of size $g$. In the other direction,
Charikar et al.~\cite{DBLP:journals/tit/CharikarLLPPSS05} and Rytter~\cite{DBLP:journals/tcs/Rytter03} independently
established that a LZ77 parse of size $z$ can be converted into a SLP consisting of
$O(z\log(n/z))$ productions. In fact, the obtained SLP is \emph{balanced}, meaning that the
depth of the derivation tree is $O(\log n)$. This has the advantage of automatically providing
a compressed random access structure with query time $O(\log n)$ and size $O(z\log(n/z))$,
which simply operates by descending down in the grammar. Verbin and Yu~\cite{DBLP:conf/cpm/VerbinY13}
proved that, roughly speaking, this is very close to optimal. Namely, any structure of size $O(z\polylog n)$
must have $\Omega(\log n/\log\log n)$ query time, for the case where $n=z^{1+\epsilon}$.
This however does not exclude the possibility of providing a compressed random access structure with size roughly proportional
to the size of the compressed representation and, say, logarithmic query time.

This goal has been in fact achieved for grammar compression. Bille et al.~\cite{DBLP:journals/siamcomp/BilleLRSSW15}
showed that, assuming the Word RAM model, given a SLP representation of size $g$ representing a string $S[1..n]$,
we can build a data structure of size $O(g)$ that supports random access queries in $O(\log n)$ time.
Later, Ganardi, Jeż, and Lohrey~\cite{DBLP:journals/jacm/GanardiJL21} showed a stronger version of this statement:
any SLP of size $g$ representing a string $S[1..n]$ can be converted into a balanced SLP of size $O(g)$ representing the same string.
This can be seen as the ultimate random access structure for SLPs, although the lower bound of Verbin and Yu~\cite{DBLP:conf/cpm/VerbinY13}
still leaves the possibility of designing a structure of size $O(g)$ and allowing random access in optimal $O(\log n/\log\log n)$ time.
Somewhat surprisingly, the best known bounds for a LZ77 compressed random access structure
follow by converting the LZ77 parse to a grammar, and achieve $O(z\log(n/z))$ space with $O(\log \frac n z)$ query time
(see~\cite{DBLP:journals/jcss/BelazzouguiCGGK21} for an alternative solution achieving the same time-space bounds).
The logarithmic space increase is a major drawback of such a solution: after all, the compression ratio
might actually be of a similar magnitude (as is the case for a random input), so by building such a structure
we are not benefiting at all from using a compressed representation. Fully understanding
the complexity of random access for LZ77 compressed strings remains to be a major open problem in the area.

Motivated by the lack of progress on understanding the complexity of random access for LZ77 compressed strings,
Kreft and Navarro~\cite{DBLP:journals/tcs/KreftN13} introduced a restricted variant of LZ77, called LZ-End,
in which the previous occurrence of each phrase needs to end at a phrase boundary. With this restriction,
the greedy parsing is not necessarily optimal.
However, finding the optimal LZ-End parsing of a string was shown to be NP-hard~\cite{BFKNSU23}, when the greedy parse can be found efficiently~\cite{DBLP:conf/esa/KempaK17},
and still behaves very well in practice~\cite{DBLP:conf/dcc/KempaK17}.
While already the original paper established that extracting substrings ending at a boundary
can be done efficiently with only $O(z)$ space for LZ-End parse consisting of $z$ phrases (in fact, this can
be seen as the reason for restricting the parsing in such a way), it was not known whether LZ-End
admits better compressed random access structures than that of a LZ77 until 10 years later Kempa and Saha~\cite{KS22}
showed that, in fact, LZ-End does admit a data structure of size $O(z)$ that supports random access queries
in polylogarithmic time. However, their construction has two major drawbacks.
\begin{enumerate}
\item The original proof is purely existential: that is, they only showed that, for any string $S[1..n]$, there
exists a data structure of size $O(z)$ with polylogarithmic query time. The construction crucially uses probabilistic method,
so we can either verify whether the obtained data structure is correct after construction, or we can compromise on
a data structure for which the bound on the query time holds with high probability.
\item The guarantee on the query time in the original proof, while polylogarithmic, is actually quite high:
$O(\log^{4}n\cdot \log\log n)$. This should be compared with the only known lower bound, which is
$\Omega(\log n/\log\log n)$ for structures of size $O(z\polylog n)$. Further, the time-space product of the original structure
(a natural measure for data structures) is as high as $O(z\log^{4}n\cdot\log\log n)$, while for the LZ77 compression
it is a much more reasonable $O(z\log^{2}(n/z))$.
\end{enumerate}

\paragraph{Our results.}
Our contribution is twofolds. First, we adapt the high-level idea used by Kempa and Saha to obtain,
for a string of length $n$ described by a LZ-End parse consisting of $z$ phrases, a random access structure
of size $O(z)$ and polylogarithmic query time that can be constructed deterministically in $\tilde O(z)$ time.
As explicitly mentioned by Kempa and Saha, their data structure can be either constructed 
in expected $O(\poly(n))$ time with worst-case queries or in worst-case $O(\poly(n))$ time
with expected-time queries. Thus, we not only completely remove the randomization, but further
bring down the construction time to depend on the size of the compressed representation of the text
instead of its length. Second, with further combinatorial insight we are able to significantly improve the query time,
which was as high as $O(\log^{4}n\cdot \log\log n)$ in the original construction. More specifically,
our query time is only $O(\log^{2}(\frac{n}{z}))$. 

\begin{theorem}\label{thm:main}
    Given an LZ-End factorization $\LZend$ consisting of $z$ phrases of a string $S[1..n]$, we can compute in time
    $O(z \log^2 (\frac{n}{z}))$ a data structure of size $O(z)$ supporting a random access query in $O( \log^2 (\frac{n}{z}))$ time.
\end{theorem}

\noindent We note that the time-space tradeoff of our structure is $O(z\log^{2}(\frac{n}{z}))$, which matches the best known tradeoffs
for the LZ77 random access structure.

Our techniques can be applied and enhanced to additionally support the stronger substring extraction query.

\begin{theorem}\label{thm:extraction}
     Given an LZ-End factorization $\LZend$ consisting of $z$ phrases of a string $S[1..n]$, we can compute in time
    $O(z \log^2 (\frac{n}{z}))$ a data structure of size $O(z)$ that given two indices $i,j\in [n]$ outputs $S[i..j]$ in $O(j-i + \log^2 (\frac{n}{z}))$ time.
\end{theorem}

\section{Preliminaries}
\paragraph{Integer intervals.}
We use standard notation to denote consecutive sets of integers.
For integers $i,j$, we denote $[i..j] = \{i,i+1,\ldots, j \}$ (if $j < i$ then $[i..j] = \emptyset$).
We also denote $[i] = [1..i]$, and $(i..j)=(i..j-1]=[i+1..j)=[i+1..j-1]$.

\paragraph{Strings.}
A string $S$ with length $n$ over alphabet $\Sigma$ is a sequence of symbols $S[1]S[2] \ldots S[n]$ where $S[i] \in \Sigma$ for every $i\in [n]$.
For integers $i,j \in [n]$ we denote $S[i..j]=S[i]S[i+1]..S[j]$.
We call $S[i..j]$ a substring of $S$.

\paragraph{LZ-End factorization.}
An LZ-End factorization of a string $S[1..n]$ is a partition of $S$ into substrings $P_1,P_2,\ldots P_z$ called phrases.
For each $i$, we have $P_i=S[a_i..b_i]$, with $a_1 = 1$, $b_z=n$, and $a_i = b_{i-1} + 1$ for every $i \in (1..z]$. 
Each phrase satisfies one of the following conditions.
Either $P_i$ is a phrase of length one, i.e. $P_i = S[a_i]$, we call such a phrase a terminal phrase.
Or, $P_i = S[a_i..b_i]$, and there is some reference phrase $P_j$ with $j<i$ such that $S(b_j - |P_i|..b_j] = P_i$.
We call such a phrase a reference phrase, and say that $P_j$ is the source of $P_i$.
When we are given an LZ-End factorization, we assume that every reference phrase $P_i$ is given alongside the index $j$ of its source phrase.
We stress that, in contrast to prior work~\cite{KS22},
we are not requiring that the given parsing is greedy, that is, it is not necessarily the case that every
$P_{i}$ is the longest phrase that occurs earlier ending at a phrase boundary.

In \cref{apx:no-big-phrases}, we show that any LZ-End factorization with $z$ phrases of a string with length $n$ can be efficiently transformed into a similar size LZ-End factorization of the same string, with all the phrases having size at most $\frac{n}{z}$.
\begin{restatable}{lemma}{nobigphrases}
\label{lem:no-big-phrases}
There is an algorithm that given $\LZend$, an LZ-End partition of a string $S$, returns $\LZend'$,  an LZ-End partition of $S$ with size at most $4z$ such that every phrase $P'\in \LZend'$ has $|P'| \le \frac{n}{z}$.
Here, $n=|S|$ and $z=|\LZend|$.
The algorithm runs in time $O(z)$.    
\end{restatable}

Throughout the paper, we often partition integer from $[0..\frac{n}{z}]$ into exponential levels.
We fix the notation $\K =\{ 1.5^i \mid i\in [0..\log_{1.5}\frac{n}{z}]\}$ for the set of integer powers of $1.5$ smaller than $\frac{n}{z}$.

\paragraph{Computational model.} We describe our algorithm in the standard Word RAM model, see e.g.~\cite{DBLP:conf/stacs/Hagerup98}.
We assume that $n$ fits in a single machine word, and basic arithmetical operations on numbers that fit in a constant
number of machine words take constant time. In particular, constant-time indirect addressing is available, allowing us
to implement arrays with constant lookup time. In this model, we have the following result (called deterministic perfect hashing).

\begin{lemma}[\cite{DBLP:conf/icalp/Ruzic08}]
\label{lem:hashing}
Given a set $S$ consisting of $n$ integers, we can build in deterministic $O(n(\log\log n)^2)$ time a data structure
of size $O(n)$ that allows checking if $x\in S$ and accessing its associated information in $O(1)$ time.
\end{lemma}

\section{Framework and Overview}

Let us start by establishing notation and terminology to be used for the rest of this section.
We are given as input an LZ-end factorization $\LZend=P_1,P_2,\ldots, P_z$ of string $S$ with length $n$.
Each phrase is given with its endpoints $a_i$ and $b_i$ such that $P_i = S[a_i..b_i]$.
Additionally, each phrase $P_i$ with length at least $2$ is given with the index $j<i$ such that $P_j$ is the source of $P_i$.
For every index $i\in[1..n]$ in the text, we denote as $r_i$ and $\ell_i$ the distance to the closest phrase boundary to the right of $i$ and to the left of $i$, respectively.
Formally, if the phrase containing index $i$ is $P=S[a..b]$, then $r_i = b-i$ and $\ell_i = i-a$.
We call $r_i$ and $\ell_i$ the $r$-value and the $\ell$-value of $i$, respectively. 

\subsection{Intuition and High-level Approach}
Let us first describe the approach of Kempa and Saha \cite{KS22}.
We are storing an array $B$ with $B[i] = S[b_i]$ for every $i\in [z]$.
For an index $i\in [n]$ with $r_i = 0$, we have that $i = b_j$ for some $j\in [z]$ and we can output $S[i] = B[j]$.
Our goal is to reduce an arbitrary index $i\in [n]$ given at query time to an index $i' \in [n]$ with $r_{i'} = 0$ and $S[i] = S[i']$.
Kempa and Saha \cite{KS22} achieve this goal by defining a 'step' rule $N(i)$ for every index $i\in [n]$.
This step rule maps $i$ with $r_i > 0$ to another index $N(i)<i$ such that $S[i] = S[N(i)]$.
The rule defined by Saha and Kempa is both \textit{compactly computable} and \textit{terminates quickly}.
By compactly computable, we mean that there is a data structure with size $O(z)$ that allows computing $N(i)$ in polylogarithmic time.
By terminates quickly, we mean that for every $i\in [n]$, repeatedly applying $N(i),N^2(i),\ldots$ will result in some index $i'$ with $r_{i'}=0$ after a polylogarithmic number of steps.
It should be clear that given such a step rule, one can support random access in $O(z)$ space and polylogarithmic time.

We provide a high level description of the step rule of \cite{KS22}.
The step function $N$ is a composition of three different functions: $J$, $sJ$, and $M$ (\textit{"Jump"}, \textit{"Stable Jump"}, and \textit{"Marked"}).
The function $J$ is defined for every index $i\in [n]$ with $r_i > 0$, while $sJ$ and $M$ are only defined for a strict subset of those indices.
For an index $i$ with $r_i > 0$, we will prioritize using $M(i)$ as the step from $i$, i,e, we have $N(i)=M(i)$ if $M(i)$ is defined.
Otherwise, we set $N(i) = sJ(i)$, and only if both $M(i)$ and $sJ(i)$ are undefined, we have $N(i)=J(i)$.

Consider a sequence $i,N(i),N^2(i),\ldots$ of $N$ steps.
Each of $J$, $sJ$, and $M$ have a particular role in guaranteeing that this sequence quickly converges to an index $i'$ with $r_{i'}=0$.
This goal is achieved in $O(\log n)$ step sequences called \textit{epochs}.
An epoch is a subsequence of $N$ steps such that at the end of the epoch, the $r$-value is in a smaller exponential level than the one seen at the start of the epoch.
Since initially $r_i\le n$, after applying $O(\log n)$ epochs we reach an index with $r$-value $0$ as required.
Let us formally define the goal of the epoch.
We fix the set $\K_n = \{1.5^k \mid k \in [0..\log_{1.5}n] \}$ of integer powers of $1.5$ smaller than $n$.
An epoch starts at an index $i$ that has $r_i \in [k..1.5k)$ for some $k\in \K_n$, and ends upon reaching an index $j$ with $S[i] =S[j]$ and $r_{j}<k$.

Let us describe the role of each function in an epoch.
First of all, all the functions share the property of being $r$-non increasing.
That is, we will have $r_{F(i)}\le r_i$ for every $F\in \{J,sJ,M\}$ for which $F(i)$ is defined.

The function $J$ will act as a naive 'searching' step. 
It transforms $i$ into an index to the left of $i$, without increasing the $r$-value and while maintaining $S[i] = S[J(i)]$.
We are not guaranteed to have any progress towards lowering the value of $r$ when using $J$, but we 'hope' that by repeatedly applying $J$, we will eventually reach an index for which $M$ is defined.
We call the part of the epoch before $M$ is activated for the first time \textit{the naive jumping} part.

The function $M$ has the role of a \textit{trigger} function and a \textit{progress} function, and it has a symbiotic relationship with $sJ$.
By trigger, we mean that after $M(i)$ applies, the naive jumping part of the epoch ends, and it enters a different state called the \textit{stable part}.
The epoch remains in the stable part until the end of the epoch.
During the stable part, all $N$ steps are guaranteed to be either $M$ or $sJ$ until we will eventually reach some index $i'$ with $r_{i'} < \frac{2}{3}r_i$, finishing the epoch.

When we say that $M$ is a progress function, we mean that $M(i)$ is guaranteed to have a significantly reduced $\ell$-value or a significantly reduced $r$-value.
Namely, it holds that $\min(r_{M(i)},\ell_{M(i)})\le \min(r_i/2,\ell_i /2)$.
Following this notion, we call an index $i$ for which $M(i)$ is defined a \textit{progress index}.

The key property of $sJ$ is that it is $\ell$-non-decreasing, in the following sense.
For every index $i$ for which $j=sJ(i)$ is defined, either $\ell_j\le \ell_i$ or $r_j < \frac{2}{3}r_{i}$.
Notice that in the latter case, the epoch is finished.
In words, $sJ$ guarantees that in each step we either do not increase $\ell$, or we achieve the end of the epoch.
We note that the function $M$ also has this property.

We have that throughout the stable part, both $r$ and $\ell$ values are non-increasing.
Furthermore, every time we reach a progress index, we have that $r$ or $\ell$ are reduced by a constant factor.
Since $\ell_i \le n$, it can only be reduced by a constant factor $O(\log n)$ times throughout the epoch. 
It follows that after visiting $O(\log n)$ progress indices, we will finally reach some $i'$ with $r_{i'}<\frac{2}{3}r_i$, finishing the epoch. 

The last important property of $M$ is the property of being common.
That is, for every index $i$, it is guaranteed that withing a polylogarithmic number of $N$ steps, we reach a progress index.

Therefore, we are guaranteed to reach a progress index after a polylogarithmic number of steps, triggering the start of the stable part.
The epoch ends after we visit $O(\log n)$ progress indices, and we reach a new progress index every polylogarithmic number of steps.
It follows that the epoch ends within a polylogarithmic number of steps.

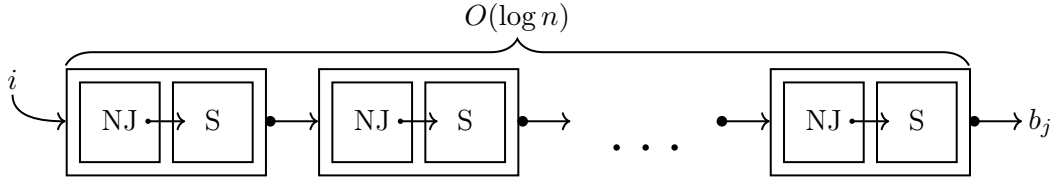
\begin{figure}[h]
    \centering
    \tikzset{every picture/.style={line width=0.75pt}} 

\tikzset{
    every node/.style={
        text height=1.5ex,
        text depth=0.25ex
    }
}
\usetikzlibrary{arrows.meta}
\usetikzlibrary{decorations.pathreplacing}
\begin{tikzpicture}[x=1pt,y=1pt,yscale=-1,xscale=1]
\pgfmathsetmacro{\height}{40}
\pgfmathsetmacro{\width}{75}
\pgfmathsetmacro{\epochshift}{20}
\pgfmathsetmacro{\x}{0} 
\pgfmathsetmacro{\njxoffset}{5} 
\pgfmathsetmacro{\stablexoffset}{40}
\pgfmathsetmacro{\syoffset}{5}

\pgfmathsetmacro{\swidth}{30}
\pgfmathsetmacro{\sheight}{30}

\node at (-20,5) {$i$};
\draw[->] (-20,10) to[out=90, in=180] (0,\height/2);

\foreach \i in {1,...,2} {
    \node[draw, minimum width=\width, minimum height=\height, anchor=north west] at (\x,0) {};
    
    \node[draw, minimum width =\swidth,minimum height=\sheight, anchor=north west] (nj) at (\x+\njxoffset, \syoffset ){};
    \node[draw, minimum width =\swidth,minimum height=\sheight, anchor=north west] (stable) at (\x+\stablexoffset, \syoffset ){};
    \node at (stable.center) {S};
    \node at (nj.center) {NJ};
    \draw[{Circle[length=4pt]}->] (\x+\width,\height/2) -- (\x+\width+\epochshift,\height/2);
     \draw[{Circle[length=2pt]}->] (\x+\njxoffset+\swidth-5,\height/2) -- (\x+\stablexoffset+5,\height/2);
    \pgfmathparse{\x+\width+\epochshift}
    \global\let\x=\pgfmathresult 
  }
\node[anchor=north west] at (\x+10,\height/2) {\Huge $\dots$};
   
    \pgfmathparse{\x+\width}
    \global\let\x=\pgfmathresult 

     \draw[{Circle[length=4pt]}->] (\x-\epochshift,\height/2) -- (\x,\height/2);

\node[draw, minimum width=\width, minimum height=\height, anchor=north west] at (\x,0) {};
 
    \node[draw, minimum width =\swidth,minimum height=\sheight, anchor=north west] (nj) at (\x+\njxoffset, \syoffset ){};
    \node[draw, minimum width =\swidth,minimum height=\sheight, anchor=north west] (stable) at (\x+\stablexoffset, \syoffset ){};
    \node at (stable.center) {S};
    \node at (nj.center) {NJ};
    \draw[{Circle[length=4pt]}->] (\x+\width,\height/2) -- (\x+\width+\epochshift,\height/2);
     \draw[{Circle[length=2pt]}->] (\x+\njxoffset+\swidth-5,\height/2) -- (\x+\stablexoffset+5,\height/2);

     \node at (\x+\width+\epochshift+7,\height/2) {$b_j$};

      \draw[decorate,decoration={brace,amplitude=10pt}]
    (0,-1) -- (\x+\width,-1)
    node[midway,above=10pt]{$O(\log n)$};

\end{tikzpicture}
    \caption{The high level structure of the query algorithm of Kempa and Saha.
    Starting with the input index $i$, we apply a sequence of $O(\log n)$ epochs (depicted as larger rectangles).
    Every epoch consists of two parts: the naive jumping part (depicted as an 'NJ' labeled square) and the stable part ('S' labeled square). 
    Each epoch outputs an index with $r$-value in a smaller exponential level than that of the input index of the epoch.
    After the last epoch, we conclude in some right boundary of a phrase $b_j$, as the $r$-value is $0$.
    In our implementation of this approach, the native jumping part corresponds to \cref{lem:naive} and the stable part corresponds to \cref{lem:stable}.
    Each is implemented in $O(\log \frac n z)$ time, and the number epochs is reduced to $O(\log \frac n z)$ due to \cref{lem:no-big-phrases}.
    }
\end{figure}

\paragraph{Implementing $J$, $sJ$, and $M$.}
The implementation of $M$ follows from a well-known approach from random access on Lempel-Ziv compressed strings.
A phrase $P = S[a..b]\in \LZend$ can be processed such that every index $i\in [a..b]$ is mapped to an index $j$ with $S[j]=S[i]$, and either $\min(r_j,\ell_j) \le \min(r_i/2,\ell_i/2)$ (and in either case, $r_j \le r_i$).
However, in order to support an $O(1)$ computation of this mapping, one needs to use $O(\log n)$ space.

The $O(\log n)$ in the space complexity for supporting $M$ arises from a partition of $P$ into $\Theta(\log |P|)$ intervals.
For each interval, the data structure stores $O(1)$ space to support computing $M$ for indices in this interval.
Kempa and Saha choose one random interval from each phrase, thus supporting the computation of $M$ only with probability $1/\log n$ for an index in $P$.
By making this compromise, they reduce the space consumption per phrase to $O(1)$.
It should be easy to see that if we design $N$ to always move between phrases (i.e, $N(i)$ is always in a phrase to the left of the phrase containing $i$), the common property of $M$ is satisfied with high probability (with an appropriate choice of a polylogarithmic function).

$J$ is the simplest function to implement: for every index $i$ contained in a phrase $P=S[a..b]$ with source $S[a'..b']$, we define $J(i) = i-b+b'$.
Notice that this maps $i$ to the index aligned with $i$ in $S[b'- |P|+1..b'] = P$, so we have both $S[i] = S[J(i)]$ and $r_{J(i)}\le r_i$.

In order to define the function $sJ$, Kempa and Saha introduce the $\pre$ function. 
The function $\pre(b,e)$ receives two parameters $b,e$ representing an interval of indices within $S$.
The output of the function corresponds to result of the following process:
check if $b$ and $e$ are contained in the same phrase of $\LZend$.
If they are not then return $b$.
Otherwise, replace $(b,e)$ with $(J(b),J(e))$ and repeat.
Intuitively, one can think of $\pre(b,e)$ as a pointer to an occurrence $S[b'..e']$ of $S[b..e]$ that is achieved by following the source links of $\LZend$ (potentially many times), and contains a phrase boundary.
In order to support efficient computation of $sJ(i)$, the data structure of Kempa and Saha stores $\pre(b,e)$ values for $O(z)$ pairs $(b,e)$.

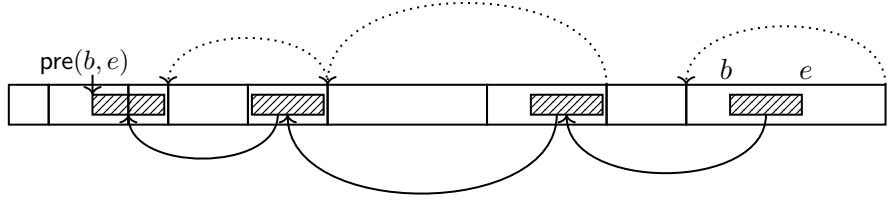
\begin{figure}
\centering
\tikzset{every picture/.style={line width=0.75pt}} 
\usetikzlibrary{patterns}
\tikzset{
    every node/.style={
        text height=1.5ex,
        text depth=0.25ex
    }
}

\begin{tikzpicture}[x=0.75pt,y=0.75pt,yscale=-1,xscale=1]
\pgfmathsetmacro{\height}{20}
\pgfmathsetmacro{\width}{10}
\pgfmathsetmacro{\x}{0} 

\draw (0,0) rectangle (20,\height);
\draw (20,0) rectangle (60,\height);
\draw (60,0) rectangle (80,\height);
\draw (80,0) rectangle (120,\height);
\draw (120,0) rectangle (160,\height);
\draw (160,0) rectangle (240,\height);
\draw (240,0) rectangle (300,\height);
\draw (300,0) rectangle (340,\height);
\draw (340,0) rectangle (440,\height);

\node[anchor=center] at(360,-7) {$b$};
\node[anchor=center] at(400,-7) {$e$};
\draw (362,5)[pattern=north east lines]  rectangle (398, \height-5);

\draw[->,dotted] (440,0 ) to[out=270, in=270] (340,0 );

\draw[->] (380,\height - 5) to[out=90, in=90] (280,\height - 5);

\draw (262,5)[pattern=north east lines]  rectangle (298, \height-5);

\draw[->,dotted] (300,0 ) to[out=270, in=270] (160,0 );
\draw[->] (275,\height - 5) to[out=90, in=90] (140,\height - 5);
\draw (122,5)[pattern=north east lines]  rectangle (158, \height-5);

\draw[->,dotted] (160,0 ) to[out=270, in=270] (80,0 );
\draw[->] (135,\height - 5) to[out=90, in=90] (60,\height -5);
\draw (42,5)[pattern=north east lines]  rectangle (78, \height-5);

\node[anchor=center,font=\small] at(38,-12) {$\pre(b,e)$};
\draw[->](42,-5) to (42, 5);
\end{tikzpicture}
\caption{A demonstration of the function $\pre(b,e)$.
The dotted arrows point from a right-end point of a phrase to the right endpoint of its source phrase.
The bottom arrows demonstrate the process of finding another copy of $S[b..e]$ to the left repeatedly, as long as the leftmost copy found so far is contained in a phrase. }\label{fig:pre}
\end{figure}

\paragraph{The challenges of constructing \cite{KS22}.}
We wish to provide an efficient, deterministic construction based on the ideas of Kempa and Saha.
There are two primary obstacles.

Firstly, the construction of Kempa and Saha requires computing $O(z)$ values of $\pre(b,e)$.
A straightforward computation of $\pre(b,e)$ can potentially take $O(z)$ time, for a total of $O(z^2)$ running time across all required $\pre(b,e)$ values.
We overcome this challenge by providing a batched algorithm for computing multiple queries $(b,e)$ simultaneously,
reminiscent of the winding phase of Farach and Thorup's algorithm for pattern matching in LZ-compressed text~\cite{DBLP:journals/algorithmica/FarachT98}.
This is achieved by processing a batch $I=(b_1,e_1),(b_2,e_2),\ldots$ of queries via a single right-to-left scan of $\LZend$.
When a phrase $P=S[a..b]$ is processed, we can 'shift' all queries contained within $P$ to the source of $P$ as a single block in polylogarithmic time.
Then, we can find every $(b_i,e_i)$ such that $b_i$ and $e_i$ are separated by $a$ in polylogarithmic time per such query.
This leads to a total running time of $O((z+|I|)\polylog(\frac{n}{z}))$.

The second challenge is deterministically implement $M$.
We take a different approach than the one of Kempa and Saha for reducing the total space complexity for $M$.
While they partially computed the function $M$ for every phrase, we pick a subset $\LZend_M\subseteq \LZend$ of size $O(z/\log \frac{n}{z})$, and compute $M$ fully for every phrase in $\LZend_M$.

Intuitively, one would like to pick a set $\LZend_M$ that acts as a cover in the following sense.
Consider for every $i\in [n]$ the sequence $P(i)$ of $\polylog n$ phrases visited by iteratively applying $J$ exactly $\polylog n$ times (for some fixed $\polylog n$) starting from $i$.
A good choice of $\LZend_M$ corresponds to a set of phrases that is both compact, and covers all $P(i)$ sequences.
Since we aim for a preprocessing time of $O(z \polylog(\frac{n}{z}))$, we can not even afford to inspect all $P(i)$ within our preprocessing time.

We sidestep this issue by substituting the demand that $\LZend_M$ covers all such sequences with a relaxed notion of covering that is satisfactory for our purpose.
First, let us simplify our analysis by assuming that prior to $M$, we only make naive steps using $J$ (and never use $sJ$).
Now, for every $i\in [n]$, we say that $\LZend_M$ \textit{light-covers} $i$ if there is some $x\in [0..\log \frac{n}{z}]$ such that $J^x(i)$ is either contained in some phrase of $\LZend_M$ or $J^x(i)$ has $r$-value less than $\frac{2}{3}r_i$.
In words, we do not require $\LZend_M$ to intersect the first $\log\frac{n}{z}$ steps following $i$ in the event when we get 'lucky', and the repeated applications of $J$ happen to incur a significant reduction in $r$-value without invoking $M$. 

With this notion of covering, our task is reduced to covering all sequences $i$,$J(i)$,$J^2(i)$,$\ldots$ ,$J^{\log(\frac{n}{z})}(i)$ such that the starting and ending $r$-values are within the same exponential level.
As it turns out, sequences with this property have a very strict structure.
Namely, we observe that given $i$ in a phrase $P$ with $r_i \in [k..1.5k)$ for some $k\in \K_n$, if $J(i)$ has $r$-value that is also in $[k..1.5k)$, there is only one phrase $P'$ (independent of $i$) that may contain $J(i)$.
In other words: all 'bad steps' from phrase $P$ are mapped to the same phrase $P'$.
This induces a tree structure over the phrases, where the parent of each phrase $P$ is the unique phrase $P'$ that may be its successor in a sequence of bad steps. 
Given such a structure, we can apply the standard approach of selecting a sparse tree layer as a compact path-covering set.

\paragraph{Improving the query time.}
While our implementation is inspired by the construction of Kempa and Saha, our approach is different in the following sense.
Kempa and Saha defined a 'step' function $N$, and prove that it has certain properties.
Then, the query algorithm iteratively applies the function, and it is guaranteed to converge quickly to a boundary index due to the properties of the function.

We use similar function definitions to the ones of Kempa and Saha, but our query algorithm does not iteratively apply these functions.
Instead, our goal is to find the critical value $x$ such that $N^x(i)$ has a sufficiently small $r$-value, and compute $N^x(i)$. 
By further studying the properties of $N$ (and slightly modifying its definition), we are able to design algorithms that find $x$ and $N^x(i)$
without explicitly iterating over $x=0,1,2,\ldots$.

For instance, the component of our algorithm equivalent to the stable part of Kempa and Saha does not use $M$.
We only use $M$ to trigger the stable part, but not to make progress within it.
Instead, we simply consider a sequence of applications of $sJ$ throughout the stable part of the epoch.
Since we gave up the progress obtained by occasionally reaching $M$, the sequence of $sJ$ applications before some progress in $r$ is achieved may be arbitrarily large.
However, by dropping $M$ we obtain a cleaner structure to the stable part.
That structure allows us to find the first $x$ such that $sJ^x(i)$ has a significantly reduced $r$-value in $O(\log \frac{n}{z})$ time.

Finally, since we wish to have $O(\log\frac{n}{z})$ epochs rather than $O(\log n)$, we introduce a method for reducing any LZ-end partition with $z$ phrases into a partition with $O(z)$ phrases such that every phrase has length $O(\frac{n}{z})$ (\cref{lem:no-big-phrases}).
This allows us to assume that the initial $r$-value is at most $\frac{n}{z}$.

Formally, we break the task of previn \cref{thm:main} into the following two lemmas.
\begin{lemma}[Naive Jumping Part]\label{lem:naive}
Given $\LZend$ an LZ-end factorization of $S$, we can construct a data structure supporting the following query.
Given an index $i \in [n]$ with $r_i \in [k..1.5k)$ for some $k\in \K$, compute an index $j$ satisfying $r_j \le r_i$, $S[i] = S[j]$, and either $r_j < k$ or $\ell_j \le 1.5k$.    

The data structure can be constructed in $O(z \log^2 \frac{n}{z})$ time, consumes $O(z)$ space, and supports queries in $O(\log \frac{n}{z})$ time.
\end{lemma}
\cref{lem:naive} corresponds to applying $J$ repeatedly until the $r$-value is either reduced significantly, or an index $i$ for which $M(i)$ is defined is met.
We choose a set $\LZend_M$ for the definition of $M$ such that this is guaranteed to happen within $O(\log^2\frac{n}{z})$ applications of $J$.
The $O(\log \frac{n}{z})$ query time is achieved by introducing shortcuts, which allows us to compute $J^{\log {\frac{n}{z}}}$ in constant time, under some conditions.

\begin{lemma}[Stable Part]\label{lem:stable}
Given $\LZend$ an LZ-end factorization of $S$, we can construct a data structure supporting the following query.
Given an index $i \in [n]$ with $r_i \in [k..1.5k)$ for some $k\in \K$ and $\ell_i \le 1.5k$, compute an index $j$ satisfying $r_j < k $ and $S[i] = S[j]$.

The data structure can be constructed in $O(z \log^2 \frac{n}{z})$ time, consumes $O(z)$ space, and supports queries in $O(\log \frac{n}{z})$ time.
\end{lemma}

\cref{lem:stable} corresponds to receiving the first index of the stable part, and efficiently finding the last index of the epoch.
As we discussed before, it is achieved by analyzing the structure of the stable part to obtain a clean characterization of the first index with $r$-value below the required threshold. 

By combining \cref{lem:naive} and \cref{lem:stable}, we obtain our main theorem.

\begin{proof}[Proof of \cref{thm:main}]
Given $\LZend$ an LZ-End parsing of $S$, we store the array $B$ with $B_i=S[b_i]$ for every $i \in [z]$.
We then construct the data structures of both \cref{lem:naive} and \cref{lem:stable} in $O(z \log^2 \frac{n}{z})$ time.

Given a query index $i$ with $r_i \in [k..1.5k)$ for some $k\in \K$, we apply the data structure of \cref{lem:naive} to obtain $i'$ with $S[i'] =S[i]$ and either $r_{i'}<k$ or $\ell_{i'}\le 1.5k$.
If the latter occurs, we apply the data structure of \cref{lem:stable} with input $i'$ to obtain $j'$ with $r_{j'} < k$ and $S[j']=S[i']=S[i]$.
In both cases, we spend $O(\log \frac{n}{z})$ time and obtained some index $t\in \{ i',j'\}$ with $S[t]=S[i]$ and $r_{t}<k$.
We call the above routine an epoch.
We repeatedly apply epochs, each receiving as input the output of the previous one, until we obtain an index with $r$-value zero.
Since each epoch results in an index $i$ with $r_i$ bounded by a smaller element of $\K$, we will reach an index with $r_{j'}=0$ and $S[j']=S[i]$ within $O(|\K|) = O(\log \frac{n}{z})$ epochs.

In particular, $j' =b_j$ for some phrase $P_j$.
We can find the phrase $P_j$ (and the index $j$) containing $j'$ in $O(\log \frac{n}{z})$ time (see \cref{lem:find-con-phrase}) and then report $S[i]= S[j'] = S[b_j] = B[j]$.

The running time of each epoch is $O(\log \frac{n}{z})$, and there are $O(\log \frac{n}{z})$ epochs, so the total running time of a query is $O(\log^2 \frac{n}{z})$ time.

The space used by our data structure consists of sorting an array of length $z$ storing all $S[b_i]$, and of storing the data structures of \cref{lem:naive} and of \cref{lem:stable}, which require $O(z)$ space each.
The total space complexity is $O(z)$, as required.
\end{proof}

The rest of the paper is dedicated to proving \cref{lem:naive} and \cref{lem:stable}.
We prove \cref{lem:naive} in \cref{sec:naive} and \cref{lem:stable} in \cref{sec:stable}
In \cref{apx:no-big-phrases}, we prove \cref{lem:no-big-phrases} , which allows us to assume without loss of generality that there are no long phrases in our input factorization $\LZend$.
In \cref{sec:kbad-parent-ds,sec:shift-merge,sec:computepre} we separately provide proofs for various auxiliary data structures used throughout the paper. 
Finally, in \cref{sec:extraction} we show how to modify and enhance our techniques to provide a compressed data structure for substring extraction. 

\section{Naive Jumping}\label{sec:naive}
In this section, we prove \cref{lem:naive}. 
We start by presenting a high-level description of the algorithm and of the concepts used by the algorithm, and then proceed to present the full version with all the required definitions and proofs.

\subsection{High-Level Overview}
The objective of this overview is to equip the reader with a conceptual understanding of the machinery behind \cref{lem:naive}. The framework developed in this section provides an intuitive abstraction of our method, meaning that the formal definitions and lemmas stated here will not be explicitly invoked in the full description of the data structure. Finally, we express all running times using $\log n$ instead of $\log \frac n z$ to ensure a smoother presentation.

For every phrase $P=S[a..b]$ with source $P' = S[a'..b']$, denote $\delta_P=b-b'$.
Clearly, for every index $i \in [a..b]$, the index $j=i-\delta_P$ satisfies $S[i]=S[j]$ and $r_j \le r_i$.
We denote by $J$ the function that maps any $i$ in a phrase $P$ to $J(i) = i-\delta_P$.

Recall that our goal is to efficiently support the following query.
Given some index $i$ with $r_i \in [k..1.5k)$, return $j$ with $r_j \le r_i$, $S[j]=S[i]$, and either $r_j < k$ or $\ell_j < 1.5k$. 

Consider the following simple algorithm for achieving this goal:
Given an input index $i$, let $P$ be the phrase containing $i$ in $\LZend$.
Notice that $|P|\neq 1$, as $|P|=1$ would imply $r_i = 0$, which contradicts $r_i \ge k$. 
Therefore, $|P|>1$ and $P$ has a source.
Consider $j=J(i)$.
If $r_j < k$, return $j$.
Otherwise, repeat this process recursively on $j$.

This algorithm clearly terminates, as every recursive call is for an index $j$ strictly smaller than $i$.
However, it may be very inefficient.
Our key insight is that the pathological case where this simple algorithm fails is highly structured. 
This structure can be described as follows.
\begin{lemma}
For every $k\in \K$ there is a forest $T_k$ with node set $\LZend$ such that every index $i\in [1..n]$ contained in phrase $P$ with $r_i \in [k..1.5k)$ and $j=J(i)$, one of the following holds.
\begin{enumerate}
    \item $r_j<k$.
    \item $j$ is in the parent of $P$ in $T_k$.
\end{enumerate}  
The tree $T_k$ can be computed in $O(z \log n)$ time given $\LZend$.
\end{lemma}
In words, a sequence of applications of $J$ where the $r$-value remains large corresponds to an upwards path in $T_k$ (see \cref{fig:tree-comparison}).

To take advantage of the tree structure of the bad steps, we use the concept of a good mapping.
\begin{definition}
A good mapping $M$ for a phrase $P=S[a..b]$ is a data structure that given an index $i\in [a..b]$ returns in constant time an index $j=M(i)$ with $r_j \le r_i$ and $S[j]=S[i]$.
Further, $\min (r_j,\ell_j) \le \min(r_i/2,\ell_i/2)$.    
\end{definition}

Notice that if $r_i <1.5k$, a good mapping $M$ maps $i$ to $j=M(i)$ with $\min(r_j,r_i)<k$.
Therefore, if $i$ is in a phrase for which we have a good mapping, we can immediately return a valid output.

We prove that given a set $M\subseteq \LZend$ of phrases, it is possible to construct good mappings for every phrase in $M$ with the following complexities.
\begin{lemma}\label{lem:batched-good-maps}
    Given $\LZend$ and a set $M\subseteq \LZend$ of phrases, we can compute a good mapping for each phrase in $M$ in $O(z \log^2 n)$ time.
    The total space consumed by the good mappings is $O(z+|M|\log n)$.
\end{lemma}

Notice that we cannot afford to apply \cref{lem:batched-good-maps} with $M=\LZend$, as this would result in a data structure with $O(z\log n )$ space. 
Instead, we apply \cref{lem:batched-good-maps} with a sparse set that covers all paths of length $\log^2 n$ in $T_k$.
Such a set can be obtained by partitioning the vertices of $T_k$ into $\log^2 n$ sets based on their distance from the root modulo $\log^2 n$.
From the pigeonhole principle, we have that one of these sets has size at most $z/ \log^2 n$ - we pick $M_k$ to be this set.
Notice that $M_k$ intersects any upwards path of length $\log^2 n$ in $T_k$.

The total size of $M_k$ sets over all values of $k$ is $O(z/\log^2 n \cdot|\K|) = O(z/\log n)$.
Therefore, by applying \cref{lem:batched-good-maps} to the union of all $M_k$'s, we obtain a good mapping for every phrase within some $M_k$ in total space $O(z+ z/\log n \cdot \log n)= O(z)$.

We are almost ready to present an $O(z)$ space data structure with a slightly slower $O(\log^2 n)$ running time for \cref{lem:naive}.
To this end, we only require the following two technical subroutines.
First, we need a data structure for efficiently finding the phrase $P$ containing an index $i$. 
This can be trivially obtained in $O(\log n)$ time via binary search.
Second, we require a data structure that given a phrase $P$ and $k\in \K$, finds the parent of $P$ in $T_k$ in $O(1)$ time.
Notice that we cannot afford to store all trees $T_k$ within our $O(z)$ space limitation.
We show in \cref{lem:find-kbad-constant} that such a data structure can be constructed in $O(z)$ space and $O(z \log n)$ construction time.

Now, given an index $i$, we first find the phrase $P$ containing $i$ in $O(\log n)$ time.
If $P\in M_k$, we simply return $M(i)$.
Otherwise, we compute $j= J(i)$ and $P'$ the $k$-bad parent of $P$.
If $j$ is not contained in $P'$, we are guaranteed to have $r_j<k$.
We can therefore return $j$.
If $j$ is contained in $P'$, we check if $r_j<k$ and if it is we return $j$.
If $r_j \ge k$, we recurse on $j$.
Because $j$ is guaranteed to be in $P'$, we additionally send $P'$ to the recursive call.
Thus, successive recursive calls do not need to dedicate $O(\log n)$ time to compute the phrase containing $j$.

Every recursive call except for the first one is implemented in $O(1)$ time.
From the path covering property of $M_k$, we have that if more than $\log^2 n$ calls are made, it is guaranteed to reach an index $i$ contained in a phrase of $M_k$.
The running time is therefore $O(\log^2 n)$.

In order to obtain further speedup, we apply a shortcut strategy.
To this end, we analyze a more general notion of 'lack-of-progress' in the simple algorithm.
In our former analysis, we focused on the case where a single application of $J$ leads to an index with a large $r$-value.
Now, we wish to analyze the case where $x$ consecutive applications of $J$ still lead to a large $r$-value.

We present the notion of an $x$-bad index.
Consider a phrase $P$ and the sequence 
$$P=P^{(0)},P^{(1)},P^{(2)},\ldots, P^{(x)}$$ 
of ancestors of $P$ in $T_k$ (with $P^{(x)}$ denoting the $x$-th parent of $P$ in $T_k$).
We say that an index $i$ in $P$ is $x$-bad if for every $y\in[0..x]$, it holds that $J^y(i)$ is in $P^{(y)}$.
It can be shown by induction that the set of $x$-bad indices within $P$ forms a consecutive interval in $P$.
We call this the $x$-bad interval of $P$.

Furthermore, we can use the fact that the sequence of visited phrases is shared by all $x$-bad indices in $P$ to efficiently compute $J^{(x)}(i)$ for an $x$-bad index $i$.
Namely, it holds that for an $x$-bad index $i$, we have $J^x(i)=i-\Delta_P$ where $\Delta_P=\delta_P+\delta_{P^{(1)}} + \ldots + \delta_{P^{(x-1)}}$.

It follows from the above discussion that given the $x$-bad interval of $P$ and $\Delta_P$, we can decide for an index $i$ in $P$ if it is $x$-bad, and if it is, find $J^x(i)$, all in constant time.
Fix $x$ to be $\log n$ from now on.

Since the total size of all $T_k$ is $O(z\log n)$, we cannot afford to store this information for every node in every $T_k$.
Instead, as in the previous algorithm, we chose a sparse set $M'_k$ of size $O(z/\log n)$ for each $T_k$.
We only store this shortcut information for the vertices in $M'_k$.
We also store for each phrase in $M'_k$ its nearest ancestor in $T_k$ that is marked (i.e, in $M_k$).

We are ready to describe our faster query (see \cref{fig:naive-alg}).
First, apply at most $\log n$ steps of $J$.
If at any point we reach an index $j$ with $r_j<k$ or in $M_k$, we halt and return as in the previous algorithm.

Otherwise, after $\log n$ steps the algorithm must reach some phrase in $M'_k$.
At this point, we start using shortcuts instead.
In every shortcut step, we check if the current $i$ is $(\log n)$-bad using the $(\log n)$-bad interval.
Additionally, we check the distance $d$ from the phrase $P$ containing $i$ to its nearest ancestor in $T_k$ that is also in $M_k$.

If $i$ is not $(\log n)$-bad, or if $d < \log n$, apply $J$ at most $\log n$ more times instead of using a shortcut -- we claim that the algorithm terminates with a valid output within those $\log n$ steps.
If none of these steps reach an index $j$ with $r_j<k$, it must hold that $i$ is $(\log n)$-bad.
Therefore, this case must have been trigger due to $d<\log n$.
It follows that one of the $\log n$ phrases visited after $i$ is in $M_k$, and we can use the good mapping associated with it to return a valid output.

Otherwise, we have that $i$ is $x$-bad, and we use $\Delta_P$ to obtain $J^{\log n}(i)$ in constant time.
Notice that since $M'_k$ is selected as a set of all phrases with the same depth modulo $\log n$ in $T_k$, it is guaranteed that $J^{\log n}$ is in a phrase of $M'_k$.

Let us analyze the running time.
We start using shortcuts after applying $J$ at most $\log n$ times.
Then, every shortcut reduces our distance from a phrase in $M_k$ by $\log n$.
This distance is initially at most $\log^2 n$.
Therefore, we will use a shortcut at most $O(\log n)$ times.
Once we stop using shortcut for any reason, we are guaranteed to terminate with a valid output within $O(\log n)$ applications of $J$.
The total running time is therefore $O(\log n)$, as required.

In order to implement the above in $O(\log n)$ time, we also require a mechanism for testing in constant time if a phrase $P$ is in $M'_k$.
This can be obtained by storing a machine word $M_P$ for every phrase in $\LZend$, where the $i$'th bit indicates whether or not $P$ is in $M'_k$ for $k = 1.5^i$.

\subsection{Data Structure for Naive Jumping}
To prove \cref{lem:naive}, we study the process of following the source phrase of the phrase $P$ containing $i$, 'jumping' to the index aligned with $i$ within the occurrence of $P$ ending with the source.
To formalize this, we define a \textit{jumping rule} $J(i)$.
Our ultimate goal is to find an index $j$ with $r_j<k$ and $S[i]=S[j]$.
Intuitively, we make progress towards this goal when we jump to an index with a lower $r$-value, and lose progress if we jump to an index with a higher $r$-value.
This motivates us to define jumping rules that do not lose progress, in the following sense.

\begin{definition}[Legal Jumping Rule]
    For an index $i\in [1..n]$, we say that $j\le i$ is a legal jump for $i$ if $S[i] = S[j]$ and $r_j \le r_i$.
\end{definition}

For a phrase $P_i=S[a_i..b_i]$ with source $P_j=S[a_j..b_j]$, we denote $\delta_{P_i} = b_i-b_j$.
We define the jumping rule $J$ based on $\delta_{P_i}$.

\begin{definition}[Jumping Rule $J$]
    For an index $i$ contained in the phrase $P$, we define $J(i) = i-\delta_P$.
    If $P$ does not have a source, $J(i)$ is undefined.
\end{definition} 

The following lemma follows directly from the definition of LZ-end.
\begin{lemma}\label{lem:j-legal}
    $J$ is a legal jumping rule.
\end{lemma}
\begin{proof}
Let $i$ be an index contained in reference phrase $P= S[a..b]$ with source $P' =S[a'..b']$.
Since $S[a..b] = S(b' -|P|..b']$, we have that $S[i]=S[b - r_i]= S[b'-r_i] =S[J(i)]$.
Also, since $b'$ is a phrase boundary, we have $r_{J(i)} \le   b'-J(i)=  (b-\delta_P)- (i-\delta_P) = b- i=  r_i$, as required.
\end{proof}

Since $J(i)$ is defined based on the phrase containing $i$, our algorithms are often required to find the phrase containing an index $i$.
This can be supported in $O(\log n)$ time and $O(z)$ space by storing $\{a_i \mid i\in [z] \}$ in a balanced search tree.
We would like to support this query slightly faster.
To this end, we provide the following data structure.
\begin{lemma}\label{lem:find-con-phrase}
    $\LZend$ can be processed in $O(z \log \frac{n}{z})$ time to construct a data structure that given $i$ returns $j$ such that $i\in [a_j..b_j]$ in $O(\log \frac{n}{z})$ time.
    The data structure uses $O(z)$ space.
\end{lemma}
\begin{proof}
    We partition the universe $[1..n]$ into $z$ uniform pieces $U_0,U_1,\ldots U_{z-1}$ such that $U_x = (x \frac{n}{z}..(x+1)\frac{n}{z}]$.
    We scan the elements of $A_{\LZend}=\{ a_i \mid i \in [z]\}$ and partition them in $O(z)$ time to obtain $A_x = U_x\cap A_{\LZend}$ for every $x\in [0..z-1]$.
    We construct a balanced search tree over each of $A_x$ in $O(|A_x| \log |A_x|) = O(|A_x| \log \frac{n}{z})$ time and $O(|A_x|)$ space.
    The total construction time is $O(z \log \frac{n}{z})$ and the total space is $O(z)$, as required.
    We also store, for every $x\in [0..z-1]$, the maximal index $a^*_x$ in $A_{\LZend}\cap [1..x \frac{n}{z}]$ and the corresponding phrase $P^*_x$ starting in $a^*_x$.
    These values can be computed straightforwardly in $O(z)$ time.

    Given a query index $i$, we first find in $O(1)$ time the value $x$ such that $i\in U_x$.
    It can be easily verified that $a_j$ such that $P_j$ contains $i$ is either in $U_x$ or is $a^*_x$.
    We use the balanced search tree of $U_x$ to find the predecessor of $i$ in $U_x$ and check which of the two candidates contains $i$.
    The query time is $O(\log \frac{n}{z})$, and then verify which one of the candidates contains $i$ in $O(1)$ time.
\end{proof}

Given an index $i\in [n]$, we can use \cref{lem:find-con-phrase} to find the phrase $P$ containing $i$ in $O(\log \frac{n}{z})$ time. 
Since every phrase $P_i$ is stored alongside the index $j$ indicating its source phrase $P_j$, we can compute $\delta_{P_i}$ in $O(1)$ time given $P_i$, by accessing the array storing the phrases alongside their reference data. 

We proceed to define the function $\pre$, which is used to define the second jumping rule used in the naive jumping part.
\begin{definition}[The function $\pre$, see \cref{fig:pre}]
    For a string $S$ and two integers $b\le e \in [1..n]$ we define $\pre(b,e)$ as follows.
    If $b$ and $e$ are contained in the same phrase $P\in \LZend$, and $P$ has a source, then $\pre(b,e) = \pre(J(b),J(e))$.
    Otherwise, $\pre(b,e) = b$.

    Additionally, we define $\delpre(b,e) = b-\pre(b,e)$
\end{definition}
In words, $\pre(b,e)$ is obtained by finding previous occurrences of $S[b..e]$ by following the source links of $\LZend$, until it is no longer possible.
Halting this process can occur either because $b$ and $e$ are separated by a phrase boundary, or because the phrase $P$ containing $[b..e]$ has no source (which can only occur if $P=S[b]=S[e]$).
The left endpoint of the occurrence on which we stop in this process is $\pre(b,e)$.
Since $\pre(b,e)$ is obtained by applying $J$ repeatedly, the following is a consequence of \cref{lem:j-legal}.

\begin{observation}\label{obs:pre-is-legal}
    Let $[b..e] \subseteq[1..n]$ be an interval.
    For every $i\in [b..e]$, it holds that $i-\delpre(b,e)$ is a legal jump for $i$.
\end{observation}

We will define a jumping rule $M$ that relies on $\delpre$.
Unlike $\delta_P$, that is essentially given alongside $P$ in $\LZend$, computing $\delpre(b,e)$ from $\LZend$ require some non-trivial work.
A naive approach would be to check if $[b..e]$ contains a phrase boundary, and if it does not - proceed according to the recursive definition.
This may result in $O(z)$ time for computing a single $\pre(b,e)$ value.
In \cref{sec:computepre}, we present the following algorithm, which will be instrumental for efficiently computing $\delpre$-dependent jumping rules.
\begin{lemma}\label{lem:computepre}
    Given $\LZend$ an LZ-End factorization of a string $S$ and a set of pairs $I \subseteq [1..n]^2$ such that $e\in [b..b+\frac{n}{z}]$ for every $(b,e) \in I$, we can compute $\pre(b,e)$ for every $[b..e]\in I$ in $O(( z+|I|) \log^2(n/z))$ time.
    Here, $z = |\LZend|$ and $n = |S|$.
\end{lemma}

In order to define $M$, we need to introduce the halved canonical partition of an interval.
This partition appears under a variety of names in compressed random access data structures (in particular in \cite{KS22}).
\begin{definition}[Canonical Partition, Halved Canonical Partition]
    For an interval $[a..b]$ of integers, we define the canonical partition of $[a..b]$ as the following intervals.
    \begin{enumerate}
        \item For every $i\in [0..\log(\frac{b-a}{2})]$, the interval $[a+2^i-1..\min(a+2^{i+1}-2,\frac{b-a}{2})]$
        \item For every $i\in [0..\log(\frac{b-a}{2})]$, the interval $[\max(b-2^{i+1}+2,\frac{b-a}{2})..b-2^i+1]$
    \end{enumerate}

We define the Halved Canonical Partition based on the Canonical Partition.
The halved canonical partition of $[a..b]$ is obtained by taking the canonical partition and splitting every interval larger than 1 into two equal length interval.
Formally, for every interval $I=[x..y]$ in the canonical partition of $[a..b]$ such that $|I| >1$, the halved canonical partition of $[a..b]$ contains $[x..x+ |I|/2)$ and $[x+|I|/2..y]$.
\end{definition}

\begin{figure}[h]
    \centering
    \tikzset{every picture/.style={line width=0.75pt}} 
\tikzset{
    every node/.style={
        text height=1.5ex,
        text depth=0.25ex
    }
}
\begin{tikzpicture}[x=0.75pt,y=0.75pt,yscale=-1,xscale=1]
\pgfmathsetmacro{\height}{20}
\pgfmathsetmacro{\width}{3}
\pgfmathsetmacro{\x}{0} 

 \draw[->] (\x,-10) -- (\x,0);
\node at (\x,-18) {$a$};
  \foreach \i in {1,...,6} {
    \draw (\x,0) rectangle (\x+\width,\height);

    \pgfmathparse{\x+\width}
    \global\let\x=\pgfmathresult 
    \pgfmathparse{2 * \width}
    \global\let\width=\pgfmathresult 
  }
    \draw(\x,0) rectangle (\x+35,\height);
    \pgfmathparse{\x+35}
    \global\let\x=\pgfmathresult 
    \draw(\x,0) rectangle (\x+35,\height);

    \draw[dashed](\x,-12) -- (\x,\height+12);

    \pgfmathparse{ \width/2}
    \global\let\width=\pgfmathresult 
     \pgfmathparse{\x+35}
    \global\let\x=\pgfmathresult 

  \foreach \i in {1,...,6} {
    \draw (\x,0) rectangle (\x+\width,\height);

    \pgfmathparse{\x+\width}
    \global\let\x=\pgfmathresult 
    \pgfmathparse{\width/2}
    \global\let\width=\pgfmathresult 
  }
   \draw[->] (\x,-10) -- (\x,0);
    \node at (\x,-18) {$b$};
\end{tikzpicture}
    \caption{A demonstration of a canonical partition.}
\end{figure}
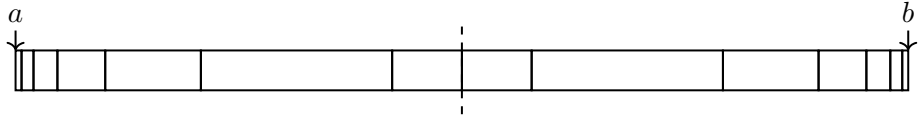

The key property of the canonical partition is that for an index $i\in [a..b]$, the interval $I$ that contains $i$ in the canonical partition has length at most $m = \min(r_i,\ell_i)+1$.
It immediately follows that in the halved canonical partition, the length of the interval containing $i$ is bounded by  $\max(m/2,1)$.
We are now ready to define the jumping rule $M$.

\begin{definition}[Jumping Rule $M$]
    Let $i$ be an index contained in a marked phrase $P=S[a..b]$.
    Let $[x..y]$ be the interval in the halved canonical partition of $[a..b]$ that contains $i$.
    We define $M(i) = i - \delpre(x,y)$.
\end{definition}
It follows from \cref{obs:pre-is-legal} that $M$ is a legal jumping rules.
\begin{corollary}\label{cor:M-sj-legal}
    $M$ is a legal jumping rule.
\end{corollary}

Notice that $M(i) = i - \delpre(b,e)$ transforms $i\in [b..e]$ into an interval of size $|[b..e]|$ that contains a phrase boundary.
 Therefore $M(i)$ has a phrase boundary that is quite close to it either from the left or from the right, inducing a significant decrease either of $r$-value or of $\ell$-value compared to $i$.  
Formally, we prove the following useful property of $M$.
\begin{claim}\label{clm:m-reduce-l-or-r}
    For an index $i$ that is contained in a phrase $P\in M$, denote $j=M(i)$.
    It holds that $r_j \le m$ or $\ell_j \le m$ where $m=\min(r_i/2,\ell_i/2)$.
\end{claim}
\begin{proof}
    Let $[x..y]$ be the interval in the halved canonical partition of $P$ containing $i$.
    Notice that $|[x..y]| \le \max(m,1)$.
    If $|[x..y]| = 1$ it must be the case that $S[j]$ is a phrase of length 1.
    Therefore $r_j=\ell_j = 0 \le m$, as required.
    
    Otherwise, we have $|[x..y]| \le m$.
    Let $x' = \pre(x,y)$ and $y' = x' + |[x..y]|$.
    The index $j= M(i)$ is in $[x'..y']$, and from the definition of $\pre$, there is some phrase boundary $b'\in [x'..y']$.
    If $j\le b'$, then $r_j\le b'-j \le |[x'..y']| = |[x..y]|\le m$.
    Otherwise, $j > b'$ and we have $\ell_j \le j-b' \le |[x'..y']| \le m$.
\end{proof}

\cref{clm:m-reduce-l-or-r} implies that if our sequence of applications of $J$ reaches upon some index $j$ in a marked phrase, then $M(j)$ is a valid output for \cref{lem:naive}.
Our first goal will be to define a 'good' set $M$, such that every sufficiently long sequence of $J$ applications eventually reaches $M$ (Roughly, see \cref{lem:compute-marked} for the precise definition).
In order to exploit this property, our data structure has to support efficient computation of $M(i)$ for an index $i$ in a marked phrase.
We present the following data structure for efficient $M$ computation.
\begin{lemma}\label{lem:ds-mi}
Given $\LZend$, we can construct a data structure of size $O(|M| \log \frac{n}{z})$ that given $i\in [n]$, and the phrase $P$ containing $i$, computes $M(i)$ in $O(1)$ time.
The data structure can be constructed in $O((M\log \frac{n}{z} + z) \log^2 \frac{n}{z})$ time.
\end{lemma}
\begin{proof}
We apply \cref{lem:computepre} to compute for every $P=S[a..b]\in M$ and every $I=[x..y]$ in the halved canonical partition of $P$ the value $P_I = \pre(x,y)$ in $O((|M| \log \frac{n}{z} + z) \log^2 n)$ total time (Due to \cref{lem:no-big-phrases}, the halved canonical partition of every $P\in \LZend$ has $O(\log \frac{n}{z})$ intervals).
Notice every query interval is contained in some phrase $P_i$, and $|P_i| \le \frac{n}{z}$ according to \cref{lem:no-big-phrases}. 
Therefore, every $(b,e)$ query has $e\in [b..b+\frac{n}{z}]$, and this is a legal input for \cref{lem:computepre}.

We then compute for each interval $I=[x..y]$ in the halved canonical partition of $P$ the value $\delpre(I)=x-\delpre(x,y)$, and store those values in an array $A_P$ specifically constructed for $P$.
We store in each phrase $P$ a link to its corresponding array $A_P$.

Given an index $i$ and its containing phrase $P$, we can retrieve the interval $I$ containing $i$ in the halved canonical partition of $P$ in $O(1)$ time using $O(1)$ bitwise operations on $a,b,i$ (We need the most significant bit set to $1$ in $i-a$ if $i$ is in the first half of $P$, or the most significant bit set to $1$ in $b-i$ otherwise).
We can then use $A_P$ to find $\delpre(I)$ and return $M(i)= i - \pre(I)$.
If $P$ is not linked to an array $A_P$, we reach upon the conclusion that $P\notin M$, and therefore $M(i)$ is undefined.

The construction time is dominated by $O((|M| \log\frac{n}{z}  +z )\log^2 n)$, the query time is $O(1)$ for applying a predecessor query, and the size is $O(M \log \frac{n}{z})$ for storing all $\delpre$ values.
\end{proof}

\subsection{Finding a Good Set of Marked Phrases}
Our first goal is to find a good set $M$ of marked phrases.
Before formally defining a good set (see \cref{lem:compute-marked}), let us provide some intuition.
We wish to find a set $M$ that is both \textit{sparse} and \textit{covering}.
We require $M$ to be small, since we will support a query for $M(i)$ using \cref{lem:ds-mi}, which consumes $O(z+ |M|\log \frac{n}{z})$ space.
Since we wish to have $O(z)$ space complexity, we need to have $|M| \in O(z/\log \frac{n}{z})$.

As for the covering property, we would like to guarantee that every long enough sequence of steps $i,J(i),J^2(i),\ldots$ eventually reaches some index contained in a phrase of $M$.
Instead, we will obtain $M$ that satisfies a weaker, yet sufficient, notion of covering.
Namely, we will guarantee that every long enough sequence of steps either reaches a phrase of $M$, or reaches an index with small enough $r$-value.

Formally, we prove the following.
\begin{lemma}\label{lem:compute-marked}
    There is an algorithm that given $\LZend$, outputs a set $M$ of size $O(z/\log \frac{n}{z})$ such that for every index $i$ with $r_i\in [k..1.5k)$ for some $k\in \K$, there is $x\in [1..{\log}^2 \frac{n}{z}]$ such that either $j = J^x(i)$ is contained in a phrase of $M$, or $r_{j}< k$.

The algorithm works in $O(z \log \frac{n}{z})$ time.
\end{lemma}

We introduce the concept of a bad jump, that corresponds to a step $J(i)$ that remains in the same exponential level.
\begin{definition}[Bad Jump, Good Jump]
Let $i$ be an index and let $k\in \K$ be the unique integer power of $1.5$ such that $r_i \in [k..1.5k)$. 
We say that $i$ is a $k$-bad jump if $j=J(i)$ has $r_j \ge k$.
If $i$ is not a $k$-bad jump, we say $i$ is a $k$-good jump.   
\end{definition}
When $k$ is clear from context or irrelevant, we omit it from the notion of $k$-bad jumps and $k$-good jumps, using bad jump and good jumps instead.

The existence of bad jumps is the main obstacle in finding a value $j$ such that $r_j<k$ via iterative application of $J$.
The pathological case would be where the sequence $i,J(i),J^2(i),\ldots, J^x(i)$ consists only of bad jumps for high values of $x$.
However, we show that bad jumps are very structured.
Namely, we show that the phrase $P'$ that contains $J(i)$ for some $k$-bad jump $i$ can be decided only based on the phrase $P$ containing $i$ and $k$.
In other words, all $k$-bad jumps within a phrase $P$ share some uniform behavior. 
We exploit this structure to efficiently cover every long sequence of bad jumps. 

We make the following claim regarding bad jumps.
\begin{lemma}\label{lem:bad-jumps-same-phrase}
    Let $k\in \K$.
    Let $P$ be a phrase and let $i$ and $i'$ be two $k$-bad jumps in $P$ with $r_i,r_i' \in [k..1.5k)$.
    Then, $J(i)$ and $J(i')$ are in the same phrase.
\end{lemma}
\begin{proof}
    For an illustration, see \cref{fig:badparent}.
    Assume $i<i'$.
    Denote $j = J(i)$ and $j' = J(i')$.
    Since both $r_i$ and $r_{i'}$ are in $[k..1.5k)$, and $i$ and $i'$ are in the same phrase, we have $i'-i< 0.5k$.
    Since $j = i- \delta_P$ and $j' = i'-\delta_P$, we have $ j'-j = i'-i < 0.5k$.
    Now, assume to the contrary that $j$ and $j'$ are in different phrases.
    This would mean that there is a phrase boundary between $j$ and $j'$, which  implies $r_{j} \le j'-j < 0.5k < k$, contradicting the assumption that $i$ is a bad jump.
\end{proof}

\begin{figure}[h]
    \centering
    \tikzset{every picture/.style={line width=0.75pt}} 
\usetikzlibrary{patterns}

\begin{tikzpicture}[x=1pt,y=1pt,yscale=-1,xscale=1]
\pgfmathsetmacro{\blockheight}{25}
\pgfmathsetmacro{\stringheight}{20}
\pgfmathsetmacro{\height}{20}

\pgfmathsetmacro{\width}{10}
\pgfmathsetmacro{\x}{0} 

\draw (20,0) rectangle (180,\blockheight);

\draw[densely dotted](20,50) rectangle (180, \blockheight+50);

\draw(35,0) rectangle (45,\blockheight);
\node[anchor=north,font=\small] at(40,5) {$i$};

\draw(85,0) rectangle (95,\blockheight);
\node[anchor=north,font=\small] at(90,5) {$i'$};

\draw(35,50) rectangle (45,50+\blockheight);
\node[anchor=north,font=\small] at(40,55) {$j$};

\draw(85,50) rectangle (95,50+\blockheight);
\node[anchor=north ,font=\small] at(90,55) {$j'$};

\node[anchor=north,font=\tiny] at(40,-10) {$r_i$};
\draw[->] (45,-5) -- (180,-5);

\node[anchor=north,font=\tiny] at(90,-15) {$r_{i'}$};
\draw[->] (95,-10) -- (180,-10);

\draw[dashed, line width = 0.7] (60,35) -- (60,65+\blockheight);

\draw[dashed, line width = 0.7] (180 ,-15) -- (180 ,0);

\draw[<->] (45,15) -- (85,15);
\node[anchor=north,font=\tiny] at(65,4) {$< 0.5k$};

\node[anchor=north,font=\tiny] at(40,40) {$r_j$};
\draw[->] (45,45) -- (60,45);

\node at (10,12) {$P$};

\node at (-5,62) {$\textsf{source}(P)$};

\end{tikzpicture}
    \caption{The proof of \cref{lem:bad-jumps-same-phrase}. 
    The top rectangle represents $P$, while the bottom rectangle represents the source of $P$ (In the text, the source appears to the left of $P$. In this figure it is displayed below $P$).
    The dashed line between $j$ and $j'$ represents a phrase boundary separating them.
    It can be seen that the existence of such phrase boundary leads to $r_{j}<k$.
    }\label{fig:badparent}
\end{figure}
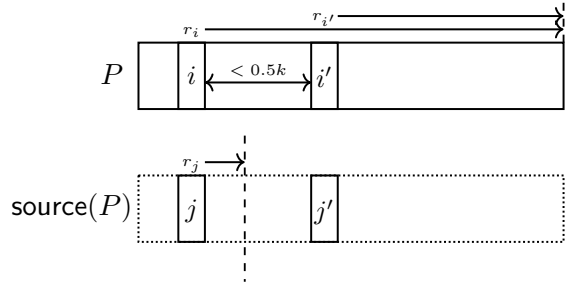

Following \cref{lem:bad-jumps-same-phrase}, we introduce the terminology of a $k$-bad parent.
For a phrase $P$ and $k\in \K$, we call 'the $k$-bad parent' of $P$ the unique phrase $P'$ such that every bad jump $i$ in $P$ with $r_i \in [k..1.5k)$ has $J(i) \in P'$ (uniqueness follows from \cref{lem:bad-jumps-same-phrase}).
If there is no such bad jump in $P$, then $P$ does not have a $k$-bad parent.
Notice that $P'\neq P$, since $J(i)$ is necessarily in a phrase strictly to the left of the phrase containing $i$.
As implied by the parent terminology, the $k$-bad parent relation induces a forest structure over the phrases of $\LZend$.
We denote as $T_k$ the forest with vertices being the phrases of $\LZend$, and the parent of every $P$ is the $k$-bad parent of $P$.

In \cref{sec:kbad-parent-ds}, we prove the following.
\begin{lemma}\label{lem:kbad-parent-ds}
    Given $\LZend$, there is a data structure supporting the following query: Given $P\in \LZend$ and $k\in \K$, return the $k$-bad parent of $P$, or report that $P$ does not have a $k$-bad parent.
    The construction time of the data structure is $O(z \log^2 \frac{n}{z})$, the space is $O(z)$ and the query time is $O(1)$.
\end{lemma}

We are now ready to prove \cref{lem:compute-marked}.

\begin{proof}[Proof of \cref{lem:compute-marked}]
We start by defining the set $M$, and then show how to efficiently compute $M$.

\paragraph{Existence of proper $M$.}
For every $k'\in \K$, we define a set $M_{k'}$ of phrases according to the structure of $T_{k'}$. 
If the depth of $T_{k'}$ is less than $\log^2\frac{n}{z}$, we set $M_{k'}=\emptyset$.
Otherwise, let $c\in [0..\log^2\frac{n}{z}-1]$ be an integer such that the number of phrases in $T_{k'}$ with depth that has remainder $c$ modulo $\log^2\frac{n}{z}$ is at most $z/\log^2 \frac{n}{z}$.
This value of $c$ exists due to the pigeonhole principle.
We set $M_{k'}$ to be the phrases with depth $c$ modulo $\log^2 \frac{n}{z}$ in $T_{k'}$.
We claim that $M = \cup_{k'\in \K}M_{k'}$ satisfies the condition of the lemma.

Clearly, we have $|M|\in O(z / \log \frac{n}{z})$.
Consider some $i\in [1..n]$ with $r_i \in [k..1.5k)$.
Consider $j=J^x(i)$ for some integer $x$.
If $j$ is a bad jump in phrase $P$ with $r_i \in [k..1.5k)$, than $J(j)$ is in the parent of $P$ in $T_k$.
It follows that a sequence of phrases containing the bad jumps corresponds to an upwards path in $T_k$, and therefore if all of the first $\log^2\frac{n}{z}$ steps following $i$ are bad jumps, one of these must be in a phrase $P\in M_k$.

Assume that $J^x(i)$ is not in a marked phrase for every $x\in [0..\log^2\frac{n}{z}-1]$.
It follows from the above that there is an index $x'\in [0..\log^2\frac{n}{z}-1]$ such that $i'= J^{x'}(i)$ is a good jump.
Let $j'= J(i')$.
From the definition of a good jump, we have that $r_{j'} < k$.

We have shown that if there is no $x\in [0..\log^2\frac{n}{z}-1]$ such that $J^x(i)$ is not in a phrase of $M$, then there is $x\in [0.. \log^2\frac{n}{z}]$ such that the $r$-value of $J^x(i)$ is below $k$, which concludes the proof.

\paragraph{Efficient construction.}
For every $k'\in \K$, the algorithm construct $T_{k'}$.
This is done by first constructing the data structure of \cref{lem:find-kbad-constant} in $O(z \log^2 \frac{n}{z}$) time, and then querying every $P\in \LZend$ for its $k'$ parent in $O(1)$ time.
Once $T_{k'}$ is constructed, we can straightforwardly count the number of vertices in each depth of $T_{k'}$ and select the correct value of $c$ to use for $M_{k'}$ (or decide that $M_{k'} = \emptyset$).

Since we construct each $T_{k'}$ in $O(z)$ time, and then apply $O(|T_{k'}|) = O(z)$ additional processing to find $M_{k'}$, the total running time for constructing and processing all $T_{k'}$ is $O(z |\K|) = O(z \log \frac{n}{z})$.
This is dominated by the $O(z \log^2 \frac{n}{z})$ time for constructing the data structure of \cref{lem:find-kbad-constant}.
\end{proof}

\begin{figure}[htbp]
    \centering
   \begin{subfigure}[b]{0.58\textwidth}
      \centering
        \tikzset{every picture/.style={line width=0.75pt}} 
\usetikzlibrary{patterns}

\begin{tikzpicture}[x=1pt,y=1pt,yscale=-1,xscale=1]
\pgfmathsetmacro{\blockheight}{25}
\pgfmathsetmacro{\stringheight}{20}
\pgfmathsetmacro{\height}{20}

\pgfmathsetmacro{\width}{10}
\pgfmathsetmacro{\x}{0} 

\draw (20,0) rectangle (180,\blockheight);

\draw(25,50) rectangle (170, \blockheight+50);

\draw(35,100) rectangle (160, \blockheight+100);

\draw(45,0) rectangle (55,\blockheight);
\node[anchor=north,font=\small] at(50,5) {$i_1$};

\draw(45,50) rectangle (55,50+\blockheight);
\node[anchor=north,font=\small] at(50,55) {$i_2$};

\draw(45,100) rectangle (55,100+\blockheight);
\node[anchor=north,font=\small] at(50,105) {$i_3$};

\node at (10,12) {$P_5$};

\node at (10,62) {$P_4$};

\node at (10,112) {$P_3$};

\draw[dashed] (140,-5)--(140,105+\blockheight)  node[at start, above] {$k$};

\draw[dashed] (185,-5)--(185,105+\blockheight) node[at start, above] {$1.5k$};

\draw[->,red,thick] (50,\blockheight)--(50,50) node[midway, right,black] {$J(i_1)$};

\draw[->,red,thick] (50,50+\blockheight)--(50,100) node[midway, right,black] {$J(i_2)$};

\end{tikzpicture}
        \caption{A sequence of $J$ steps.
        In each step, the containing rectangle represents the phrase containing the reached index.}
    \end{subfigure}
    \hfill
    \begin{subfigure}[b]{0.35\textwidth}
          \centering
        \begin{forest}
  for tree={
    draw, 
    circle, 
    minimum size=8mm
  }
  [$P_1$, name=root 
    [$P_2$, name=leftNode]
    [$P_3$, name=LeafC
      [$P_4$, name=leafD, edge={red, <-}
      [$P_5$, name= LeafE, edge={red, <-}]
      ] 
    ]
  ]
  \node[draw=red, dashed, inner sep=5pt, fit=(LeafC) (LeafE), label=below:Bad path] {};
\end{forest} 
        \caption{The tree $T_k$}
        
    \end{subfigure}
    \caption{A demonstration of the relationship between the tree $T_k$ and a sequence of $J$ steps where the $r$-value remains in $[k..1.5k)$.
    The phrases visited throughout the sequence correspond to an upwards path in $T_k$.}
    \label{fig:tree-comparison}
\end{figure}
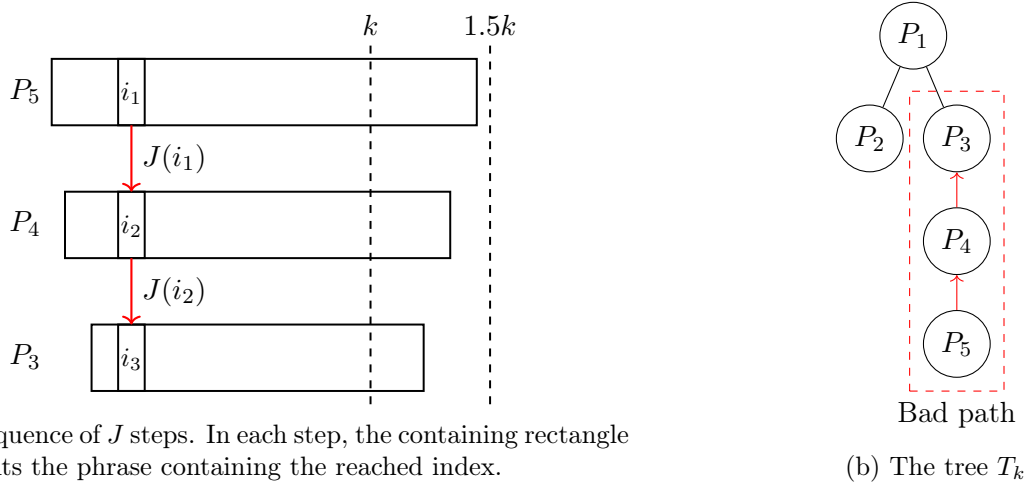

\paragraph{Reaching a marked phrase fast.}
So far, we have guaranteed that if we repeatedly apply $J$ starting from some index $i$, we will either reach a marked phrase or reduce the $r$-value to the next exponential level within $O(\log^2 \frac{n}{z})$ jumps.

Now, we introduce 'shortcuts', allowing us to find $J^{\log \frac{n}{z}}(i)$ of an index $i$ in constant time rather than by computing $J$ iteratively step-by-step.
The shortcut mechanism can only be applied if the next $\log \frac{n}{z}$ steps are all 'bad'.
This will allow us to retrieve a valid output for the naive jumping part in $O(\log \frac{n}{z})$ time, as if one of the next $\log \frac{n}{z}$ steps are not bad - we can afford to make these steps one-by-one as we reach a good enough index after reaching a good step.

For a phrase $P$, $k\in \K$ and a non-negative integer $x$, we denote as $\Parent{P}{x}{k}$ the $x$'th ancestor of $P$ in $T_k$.  
In particular, $\Parent{P}{0}{k}=P$ for every phrase $P$ and $k\in \K$.
We also call $\Parent{P}{x}{k}$ the $(x,k)$-bad parent of $P$. 

We introduce the notion of $(x,k)$-bad jumps.
\begin{definition}[$(x,k)$-bad jumps]
For a phrase $P$, $k\in \K$ and non-negative integer $x$, we say that $i$ is an $(x,k)$-bad jump of $P$ if for every $y\in [0..x]$ it holds that $J^y(i)$ is contained in $\Parent{P}{x}{k}$.
\end{definition}
Notice that every index $i$ contained in a phrase $P$ is $(0,k)$-bad.
For a node $P$, $k\in \K$ and non-negative integer $x$ we denote $\delta_x(P)=\sum_{y=0}^{x} \delta_{\Parent{P}{y}{k}}$.
In words, $\delta_x(P)$ is the sum of $\delta$ values taken over the first $x+1$ ancestors of $P$ in $T_k$ (starting from $P$ itself).

We make the following simple observation.

\begin{observation}\label{obs:same-jumps-same-shift}
Let $i\in [n]$ and $a$ be a non-negative integer.
Let $P_{i_0},P_{i_1},\ldots,P_{i_a}$ be the phrases containing $i,J(i),\ldots,J^a(i)$, respectively.
We have $J^a(i)= i -\sum_{x=0}^{a-1}\delta_{P_{i_x}} $.
\end{observation}

The following follows directly from \cref{obs:same-jumps-same-shift} and from the definition of an $(x,k)$-bad jump.

\begin{corollary}\label{cor:fast-jump-formula}
    Let $P$ be a phrase, $k\in \K$, and non-negative integer $x$.
    If $i$ is an $(x,k)$-bad index of $P$, then $J^{x+1}(i) = i -\delta_x(P)$.
\end{corollary}

We show that the $(x,k)$-bad jumps of $P$ form a consecutive interval within $P$.
\begin{lemma}\label{lem:x-plus-1-bad-interval-from-x}
Let $P=S[a..b]$ be a phrase and let $k\in \K$ and $x$ be a non-negative integer.
The set of $(x,k)$-bad jumps of $P$ is an interval $[a'..b']$.
Furthermore, let $P'=[c..d]$ be the $(x+1,k)$-bad parent of $P$.
The $(x+1,k)$-bad jumps of $P$ are exactly $[a'..b']\cap [c+\delta_{x}(P) .. d + \delta_x(P)]$.
\end{lemma}
\begin{proof}
For an illustration, see \cref{fig:xbadinterval}
For $x=0$, the first statement is trivial since the $(0,k)$-bad jumps of $P=S[a..b]$ are simply $[a..b]$.

We inductively prove that the second statement is true for every $x \ge 0$.
This, in turn, leads to the first statement also being true, as it shows that the $(x,k)$-bad jumps are obtained by an intersection of intervals.

Assume that for some non-negative $x$, the $(x,k)$-bad jumps of $P$ are indeed an interval $[a'..b']$. 
Let us prove both directions of the equality between $[a'..b']\cap [c+\delta_x..d+\delta_x]$ and the set of $(x+1,k)$-bad jumps of $P$.

First, let $i \in [a'..b'] \cap [c+\delta_x..d+\delta_x]$.
Since $i\in [a'..b']$, it is an $(x,k)$-bad jump of $P$.
By \cref{cor:fast-jump-formula}, we have $J^{x+1}(i) = i - \delta_x(P)$.
Since we have $J^{x+1}(i) = i - \delta_x(P) \in [c..d]$, it holds that $J^{x+1}(i)$ is in the $(x,k)$-bad parent of $P$, making $i$ an $(x+1,k)$ bad index of $P$.

For the other direction, consider an $(x+1,k)$-bad index $i$ of $P$.
Since $i$ is, in particular, an $(x,k)$-bad index of $P$, we have $i \in [a'..b']$.
Due to the same reasoning as before, since $i$ is $(x,k)$-bad we have $J^{x+1}(i) = i - \delta_x(P)$. 
Since $i$ is $(x+1)$ bad, we know that $J^x(i) $ is in $[c..d]$, which means that $i - \delta_x(P) = J^x(i) \in [c..d]$.
This is equivalent to $i \in [c+\delta_x(P)..d+\delta_x(P)]$, which concludes the proof. 
\end{proof}

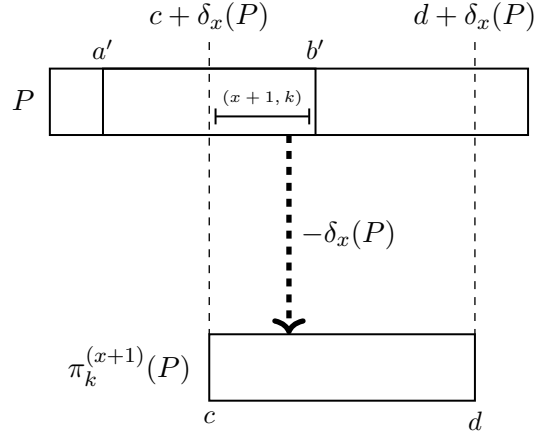
\begin{figure}[h]
    \centering
    \tikzset{every picture/.style={line width=0.75pt}} 
\usetikzlibrary{patterns}

\begin{tikzpicture}[x=1pt,y=1pt,yscale=-1,xscale=1]
\pgfmathsetmacro{\blockheight}{25}
\pgfmathsetmacro{\stringheight}{27}
\pgfmathsetmacro{\height}{20}

\pgfmathsetmacro{\width}{10}
\pgfmathsetmacro{\x}{0} 

\draw (0,0) rectangle (180,\blockheight);

\draw(60,100) rectangle (160, \blockheight+100);

\node[anchor=north,font=\small] at(60,\blockheight+100) {$c$};

\node[anchor=north,font=\small] at(160,\blockheight+100) {$d$};

\draw(20,0) rectangle (100,\blockheight);
\node[anchor=north,font=\small] at(20,-15) {$a'$};

\node[anchor=north,font=\small] at(100,-15) {$b'$};

\draw[->,dashed,line width=2] (90,\blockheight) -- (90,100) node[midway,right]{$-\delta_{x}(P)$};

\draw[dashed, line width = 0.5] (60,-10) -- (60,100) node[at start,above]{$c+\delta_{x}(P)$};

\draw[dashed, line width = 0.5] (160,-10) -- (160,100)node[at start,above]{$d+\delta_{x}(P)$};

\node at (-10,12) {$P$};

\node at (30,112) {$\pi_k^{(x+1)}(P)$};

\draw[{|}-{|}] (62,18)--(98,18) node[midway,above]{\tiny {$(x+1,k)$}};

\end{tikzpicture}
    \caption{A demonstration of \cref{lem:x-plus-1-bad-interval-from-x}.
    An $(x+1,k)$-bad index must be in the $(x,k)$-bad interval.
    Furthermore, the process of jumping $x+1$ times must map it to the $(x+1,k)$-bad parent of $P$.
    This makes $x+1$ bad jumps correspond to decreasing the value of the initial index by $\delta_{x}(P)$.
    Therefore, we must have that an $(x+1,k)$-bad index is in $[c+\delta_x(P)..d+\delta_x(P)]$, where $S[c..d]$ is the $(x+1,k)$-bad parent of $P$.
    The interval denotes as $(x+1,k)$ is the $(x+1,k)$-bad interval, obtained by intersecting $[a'..b']$ and $[c+\delta_x(P)..d+\delta_x(P)]$.
    }\label{fig:xbadinterval}
\end{figure}

We call the interval containing exactly the $(x,k)$-bad jumps of $P$ (implied by \cref{lem:x-plus-1-bad-interval-from-x}) the $(x,k)$-bad interval of $P$.
We proceed to show that the $(x,k)$-bad interval can be computed efficiently.

\begin{lemma}\label{lem:xk-bad-exist-compute}
    Let $P$ be a phrase, let $k\in \K$ and let $x$ be a non-negative integer.
    The $(x,k)$-bad interval of $P$ can be computed in $O(x)$ time, given access to $T_k$.
\end{lemma}
\begin{proof}
 We use \cref{lem:x-plus-1-bad-interval-from-x} to compute the $(x,k)$-interval of a phrase $P=S[a..b]$ iteratively.

Firstly, the $(0,k)$-bad interval is trivially $[a..b]$.
Now, we show how to find the $(y,k)$-bad interval for every $y\in [1..x]$, given the $(y-1)$-bad interval $[a'..b']$.
According to \cref{lem:x-plus-1-bad-interval-from-x}, the $(y,k)$-bad interval of $P$ is $[a'..b']\cap [c'+\delta_y(P)..d'+ \delta_y(P)]$ where $\Parent{P}{y+1}{k}= S[c..d]$ is the $(y+1)$'th ancestor of $P$ in $T_k$.
We accumulate the $\delta$ values of the ancestors of $P$ as we ascend up the tree $T_k$, so we have access to $\delta_y(P)$ when computing the $(y,k)$-bad interval of $P$.
It follows that we can compute the $(y,k)$-bad interval from the $(y-1,k)$-bad interval in $O(1)$ time.
In total, the time to obtain the $(x,k)$-bad interval is $O(x)$, as required.   
\end{proof}

We are now ready to prove \cref{lem:naive}.
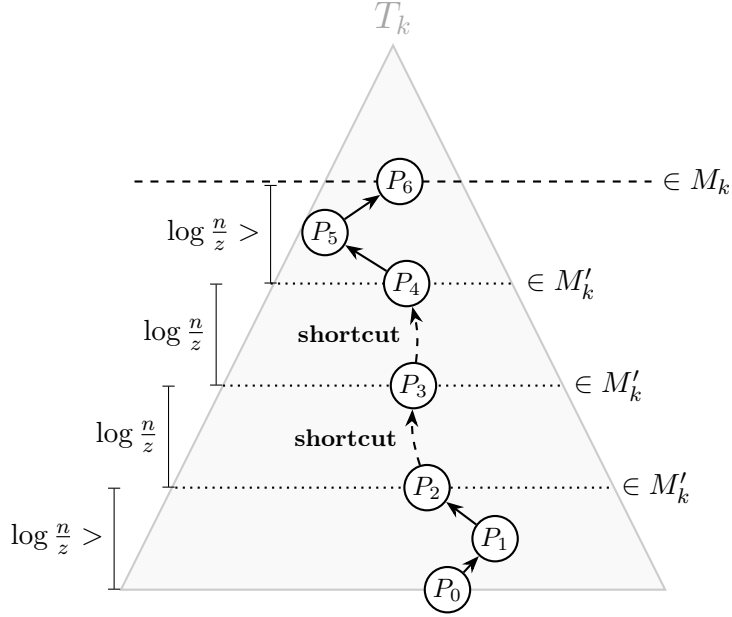
\begin{figure}
 \centering
    \begin{tikzpicture}[
    scale=0.9,
    tree node/.style={circle, draw=black!60, fill=white, minimum size=6mm, inner sep=0pt, font=\small},
    active node/.style={circle, draw=black, fill=white, thick, minimum size=6mm, inner sep=0pt, font=\small},
    marked layer/.style={black, dashed, thick},
    shortcut layer/.style={ dotted, thick},
    step line/.style={black, thick, ->, >=Stealth},
    jump line/.style={black, thick, ->, >=Stealth, dashed}
]

    \draw[fill=gray!5, draw=gray!40, thick] (0,9) -- (-4,1) -- (4,1) -- cycle;
    \node[gray!70] at (0,9.4) {\Large $T_k$};

    
    \draw[marked layer] (-3.8, 7) -- (3.8, 7) node[right] {$\in M_k$};
    \draw[shortcut layer] (-1.8, 5.5) -- (1.8, 5.5) node[right] {$\in M'_k$};
    \draw[shortcut layer] (-2.5, 4) -- (2.5, 4) node[right] {$\in M'_k$};
    \draw[shortcut layer] (-3.2, 2.5) -- (3.2, 2.5) node[right] {$\in M'_k$};

    \node[active node] (v0) at (0.8, 1) {$P_0$};
    
    \node[active node] (v1) at (1.5, 1.75) {$P_1$};
    \node[active node] (v2) at (0.5, 2.5) {$P_2$}; 
    
    \node[active node] (v3) at (0.3, 4) {$P_3$}; 
    
    \node[active node] (v4) at (0.2, 5.5) {$P_4$};
    \node[active node] (v5) at (-1, 6.25) {$P_5$};
    \node[active node] (v_target) at (0.1, 7.0) {$P_6$}; 

    \draw[step line] (v0) -- (v1);
    \draw[step line] (v1) -- (v2);
    
    \draw[jump line] (v2) to[bend left=10] node[left=2pt, font=\footnotesize\bfseries] {shortcut} (v3);
    
    \draw[jump line] (v3) to[bend right=10] node[left=2pt, font=\footnotesize\bfseries] {shortcut} (v4);

    \draw[step line] (v4) -- (v5);
    \draw[step line] (v5) -- (v_target);
   
    \draw[{|}-{|}] (-4.1,1)--(-4.1,2.5) node[midway,left]{$\log \frac{n}{z}>$};
    
    \draw[{|}-{|}] (-3.3,2.5)--(-3.3,4) node[midway,left]{$\log \frac{n}{z}$};

    \draw[{|}-{|}] (-2.6,4)--(-2.6,5.5) node[midway,left]{$\log \frac{n}{z}$};

    \draw[{|}-{|}] (-1.8,5.5)--(-1.8,6.95) node[midway,left]{$\log \frac{n}{z}>$};

\end{tikzpicture}
    \caption{A demonstration of the algorithm in \cref{lem:naive}.
    The figure demonstrates the phrases visited throughout the algorithm as seen in the tree $T_k$. 
    From $P_0$ to $P_2$ the algorithm advances by applying $J$, which leads the algorithm to the bad parent of the current node at every step (assuming that all steps are bad).
    Since $P_2$ is in $M'_k$, it has a pointer to its $\log \frac{n}{z}$ ancestor in $T_k$ (and auxiliary information that allows check if all next $\log \frac{n}{z}$ jumps from the current index are bad).
    Then, we keep applying shortcut jumps until we finally reach a layer that is within $\log \frac{n}{z}$ from a marked layer ($P_2,P_3,P_4$).
    When we are finally close enough to a marked layer, we return to jumping one-by-one ($P_4,P_5,P_6$).
    At every initial point throughout this process, we may terminate prematurely since some jump was good.
    In this event, the containing phrase is not the parent of the previously visited phrase in $T_k$, but the good jumps leads to an index with a small $r$-value.
    }\label{fig:naive-alg}
\end{figure}
\begin{proof}

 We start by describing the construction of the data structure.
 We construct the set $M$ of $O(z/\log \frac{n}{z})$ marked phrases using \cref{lem:compute-marked}, and then construct the data structure for calculating $M(i)$ for an index $i$ in a marked phrase using \cref{lem:ds-mi}.
We store a binary array $A_M$ of size $z$ representing $M$ with $A_M[i]=1$ if and only if $P_i\in M$.
 The data structure is constructed in $O((z+M \log \frac{n}{z}) \log ^2 \frac{n}{z}) = O(z \log^2\frac{n}{z})$ time and uses $O((z+ |M|\cdot \log\frac{n}{z}))=O(z)$ space.
 We also construct the data structure of \cref{lem:kbad-parent-ds} for finding $k$-bad parents in $O(1)$ time.
The construction of \cref{lem:kbad-parent-ds} requires $O(z \log^2 \frac{n}{z})$ time and $O(z)$ space.

For every $k\in \K$, we construct $T_k$ and do the following.
If the depth of $T_k$ is less than $\log \frac{n}{z}$, we do not store any additional information.
Otherwise, we pick some $c \in [0..\log \frac{n}{z})$ such that the number of phrases in $T_k$ of depth that has remainder $c$ modulo $\log \frac{n}{z}$ is at most $z/\log \frac{n}{z}$.
Such $c$ should exist due to the pigeonhole principle.
Let $M'_k$ be the set of vertices in $T_k$ with height $c$ modulo $\log \frac{n}{z}$.
For every phrase $P \in M'_k$, we store: 
\begin{enumerate}
    \item The $(\log \frac{n}{z}-1,k)$-bad interval of $P$
    \item $\delta_{(\log \frac{n}{z}-1)}(P)$
    \item The distance from $P$ to the nearest marked ancestor of $P$
    \item The $(\log \frac{n}{z},k)$-bad parent of $P$.
\end{enumerate}
For each phrase $P$ of $\LZend$, we attach a machine word $m_P$ such that the $k$'th bit of $m_P$ is $1$ if and only if $P \in M'_k$.
While computing $M'_k$ for all $k\in \K$, the values of $m_{P}$ can be set without asymptotically affecting the running time.

We use $O(1)$ space for each $P\in M'_k$ so for each value of $k$ we use $O(|M'_k|)= O(z/ \log \frac{n}{z})$ space for a total of $O(z)$ space over all $k\in \K$.
We also store the word $m_P$ for each $P\in \LZend$ which requires additional $O(z)$ space. 

Finding the correct modulo $c$ takes $O(|T_k|)= O(z)$ time.
Computing the $(\log \frac{n}{z},k)$-bad parent and $\delta_{(\log \frac{n}{z}-1)}(P)$ can be straightforwardly done in $O(\log \frac{n}{z})$ time for every vertex in $M'_k$  by traversing $\log \frac{n}{z}$ edges towards the root of $T_k$ and accumulating the $\delta$ values of the visited nodes.
Computing the distance to the nearest marked ancestor for every $P\in M'_k$, and the nearest strict ancestor in $M'_k$ could be implemented in $O(|T_k|) = O(z)$ time by a simple iteration over $T_k$.

Computing the $(\log \frac{n}{z}-1,k)$-bad interval is done in $O(\log \frac{n}{z})$ time per node in $M'_k$ by applying \cref{lem:xk-bad-exist-compute}.
The total running time per $k\in \K$ is $O( \log \frac{n}{z} \cdot |M'_k|) = O(z)$.
Over all values of $k\in \K$, the total running time is $O(z \log \frac{n}{z})$.

\paragraph{Query.}
For an illustration of the query algorithm, see \cref{fig:naive-alg}.
Let $i$ be some index query and let $k\in \K$ such that $r_i \in [k..1.5k)$.
By \cref{lem:compute-marked}, there is some $x\in [0..\log^2\frac{n}{z}]$ such that $J^x(i)$ is either contained in a marked phrase, or has $r$-value below $k$.

First, we compute $J(i),J^2(i),\ldots$ iteratively until we reach some $j$ in a phrase $P$ such that one of the following is satisfied.
\begin{enumerate}
    \item $r_j < k$
    \item $P \in M$
    \item $P\in M'_k$.
\end{enumerate}
Notice that each of the above conditions can be checked in $O(1)$ time given $j$ and $P$ (using the endpoints of $P$, the array representation of $M$, and $m_P$, respectively).

We find $J(i),J^2(i)\ldots$ as follows.
For the input index $i$, we use \cref{lem:find-con-phrase} to find the phrase containing $i$ in $O(\log \frac{n}{z})$ time.

Now, we show how given some $i'$ and the phrase $P$ containing $i'$, we can find $J(i')$ and the phrase $P'$ containing $J(i')$ in constant time.
To be more precise, we will either find the phrase $P'$ containing $J(i')$, or guarantee that one of our conditions are met for $J(i')$.
Assume we already have $i'$ and $P$ at hand.
We find $j = J(i')$ using $\delta_P$ in constant time, and we also find the bad parent $P'$ of $P$ in constant time using \cref{lem:find-kbad-constant}.
If $j$ is in $P'$ , we have the phrase containing $j$ at hand.
Otherwise, we to have $r_j < k$ due to the definition of a bad parent.
We can therefore return $j$ and terminate.
It follows that we can calculate $J(i),J^2(i),\ldots J^y(i)$ such that $J^y(i)$ satisfies one of our conditions in $O(\log \frac{n}{z} + y)$ time.

If we reach, at any point, some $j$ with $r_j < k$, the algorithm returns $j$ and terminates.
If we reach some $j$ in a marked phrase $P\in M$, the algorithm returns $M(j)$ and terminates.

If we end up reaching some $i' \in P$ for some $P$ in $M'_k$, we stop iteratively applying $J$ and switch to the following strategy instead.
We use the stored data of $P$ to check, in constant time, if $i'$ is in the $(\log \frac{n}{z}-1, k)$-bad interval of $P$, and if the distance from $P'$ to its nearest ancestor in $M$ is less than $\log \frac{n}{z}$.
We consider several cases depending on results of these two tests.
\paragraph{Case 1.}
The distance to the next marked ancestor is more than $\log \frac{n}{z}$, and $i'$ is in the $(\log \frac{n}{z}-1,k)$-bad interval of $P$.

In this case, we compute $j'= J^{\log \frac{n}{z}}(i')$ in constant time, using $\delta_{\log \frac{n}{z}-1}(P)$ and applying \cref{cor:fast-jump-formula}. 
We also have the phrase $P'$ containing $j'$ stored in $P$ as the $(\log \frac{n}{z},k)$-bad parent of $P$, so we can find it $P'$ in $O(1)$ time. 
Notice that $P'\in M'_k$ since it has depth $c$ modulo $\log \frac{n}{z}$.
Therefore, the algorithm will keep iteratively using shortcuts as long as we remain in Case 1.

\paragraph{Case 2.}
$i'$ is not in the $(\log \frac{n}{z}-1,k)$-bad interval of $P$.

In this case, we switch back to applying $J$ as at the start of the algorithm.
We keep applying $J$ until we reach an index $j$ with $r_j <k$.
We will later show that this process is guaranteed to terminate within $O(\log \frac{n}{z})$ steps.

\paragraph{Case 3.}
The distance to the next marked ancestor is less than $\log \frac{n}{z}$.

In this case, we switch back to applying $J$ until we reach an index $j$ contained in a marked phrase $P'\in M$ or with $r_j<k$.
When this occurs, we return $j$ if $r_j<k$ or $M(j)$ if $j$ is contained in $P'\in M$. 
We will later show that we are guarantee to find such $j$ within $O(\log \frac{n}{z})$ steps.

\paragraph{Running time.}
The first part of the algorithm runs in $O(\log \frac{n}{z} +y)$ time where $j=J^y(i)$ is the first index we meet satisfying one of the specified conditions.
Notice that if we jump from $i\in P$ to a vertex $i'= J(i)$ that is not the $k$-bad parent of $P$, the algorithm terminates because $r_{i'}<k$ by the definition of a bad parent.
Therefore, the phrases containing $i,J(i),J^2(i)...$ correspond to an upwards path in $T_k$, as long as the algorithm does not terminate.
It follows that we are guaranteed to reach a phrase $P$ that has level $c \bmod(\log \frac{n}{z})$ in $T_k$, or terminate for some other reason within $O(\log \frac{n}{z})$ steps.
We have shown that $y\in O(\log \frac{n}{z})$ and therefore the running time of the first part of the algorithm is $O(\log \frac{n}{z})$.

If the first part of the algorithm reaches some index $j$ with $r_j<k$ or in a phrase $P\in M$, the algorithm terminates in $O(1)$, returning either $j$ or $M(j)$.

Otherwise, the algorithm reaches some index $j$ in a phrase $P\in M'_k$.
In this case, the algorithm checks into which of the three cases $j$ falls and applies either Case 1, Case 2, or Case 3.

Case 1 is applied in $O(1)$ time.
We claim that it can be applied at most $\log \frac{n}{z}$ times.
After $\log \frac{n}{z}$ application, we reach some $j=J^y(i)$ for $y\ge \log^2\frac{n}{z}\ge x$.
Since we did not reach or skip over a marked phrase of $P \in M$, we have that $J^x(i) \notin M$ and therefore the $r$-value of $J^x(i)$ is less than $k$.
Since applying $J$ never increases the $r$-value (\cref{lem:j-legal}), we have $r_j \le r_{J^x(i)} < k$, which means that the algorithm terminates and returns $j$.

It remains to analyze the running times of Case 2 and Case 3.
Notice that each of them, by definition, occurs only once - when each of these two cases occur, they define the behavior of the remainder of the algorithm.

Case 2 occurs where $i'$ is not in the $(\log \frac{n}{z}-1,k)$-bad interval of $P'$. 
Therefore, there is some $y\in[0..\log \frac{n}{z})$ such that $J^y(i')$ is a good jump, which means that we will reach an index satisfying one of the desired conditions with $O(\log \frac{n}{z})$ jumps.

Case 3 occurs when the distance to the nearest marked ancestor is less than $\log \frac{n}{z}$.
Clearly, either there is $y\in [\log \frac{n}{z}]$ such that $J^y(i')$ is a good jump, leading to termination as in the previous case, or all $\log \frac{n}{z}$ jumps following $i'$ are contained in the ancestors of $P$ containing $i'$, in the corresponding order.
In the latter case, we are guaranteed to reach some $j$ contained in $P\in M$ within $\log \frac{n}{z}$ steps, as required.

In conclusion, the first part of the algorithm is executed in $O(\log \frac{n}{z})$ time, Case 1 may be applied up to $O(\log \frac{n}{z})$ times, with each application taking $O(1)$ time, and Case 2 and 3 may take $O(\log \frac{n}{z})$ time.
The total running time of the query algorithm is $O(\log \frac{n}{z})$.

\paragraph{Correctness.}
The algorithm terminates by finding some $j = J^y(i)$ that either has $r_j < k$ or that is contained in a marked phrase.
In the first case, the algorithm returns $j$ which is clearly a valid output.
In the latter case, the algorithm returns $M(j)$ which satisfies either $r_j < r_i/2 <1.5k/2 < k$ or $\ell_j \le r_i /2 <1.5k$ due to \cref{clm:m-reduce-l-or-r}.
Therefore, $M(j)$ is a valid output.
In both cases, we have $S[x]=S[i]$ and $r_x\le r_i$ (where $x$ is the output index) since $x$ is obtained by iteratively applying legal jumping rules (either $J$ or $M$) starting from $i$. 
\end{proof}

\section{The Stable Part}\label{sec:stable} 
In this section, we prove \cref{lem:stable}.

\paragraph{The jumping rule $sJ$.}
We define the jumping rule $sJ$.
For every $P=S[a..b]\in \LZend$, we denote $L(P) = S[a..a+\floor{\frac{2}{3}|P|}-1]$ the prefix of size $\floor{\frac{2}{3}|P|}$ of $P$.
Let us denote the set $\Lefts=\{ L(P) \mid P \in \LZend \}$.

We define $sJ(i)$ only for indices $i$ that are contained in some $L\in \Lefts$.
\begin{definition}[Jumping Rule $sJ$]
    For an index $i\in [1..n]$ contained in the phrase $P[a..b]$.
    If $i\in [a..a+\floor{\frac{2}{3}|P|}-1]$, we define $sJ(i) = i-\delpre(a,a+\floor{\frac{2}{3}|P_i|}-1)$.
    Otherwise, $sJ(i)$ is undefined.
\end{definition}
Although not explicitly stated as a jumping rule, the logic expressed by $sJ$ is used in \cite{KS22} as well.
It follows from \cref{obs:pre-is-legal} that $sJ$ is a legal jumping rule.

On several occasions throughout the section, we are required to prove that $sJ(i)$ is defined for some index $i$. 
The following provides a clean characterization of the event in which $sJ(i)$ is defined for an index $i$.
\begin{observation}\label{obs:sJ-defined-nice}
    $sJ(i)$ is defined for an index $i\in [n]$ if and only if $\ell_i \le 2r_i$.
\end{observation}

As we did for $M$ in \cref{sec:naive}, we use \cref{lem:computepre} to construct the following data structure for computing $sJ$.
\begin{lemma}\label{lem:ds-sji}
Given $\LZend$, we can construct a data structure of size $O(z)$ that given $i\in [n]$ and the phrase $P$ containing $i$, computes $sJ(i)$ in $O(1)$ time.
The data structure can be constructed in $O(z \log^2 \frac{n}{z})$ time.
\end{lemma}
\begin{proof}
We apply \cref{lem:computepre} to compute for every $P=S[a..b]\in \LZend$ the value $P_a = \pre(a,a+\floor{\frac{2}{3}|P|}-1)$ in $O(z \log^2 \frac{n}{z})$ time.
Notice that every pair $(b,e)$ in the query set satisfies $[b..e] \subseteq [a_i..b_i]$ for some $P_i$, and we can assume due to \cref{lem:no-big-phrases} that $|P_i| \le \frac{n}{z}$.
Therefore, we indeed have $e \in [b..b+\frac{n}{z}]$ for every query pair, so this is a valid input for \cref{lem:computepre}.

We then compute for each phrase $P$ the value $\delpre(P)=\delpre(a,a+\floor{\frac{2}{3}|P|}-1) = a-P_a$.
We store an array of length $z$, where the $i$'th entry stores $\delpre(P_i)$.
Now, given an index $i$ and the phrase $P$ containing $i$, we can simply inspect the array to obtain $\delpre(i)$ and return $sJ(i) = i-\delpre(i)$.
The construction time is dominated by $O(z \log^2 \frac{n}{z})$, the query time is $O(1)$, and the size is $O(z)$ for storing all $\delpre(P)$ values.
\end{proof}

We assume that for every phrase $P$ with $L(P) = S[b..e]$, we have access to $\delpre(b,e)$.
This is achieved as a consequence of applying \cref{lem:computepre} with the query set $\Lefts$.

\paragraph{High-level idea for proving \cref{lem:stable}}
Let us fix some index $i_0$ as the input for the stable part, and denote as $k\in \K$ the unique integer power of $1.5$ such that $r_{i_0} \in [k..1.5k)$.
We are interested in finding the first integer $x^*$ such that $sJ^{x^*}(i_0)$ has $r$-value less than $k$.
As in \cref{sec:naive}, we will show that there is a tree structure capturing the bad sequences of applications of $sJ$.
Unlike in \cref{sec:naive}, where we had a separate tree $T_k$ for every possible $k\in \K$ capturing the bad sequences with $r$-values roughly $k$, here we will have a single tree $T_{\Bad}$ capturing all bad sequences.

This allows us to store $T_{\Bad}$ and preprocess it for level-ancestor queries.
We provide a characterization of the sought value $x^*$ such that $sJ^{x^*}(i_0)$ has $r$-values less than $k$.
The characterization allows us to assign a numerical values $\delta_\ell(P)$ and $\delta_r(P)$ to every phrase $P$.
Then, when given $i_0$ contained in some phrase $P_0$, we can identify $x^*$ as a level in $T_{\Bad}$ where the aggregated sums of $\delta_{\ell}(P)$ and $\delta_r(P)$ on the path from $P_0$ to its ancestor in this level reach some threshold.
We compute this sum in logarithmic time using level-ancestor queries.
Given $x^*$, and auxiliary information stored in $T_{\Bad}$, we can obtain $sJ^{x^*}(i_0)$ in constant time.
 
\paragraph{Existence of $x^*$}

Recall that $sJ(i)$ is not defined for every index in $i\in [n]$.
Therefore, it is not immediate that $x^*$ exists, as it may be the case that when repeatedly applying $sJ$, we reach an index $i'=sJ^{x'}(i_0)$ for which $sJ(i')$ is undefined.
In the following lemma, we prove that if this process reaches such an index, we have $r_{i'}<k$, which means that $x' =x^*$.
In \cite{KS22}, the following properties are proven, but they are hidden within the proof of another claim (See Lemma 4.4 in \cite{KS22}).
We repeat their proof for the sake of self-containment.

\begin{lemma}[Stable Part Properties]\label{lem:stable-properties}
    Let $i$ be an index with $r_i \in [k..1.5k)$ and $\ell_i < 1.5k$.
    Let $j = sJ(i)$.
    Either $r_j < k$, or $\ell_j \le \ell_i$
\end{lemma}
\begin{proof}
For an illustration, see \cref{fig:stable-properties}.

Since $\ell_i \le 1.5k \le 1.5 r_i<2r_i$, we have that $sJ(i)$ is well defined by \cref{obs:sJ-defined-nice}.
Let $P=S[a..b]$ be the phrase containing $i$.
Let $\pre(a,a+\floor{\frac{2}{3}|P|}-1) = a'$ and $b' = a'+ \floor{\frac{2}{3}|P|}-1$.
By definition of $j=sJ(i)$, we have that $j\in [a'..b']$, and $b'-j= a+\floor{\frac{2}{3}|P|} - 1 -  i $, and $a'-j = a-i$.
By the definition of $\pre$ there is a phrase boundary $\hat b \in [a'..b']$.
If $j< \hat b$, we have that $r_j \le \hat b - j \le b' - j$, which implies 
$$r_{j} \le  b' - j = a+\floor{\frac{2}{3}|P|} -1-i  = b - i -\ceil{\frac{|P|}{3}} = r_i - \ceil{\frac{|P|}{3}} \le \floor{\frac{2}{3}r_i} < 2/3 \cdot 1.5k =  k$$.
And we have $r_j <k$, as required.

Otherwise, we have $j>\hat b$ which means that $ \ell_{j} \le j-\hat b \le  j - a' = i-a =\ell_i$, as required.
\end{proof}

\begin{figure}[htbp]
    \centering
   \begin{subfigure}[b]{0.45\textwidth}
      \centering
        \tikzset{every picture/.style={line width=0.75pt}} 
\usetikzlibrary{patterns}

\begin{tikzpicture}[x=1pt,y=1pt,yscale=-1,xscale=1]
\pgfmathsetmacro{\blockheight}{25}
\pgfmathsetmacro{\stringheight}{20}
\pgfmathsetmacro{\height}{20}

\pgfmathsetmacro{\width}{10}
\pgfmathsetmacro{\x}{0} 

\draw (20,0) rectangle (180,\blockheight);

\draw[dotted] (22,2) rectangle (130,\blockheight-2);

\draw[densely dotted](20,50) rectangle (130, \blockheight+50);

\draw(35,0) rectangle (45,\blockheight);
\node[anchor=north,font=\small] at(40,5) {$i$};

\draw(35,50) rectangle (45,50+\blockheight);
\node[anchor=north,font=\small] at(40,55) {$j$};

\node[anchor=north,font=\tiny] at(40,-10) {$r_i$};
\draw[->] (45,-5) -- (180,-5);

\draw[dashed, line width = 0.7] (60,35) -- (60,65+\blockheight);

\draw[dashed, line width = 0.7] (180 ,-15) -- (180 ,55+\blockheight);

\draw[{|}-{|}] (132,18)--(178,18) node[midway,above]{$|P|/3$};

\node[anchor=north,font=\tiny] at(40,40) {$r_j$};
\draw[->] (45,45) -- (60,45);
\draw[->,dotted] (60,45) -- (180,45);

\node[anchor = east] at (17,12) {$P$};

\node[anchor = east] at (17,62) {$\pre$};

\end{tikzpicture}
        \caption{ If $j$ is to the right of the phrase boundary, its $r$-value is reduced by at least $|P|/3$, which results in its total value decreasing by a factor of $2/3$. 
        }
    \end{subfigure}
    \hfill
    \begin{subfigure}[b]{0.45\textwidth}
          \centering
        \tikzset{every picture/.style={line width=0.75pt}} 
\usetikzlibrary{patterns}

\begin{tikzpicture}[x=1pt,y=1pt,yscale=-1,xscale=1]
\pgfmathsetmacro{\blockheight}{25}
\pgfmathsetmacro{\stringheight}{20}
\pgfmathsetmacro{\height}{20}

\pgfmathsetmacro{\width}{10}
\pgfmathsetmacro{\x}{0} 

\draw (20,0) rectangle (180,\blockheight);

\draw[dotted] (22,2) rectangle (130,\blockheight-2);

\draw[densely dotted](20,50) rectangle (130, \blockheight+50);

\draw(75,0) rectangle (85,\blockheight);
\node[anchor=north,font=\small] at(80,5) {$i$};

\draw(75,50) rectangle (85,50+\blockheight);
\node[anchor=north,font=\small] at(80,55) {$j$};

\node[anchor=north,font=\tiny] at(80,-12) {$\ell_i$};
\draw[->] (75,-5) -- (20,-5);

\draw[dashed, line width = 0.7] (60,35) -- (60,65+\blockheight);

\node[anchor=north,font=\tiny] at(80,38) {$\ell_j$};
\draw[->] (75,45) -- (60,45);
\draw[->,dotted] (75,45) -- (20,45);

\node[anchor = east] at (17,12) {$P$};

\node[anchor = east] at (17,62) {$\pre$};

\end{tikzpicture} 
        \caption{If $j$ is to the left of the phrase boundary, its $\ell$ value is clearly smaller than that of $i$.}

    \end{subfigure}
    \caption{Illustration of the proof of \cref{lem:stable-properties}.
    The dotted upper rectangle represents $L(P)$ and the bottom doted rectangle represents the occurrence of $L(P)$ mapped by $\pre(a,a+\frac{2}{3}|P|)$.
        By definition of $\pre$, the lower rectangle contains a phrase boundary.}\label{fig:stable-properties}
\end{figure}
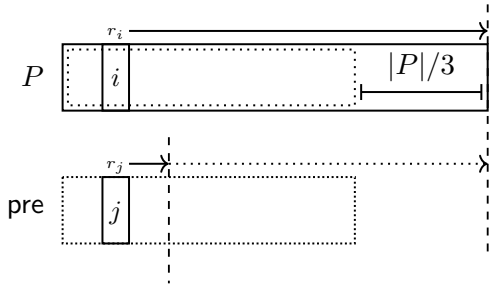
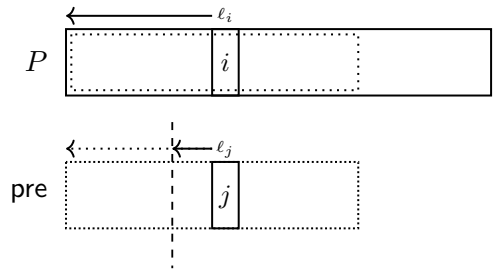
The following is an immediate corollary of \cref{lem:stable-properties}.
\begin{corollary}\label{col:x-star-exists}
    Let $x$ be a non-negative integer.
    If $j = sJ^x(i_0)$ is defined, either $r_j < k$ or $sJ(j)$ is also defined.
\end{corollary}
\begin{proof}

Initially have $\ell_{i_0} \le 1.5k$.
Assume that $r_j \ge k$.
Since application of $sJ$ can not decrease the $r$-value (\cref{cor:M-sj-legal}), we have that $r_{sJ^y(i_0)} \ge k$ for every $y\in [0..x]$. 
It follows from \cref{lem:stable-properties} that if $\ell_{sJ^y} \le \ell_{sJ^{y-1}}$ for every $y\in [1..x]$ and in particular $\ell_j \le \ell_0 \le 1.5k$. 
We therefore have $\ell_j \le1.5k \le 1.5 r_{j}\le 2r_{j}$.
By \cref{obs:sJ-defined-nice}, we have that $sJ(j)$ is well defined.
\end{proof}
With \cref{col:x-star-exists}, we have shown that $x^*$ exists.
For the rest of the section, we show how to efficiently find $x^*$ and $sJ^{x^*}(i_0)$.

\paragraph{Tree structure.}
We override the definition of bad jumps from \cref{sec:naive} to match our current context.
First, we define the bad phrase and the bad suffix of $L\in \Lefts$ as follows.
\begin{definition}[Bad Suffix ,Bad Phrase]
    For $L=S[a..b]\in \Lefts$, the bad phrase $P'$ of $L$ is the phrase containing $sJ(b) = b-\delpre(a,b)$.
    The bad suffix of $L$ is the suffix $[c..b]$ containing exactly the indices $i$ such that $sJ(i)$ is in $P'$.
    We lift this terminology to phrases.
    The bad suffix (resp. bad phrase) of a phrase $P$ is the bad suffix (resp. bad phrase) of $L(P)$ (notice that the bad suffix of $P$ is typically not a suffix of $P$).
\end{definition}
Notice that unlike in the previous section, where we associated each phrase $P$ with possibly $|\K| \in O(\log \frac{n}{z})$ 'bad' phrases, here we only define one bad phrase for $P$.

In the following, we establish the importance of bad phrases and bad suffixes, linking them to the terminal step that finally leads to an index with sufficiently small $r$-value.
\begin{lemma}\label{lem:stable-bad-jumps-same-phrase}
    Let $P=S[a..b]$ be a phrase.
    Let $i$ be an index in $L(P)$ with $r_i \in [k..1.5k)$ for some $k\in \K$ and $i'=sJ(i)$.
    If $r_{i'} \ge k$, then $i$ is in the bad suffix of $L$.
\end{lemma}
\begin{proof}
    Denote $L = L(P)=S[a..c]$ and notice that $c=a+ \floor{\frac{2}{3}|P|}-1$.
    Notice that $r_i= (c-i) + \ceil{\frac{1}{3}|P|}$.
    Also, since $i$ is in the phrase $P$, we have $r_i < |P|$.
    It follows that $c-i < \frac{2}{3}r_i$.
    Let $\pre(a,c) = S[a'..c']$.
    Clearly, $c'$ is in the rightmost phrase that intersects $[a'..c']$.
    We also have that $sJ(c) = c'$ and $c-i= c'-i'$ since $sJ$ subtracts  $\delpre(a,c)$ from both $i$ and $c$.
    We will now show that if  $i'$ is not in the same phrase as $c'$, then $r_{i'} < k$.
    If $i'$ and $c'$ are not in the same phrase, there is a boundary between $i'$ and $c'$.
    This would imply that $r_{i'} \le |i'-c'| = |i-c|< \frac{2}{3}r_i<k$, as required.
\end{proof}

We define a tree structure over the phrases.
Let $T_{\Bad}$ be a tree with vertices being the phrases of $\LZend$.
For every phrase $P$, the parent of $P$ in $T_{\Bad}$ is the bad phrase of $P$.
We denote the $x$'th ancestor of $P$ in $T_{\Bad}$ as $\Parent{P}{x}{\Bad}$.

Similarly to \cref{sec:naive}, we have that if applying $sJ(i)$ does not result in an $r$-value below $k$, the index $i$ is in the bad suffix of $P$.
Therefore, we focus on analyzing the structure of a sequence of $sJ(i)$ applications in which every step falls within the bad suffix of the phrase containing it.
We make the following simple observation, showing that the difference in $\ell$ and $r$ values as a result of applying $sJ(i)$ is identical across all $i$ in the bad suffix of a phrase $P$.
\begin{lemma}\label{lem:bad-jump-same-delta}
    Let $i_1$ and $i_2$ be two indices in the bad suffix of a phrase $P$.
    Let $j_1=sJ(i_1)$ and $j_2 = sJ(i_2)$.
    It holds that $\ell_{i_1} - \ell_{j_1} = \ell_{i_2}-\ell_{j_2}$ and $r_{i_1} - r_{j_1} = r_{i_2}-r_{j_2}$.
\end{lemma}
\begin{proof}
Let $L=S[x..y]=L(P)$.
We have $j_1 = i_1 - \delpre(x,y)$ and $j_2 = i_2 - \delpre(x,y)$.
Therefore, $i_1 - i_2 = j_1 - j_2$.
Denote $P=S[a..b]$ and let $P'= S[a'..b']$ be the bad phrase of $P$.
By definition, $r_{i_1} = b-i_1$,  $r_{j_1} = b'-j_1$, $r_{i_2} = b-i_2$,  and $r_{j_2} = b'-j_2$.
It immediately follows that  $r_{i_1} - r_{j_1} = r_{i_2}-r_{j_2}$.
The claim follows for $\ell$-values in a similar manner.
\end{proof}

Following \cref{lem:bad-jump-same-delta}, we introduce additional notation to capture the common change in $\ell$-values and in $r$-values among indices in the bad suffix.
For a phrase $P=S[a..b]$, with $L(P)= L=[a..b']$, we define $\delta_{\ell}(P) = \ell_{b'}-\ell_{sJ(b')}$ and $\delta_r(P)=r_{b'}-r_{sJ(b')}$.
Notice that $b'$, as the last index of $L$, is in the bad suffix of $L$.

The following is directly implied by \cref{lem:bad-jump-same-delta}.
\begin{corollary}\label{cor:del-is-del}
    Let $i$ be an index in the bad suffix of a phrase $P$ and let $j=sJ(i)$.
    It holds that $r_j = r_i - \delta_r(P)$ and $\ell_j = \ell_i - \delta_{\ell}(P)$.
\end{corollary}

We notice that for every $P$ we have that $\delta_{\ell}(P)$ and $\delta_r(P)$ are always non-negative.
\begin{lemma}\label{lem:delt-l-non-neg}
    For every phrase $P\in \LZend$, it holds that $\delta_r(P)\ge 0$ and $\delta_\ell(P) \ge 0$. 
\end{lemma}
\begin{proof}
Let $P=S[a..c]$ with $L(P) = S[a..b]$.
Let $a' = \pre(a,b)$ and $b' = a' + |L(P)|-1$.
By definition, we have $sJ(b) = b- \delpre(a,b) =b'$.
From the definition of $\pre$, we have some phrase boundary $\hat b \in [a'..b']$.
Therefore, $\ell_{b'} \le b'-\hat b \le b'-a' = b-a = \ell_b$.
It follows that $\delta_{\ell}(P) = \ell_b-\ell_{b'} \ge 0$, as required. 

The fact that $\delta_r(P)\ge 0$ follows from $sJ$ being a legal jumping step (\cref{obs:pre-is-legal}), which implies $r_{b'}\le r_b$.
\end{proof}

Let us define the set of $(x)$-bad indices of a phrase.
\begin{definition}
    For an integer $x\ge 0$, an index $i$ in phrase $P$ is $(x)$-bad if for every $y\in [0..x]$ it holds that $sJ^y(i)$ in the bad suffix of $P^{(y)}(i)$.
\end{definition}

We also define, for every phrase $P$ with $L(P)=S[a'..b']$ the value $\delta_{\Bad}(P) = \delpre(a',b')$.
Recall that for an index $i$ in $L(P)$, it holds that $sJ(i) = i - \delpre(a',b')$.
The following follows directly from the definition of an $(x)$-bad jump.

\begin{observation}\label{obs:sum-del-direct-jump}
If $i$ is a $(x)$-bad in a phrase $P$ for some non-negative integer $x$, it holds that $J^{x+1}(i) = i- \sum_{y=0}^{x} \delta_{\Bad}(\Parent{P}{y}{\Bad})$. 
\end{observation}

We exploit \cref{obs:same-jumps-same-shift} to show that we can quickly find $sJ^{x+1}(i)$, provided that $i$ is $(x)$-bad.

\begin{lemma}\label{lem:algorithm-jump-bad-quickly}
    Given $T_{\Bad}$, we can construct in $O(z)$ time a data structure of size $O(z)$ supporting the following query.
    Given an index $i\in [1..n]$ , a non-negative $x$ such that $i$ is $(x)$-bad, the phrase $P$ containing $i$, and $P'= \Parent{P}{x}{\Bad}$, report $sJ^{x+1}(i)$.
    The query time is $O(1)$. 
\end{lemma}
\begin{proof}
We attach to every node of $T_{\Bad}$ the number $\Delta'_{\Bad}(P)$ that is equal to the sum of $\delta_{\Bad}$ values on the path from $P$ to the root.
Formally, $\Delta'_{\Bad}(P) = \sum_{y=0}^{d}\delta_{\Bad}(\Parent{P}{y}{\Bad})$ where $d$ is the depth of $P$ in $T_{\Bad}$.
This concludes the construction of the data structure.

Given a query $i$,$x$, $P$, and $P'=\Parent{P}{x}{\Bad}$, we retrieve $\Delta_{\Bad}(P)$ and $\Delta_{\Bad}(\Parent{P}{x+1}{\Bad})$ from $T_{\Bad}$.
Notice that we can find $\Parent{P}{x+1}{\Bad}$ in constant time, as the parent of $P'$.
We compute $\delta = \Delta'_{\Bad}(\Parent{P}{x+1}{\Bad})-\Delta'_{\Bad}(P) + \delta_{\Bad}(P) = \sum_{y=0}^{x} \delta_{\Bad}(\Parent{P}{y}{\Bad})$.
We output $sJ^{x+1}(i) = i - \delta$, which is correct by \cref{obs:sum-del-direct-jump}.

The preprocessing time consists of finding all $\Delta'_{\Bad}$ values, which can be done in $O(|T_{\Bad}|)$ via a straightforward iteration on $T_{\Bad}$ (given the values $\delta_{\Bad}(P)$ of all phrases).
The total construction time is $O(|T_{\Bad}|) = O(z)$, as required.

The query consists of retrieving data stored in the tree.
Assuming that every phrase stores a link to the corresponding tree node, the stored data can be retrieved in constant time, as required.
\end{proof}

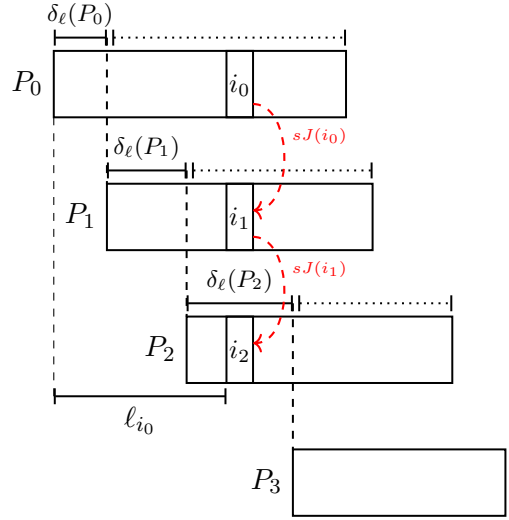
\begin{figure}
\begin{center}
    \tikzset{every picture/.style={line width=0.75pt}} 
\usetikzlibrary{patterns}

\begin{tikzpicture}[x=1pt,y=1pt,yscale=-1,xscale=1]
\pgfmathsetmacro{\blockheight}{25}
\pgfmathsetmacro{\stringheight}{20}
\pgfmathsetmacro{\height}{20}

\pgfmathsetmacro{\width}{10}
\pgfmathsetmacro{\x}{0} 

\draw[{|}-{|}] (20,-5) -- (40,-5) node[midway,above]{\footnotesize $\delta_{\ell}(P_0)$};

\draw[{|}-{|},dotted] (42,-5) -- (130,-5);
\draw (20,0) rectangle (130,\blockheight);

\draw[dashed] (40,-5) -- (40, 50 );
\draw[{|}-{|}] (40,45) -- (70,45) node[midway,above]{\footnotesize $\delta_{\ell}(P_1)$};

\draw[{|}-{|},dotted] (72,45) -- (140,45);

\draw(40,50) rectangle (140, \blockheight+50);

\draw[dashed] (70,45) -- (70, 100 );
\draw(70,100) rectangle (170, \blockheight+100);

\draw[{|}-{|}] (70,95) -- (110,95) node[midway,above]{\footnotesize $\delta_{\ell}(P_2)$};
\draw[{|}-{|},dotted] (112,95) -- (170,95);
\draw[dashed] (110,95) -- (110, 150 );

\draw(110,150) rectangle (190, \blockheight+150);

\draw(85,0) rectangle (95,\blockheight);
\node[anchor=north,font=\small] at(90,5) {$i_0$};

\draw(85,50) rectangle (95,50+\blockheight);
\node[anchor=north,font=\small] at(90,55) {$i_1$};

\draw(85,100) rectangle (95,100+\blockheight);
\node[anchor=north,font=\small] at(90,105) {$i_2$};

\draw[{|}-{|}] (85,105+\blockheight) -- (20,105+\blockheight) node[midway,below]{$\ell_{i_0}$};
\draw[dashed,thin] (20,\blockheight) -- (20,105+\blockheight);

\node at (10,12) {$P_0$};

\node at (30,62) {$P_1$};

\node at (60,112) {$P_2$};

\node at (100,162) {$P_3$};

\node[red] at(120,32) {\tiny $sJ(i_0)$};
\draw[->,red,dashed,thick] (95,20) to[out=0,in=0] (95,60);

\node[red] at(120,82) {\tiny $sJ(i_1)$};
\draw[->,red,dashed,thick] (95,70) to[out=0,in=0] (95,110);

\end{tikzpicture}
    \caption{The relationship between the sum of $\delta_\ell$ values and membership in the bad suffix.
    In this figure, the top phrase is the phrase $P_0$ containing $i_0$. 
    Below each phrase, the figure displays its bad phrase.
    The bad suffix of each phrase is denoted above the phrase as a dotted line.
    For the sake of clear presentation, the figure displays each bad suffix as a suffix of the phrase (in reality, the bad suffix of $P$ is a suffix of $L(P)$.) 
    Notice that $\delta_{\ell}(P)$ is always equal to the offset between the occurrence of $L(P)$ mapped by $\pre$ and the bad phrase of $P$.
    The index $i_0$ is in the bad suffix of $P_0$ since $\ell_{i_0}\le \delta_{P_0}$.
    It is therefore mapped by $sJ(i_0)$ to an index $i_1$  that is in the bad parent of $P_0$, which is $P_1$.
    Since $\ell_{i_0} > \delta_{\ell}(P_0)+\delta_{\ell}(P_1)$, we also have that $i_1$ is in the bad suffix of $P_1$, so it is mapped to $i_2$ in $P_2$.
    Since $\ell_{i_0}$ is not larger than the sum $\delta_{\ell}(P_0)+\delta_{\ell}(P_1)+\delta_{\ell}(P_2)$, we have that $i_2$ is not in the bad suffix of $P_2$ and is not mapped to $P_3$.
    }\label{fig:sum-left-bad}
\end{center}
\end{figure}
\paragraph{Characterizing $x^*$.}
We will use the $\delta_r$ and $\delta_\ell$ values of the phrases to provide a characterization of the first integer $x^*$ such that $j = sJ^{x^*}(i_0)$ has $r_j < k$.
Before formally providing the characterization (see \cref{lem:connecting_everything}), let us provide some intuition.

Consider the process of iteratively applying $sJ$, starting from $i_0$ until $j$ is reached.
Before $j$ is reached, it should be clear that all visited indices are bad, and it follows from \cref{lem:stable-bad-jumps-same-phrase} that every visited index is in the bad suffix of the phrase containing it.
Therefore, the phrases visited correspond to an upwards path in $T_{\Bad}$.
From \cref{cor:del-is-del}, it follows that the decrease in $\ell$ and $r$ values when we jump from $i$ to $j$ is exactly the $\delta_{\ell}$ and $\delta_r$ values of the phrase containing $i$, respectively.

In particular, if $i_D = sJ^{D+1}(i_0)$ for some integer $D$, it holds that $\ell_{i_D} = \ell_{i_0} -\sum_{y = 0}^D\delta_{\ell}(P^{(y)})$ and similarly, $r_{i_D} = r_{i_0} -\sum_{y = 0}^D\delta_{r}(P^{(y)})$.
The two sums representing the total decrease in $\ell$ and $r$ correspond to the sum  $\delta_\ell$ values and $\delta_r$ values of an upwards path from $P$ in $T_{\Bad}$, respectively.
We can use level ancestor queries to quickly identify the lowest ancestor of $P$ in which the sum crosses a given threshold.

A natural threshold to consider for the sum of $\delta_r$ is $k-r_{i_0}$, as once this threshold is surpassed, the $r$-value drops below $k$, which means that we have reached an index with sufficiently small $r$-value.
However, this bound alone is not sufficiently precise.
It may be the case that at some point, before the $\delta_r$ values accumulate to $k-r_{i_0}$, we visit some index $j$ that is not in the bad suffix of the phrase containing it.
If that is the case, our claim that the $r$ value decreases according to the $\delta_r$ values no longer holds (as its jumps no longer correspond to an upward path).
Fortunately, according to \cref{lem:stable-bad-jumps-same-phrase}, if such $j'$ is reached, either $j'$ or its successor $sJ(j')$ have $r$-value below $k$.
Furthermore, it turns out that the first such $j$ can be characterized using the accumulated $\delta_\ell$ values.

To see this relationship between the $\delta_\ell$ values in the first $j$ that is not in a bad suffix, consider first $\delta_\ell(P)$ (See \cref{fig:sum-left-bad}).
Let $L=S[a' ..b']$ be the bad suffix of $P=S[a..b]$.
It should be relatively easy to see that $\delta_{\ell}(P)=a'-a$.
Now, it is clear that $i$ is in the bad suffix of $P$ if and only if $\ell_i =i-a \ge a'-a=\delta_{\ell_P}$.
Assume that $i$ is indeed in the bad suffix of $P$ and consider the bad phrase $P'$ of $P$ that contains $j=sJ(i)$.
We can apply the same reasoning and show that $j$ is in the bad suffix of $P'$ if and only if $\ell_j \ge \delta_{\ell}(P')$.
But when we reached $P'$, we already decreased the $\ell$ value of $i$ by $\delta_{\ell}(P)$.
So, the condition $\ell_j \ge \delta_{\ell}(P')$ translates to $\ell_i \ge \delta_{\ell}(P)+\delta_{\ell}(P')$.

This pattern persists inductively, and we have that $sJ^x(i)$ is in the bad suffix of the phrase containing it 'if and only if' the sum of $\delta_{\ell}$ values does not exceed $\ell_i$.
Again, 'if and only if' is not completely precise here, because it may be the case that the $r$-value drops below $k$ while visiting a bad suffix.
Recall that we are only guaranteed to have the value of $sJ$ well defined for the indices we meet in this process as long as their $r$-value exceeds $k$ (\cref{col:x-star-exists}).

In conclusion, there are two events that may lead to the $r$-value dropping below $k$ for the first time.
One event corresponds to passing a certain threshold of accumulated $\delta_r$ values and the other event corresponds to passing a certain threshold of accumulated $\delta_{\ell}$ values.
We need to specifically find the earlier of the two events, as the characterization of the $r$-values and the $\ell$-values on which we rely persists only as long as none of these two events occur.

We formalize the above intuition as follows.
For any phrase $P\in \LZend$ and integer $B$, we define $\Delta_L(P,B) = \sum_{x=0}^{B} \delta_{\ell}(\Parent{P}{x}{\Bad})$.
Similarly, we define $\Delta_R(P,B) = \sum_{x=0}^{B} \delta_{r}(\Parent{P}{x}{\Bad})$.

We prove the following fact.
\begin{lemma}\label{lem:connecting_everything}
   
   Let $P$ be the phrase containing $i_0$ with $\ell_{i_0} = L$ and $r_{i_0} = R$.
   Let $B$ be a non-negative integer.
    \begin{enumerate}
        \item If $i_0$ is a $(B)$-bad index, than $\Delta_L(P,B) \le L$.
        \item If $\Delta_{L}(P,B) \le L$ and $\Delta_R(P,B-1) < R-k$, then $i_0$ is a $(B)$-bad index.
    \end{enumerate}

\end{lemma}
\begin{proof}
For the sake of clear presentation, we omit $P$ from the notations $\Delta_L(P,B)$ and $\Delta_R(P,B)$, and write $\Delta_L(B)$ and $\Delta_R(B)$ instead, respectively.

Let us first prove the first statement.
Assume that $i_0$ is $B$-bad.
Let $j= sJ^{B+1}(i_0)$
Since for every $x\in [0..B]$ we have that $sJ^x(i_0)$ is in the bad suffix of $\Parent{P}{x}{\Bad}$, \cref{cor:del-is-del} yields that $\ell_j = L - \Delta_L(B)$.
Since $\ell_j$ is never negative, we have $\Delta_{L}(B) \le L$ as required.

We proceed to prove the second claim by induction on $B$.

\paragraph{Base case.}
 Let $P=S[a..b]$ and let the bad interval of $L(P)$ be $S[a'..c]$.
For $B=0$, we have that $\Delta_L(0)= \delta_{\ell}(P)$.
By definition, $\ell_{a'} = a'-a$.
Notice that $j'= sJ(a')$ must be the first index of the bad parent of $P$, so $\ell_{j'}=0$.
By \cref{cor:del-is-del}, we have $\delta_{\ell}(P) = a'- a$.

Recall that $L=i_0-a$, and observe that $i_0$ is in the bad suffix of $P$ if and only if $a'-a \le i-a$ which is equivalent to $\Delta_{\ell} \le L$, as required.

\paragraph{Induction step.}
Assume that the claim is true for $B \ge 1$.
Assume that $\Delta_{L}(B+1) \le L$ and $\Delta_R(B) \le R-k$.
We have that $\Delta_{L}(B) = \Delta_L(B+1) - \delta_{\ell}(\Parent{P}{B}{B})$ and $\Delta_R(B) =  \Delta_R(B) - \delta_{r}(\Parent{P}{B-1}{\Bad})$.
Since $\delta_\ell$ and $\delta_r$ value of a phrase are never negative (\cref{lem:delt-l-non-neg}), we have that $\Delta_{L}(B)\le L$ and $\Delta_R(B-1) \le R-k$.
Therefore, from the induction hypothesis we have that $i_0$ is $B$-bad.
We will show that $j = sJ^{B+1}(i_0)$ is in the bad suffix of $\Parent{P}{B+1}{\Bad}$. 
 
It follows from \cref{cor:del-is-del} that $\ell_j = L - \Delta_L(B)$ and that $r_j = R - \Delta_R(B)$.
Since $r_j = R - \Delta_R(B) \ge k$, we have that $sJ(j)$ is well defined (i.e., is in $j\in L(P)$) due to \cref{col:x-star-exists}.
Since $j$ is the result of $B+1$ jumps from bad suffixes, we have that $j$ is in $\Parent{P}{B+1}{\Bad}= S[a'..b']$.
So $\ell_j = j- a'$.
Recall that $\delta_{\ell}(\Parent{P}{B+1}{\Bad})= \hat a - a'$ such that $\hat a$ is the leftmost index in the bad suffix of $\Parent{P}{B+1}{\Bad}$.

We have $j = a' + \ell_j = a' +L - \Delta_L(B) = a' + L - \Delta_L(B+1) +\delta_{\ell}(P^{(B+1)}) = \hat a + L - \Delta_L(B+1)$.
Since $\Delta_L(B+1) \le L$, this implies $j \ge \hat a$ which implies that $j$ is in the bad suffix of $P^{(B+1)}$, as required.
\end{proof}

We proceed to use \cref{lem:connecting_everything} to characterize $x^*$.
Recall that $i_0$ is the input index with $r_{i_0}=R \in [k..1.5k)$ for $k\in \K$
We also denote $L=\ell_{i_0}$.
Let $P$ be the phrase containing $i_0$.
For an integer $D$, we denote $\Delta_R(D)=\sum_{x=0}^{D}\delta_r(\Parent{P}{x}{\Bad})$ and $\Delta_L(D)=\sum_{x=0}^{D}\delta_r(\Parent{P}{x}{\Bad})$.
In words, $\Delta_R(D)$ (resp. $\delta_L(D)$) is the sum of $\delta_r$ values (resp. $\delta_\ell$ values) on the path of length $D$ from $P$ towards the root of $T_{\Bad}$.

We denote as $D_R$ the maximal integer such that $\Delta_R(D_R) \le R-k$.
Similarly, we denote $D_L$ as the maximal integer satisfying $\Delta_L(D_L) \le L $.
Notice that $\Delta_L(-1) = \Delta_R(-1) = 0$, and that $L \ge 0$ and $R \ge k$, so $D_L$ and $D_R$ are well defined.

We prove the following fact.
\begin{lemma}\label{lem:dl-dr-is-what-we-need}
Let $D = \max(0, \min (D_L,D_R+1))$.
We have that for every $x\in [0..D]$, the $r$-values of $sJ^x(i_0)$ is at least $k$, and for $j = sJ^{D+1}(i_0)$ we have $r_j<k$.  
\end{lemma}
\begin{proof}
We start by considering the case in which either $D_L= - 1$.
Notice that if this occurs, we have $D=0$ and the claim for $x\in [0..D] = \{0 \}$ holds since $sJ^0(i_0) = i_0$ and $r_{i_0} \ge k$.
If $D_L = -1$, we have that $\Delta_L(0) = \delta_\ell(P) >L$.
According to \cref{lem:connecting_everything}, it means that $i_0$ is not $(0)$-bad.
According to \cref{lem:stable-bad-jumps-same-phrase} we have that $j = sJ^{D+1}(i_0) =sJ(i_0)$ has $r_j < k$ as required.

From now on, we assume that $D_L \ge -1$ which implies that $D= \min(D_L, D_R+1)$.

From the definition of $D$, we have that $\sum_{x=0}^{D}\delta_{\ell}(\Parent{P}{x}{\Bad}) \le L$ and  $\sum_{x=0}^{D-1}\delta_{r}(\Parent{P}{x}{\Bad}) \le R-k$.
According to \cref{lem:connecting_everything} we have that $i_0$ is $(D)$-bad, and therefore $j'=sJ^D(i_0)$ is in the phrase $\Parent{P}{D}{B}$.

It follows from \cref{cor:del-is-del} and the fact that $i_0$ is $(D)$-bad that $r_{j'} = R - \Delta_R(D-1) \ge  k$, so from the monotonicity of $r$-values when applying $sJ$ we have that for every $x\in [0..D]$ it holds that $sJ^x(i_0)$ has $r$-values at least $k$.

Since $r_{j'} \ge k$, it follows from \cref{col:x-star-exists} that $j = sJ(j') = sJ^{D+1}(i_0)$ is well-defined.

We consider two cases.
If $j'$ is not in the bad suffix of $\Parent{P}{D}{\Bad}$, and $r_{j'} \in [k..1.5k)$, we have $r_j <k$ by \cref{lem:stable-bad-jumps-same-phrase}.

Let us now consider the case where $j'$ is in the bad suffix of $\Parent{P}{D}{\Bad}$.
We have that $i_0$ is $(D+1)$-bad, and in particular $\Parent{P}{D}{\Bad}$ has a parent $\Parent{P}{D+1}{\Bad}$ in $T_{\Bad}$.
According to \cref{lem:connecting_everything} this implies that $\Delta_{L}(D+1) \le L$, which means that $D_L > D$.

Therefore, we have $D= D_R + 1$ and in particular $D > D_R$. 
It follows that $\Delta_R(D) > R-k$.
Since all the jumps from $i_0$ to $j$ were from the bad suffix of the block containing them, \cref{cor:del-is-del} suggests that $r_{j} = R - \Delta_R(D) < R - (R-k) = k$.
We have shown $r_j < k$, as required.
\end{proof}

Having established the connection between $D_L$,$D_R$, and the sought value $x^*$, we are now interested in constructing a data structure for efficiently computing $D_L$ and $D_R$.
We prove the following.
\begin{lemma}\label{lem:algorithm-dl-dr}
    Given $T_{\Bad}$, we can construct in $O(z)$ time a data structure taking $O(z)$ space such that given $i_0$, we can retrieve $D_L$ and $D_R$ in $O(\log \frac{n}{z})$ time.
    Within the same query time, the data structure also returns the corresponding ancestors $\Parent{P}{x}{\Bad}$ for the phrase containing $P$ and $x\in \{D_L,D_R\}$.
\end{lemma}
\begin{proof}
We will describe a data structure with $O(\log n)$ query time.
We will later describe how the running time can be improved to $O(\log \frac{n}{z})$.

We build two level-ancestor data structure over the same forest $T_{\Bad}$.
In the first data structures, every phrase $P$ is assigned $\Delta'_L(P)$ which is the sum of $\delta_\ell$ values of phrases on the path from $P$ to the root (excluding $P$).
Formally, $\Delta'_L(P) = \sum_{i=1}^{d}\delta_\ell(\Parent{P}{i}{B})$ where $d$ is the depth of $P$ in $T_{\Bad}$.

Similarly, we build another level ancestor data structure in which each phrase $P$ is associated with the sum of $\delta_r$ values on the path to the root (excluding $P$).
Formally, $\Delta'_R(P) = \sum_{i=1}^{d}\delta_r(\Parent{P}{i}{\Bad})$.

Given $i_0$ with $\ell_{i_0} = L$ and $r_{i_0} = R$ in phrase $P$,  we query the level-ancestor data structure for the lowest ancestor $A_L$ of $P$ with $\Delta'_L(A_L) \ge \delta_\ell(P) +  \Delta'_L(P) - L$.
Similarly, we query for the lowest ancestor $A_R$ with $\Delta'_R(A_P) \ge \delta_r(P)+ \Delta'_R(P) - (R-k)$.
A more natural interpretation of $A_R$ (and $A_L$) is the lowest ancestor of $P$ such that the path from $P$ to $A_P$ (including both endpoints) has total $\delta_r$ values at least $R-k$ (resp. the path from $P$ to $A_L$ has total $\delta_{\ell}$ values at least $L$).
Clearly, the depth difference between $P$ and $A_R$ is exactly $D_R$, and the depth difference between $P$ and $A_L$ is exactly $D_L$ (and if either of $A_L,A_R$ two does not exist, the corresponding $D_L$ or $D_R$ value is $-1$).

We have found each of $D_L$ and $D_R$ and the corresponding ancestors using a single weighted ancestor query, so the query running time is $O(\log n)$.

The construction time consists of finding all $\Delta'_L(P)$ and $\Delta'_R(P)$ values and constructing two level-ancestor data structures for $T_{\Bad}$ equipped with these values.
We can compute $\delta_{\ell}(P)$ and $\delta_r(P)$ for every $P$ in $O(1)$ straightforwardly by taking index $b$ such that $L(P)=S[a..b]$, compute with $b' = sJ(b)$ and subtract the $\ell$ and $r$ values of $b$ and $b'$ (to find the $r$ and $\ell$ values of $b'$, we need access to the phrase containing $b'$ which we have as it is the parent of $P$ in $T_{\Bad}$).

Computing $\Delta'_L(P)$ and $\Delta'_R(P)$ for all $P$ can be implemented in $O(|T_{\Bad}|) = O(z)$ time by an iteration on $T_{\Bad}$, keeping track of cost of root-to-node path at every step.
The total construction time is therefore $O(z)$, as required.

\paragraph{Improving the query time to $O(\log \frac{n}{z})$.}
In the above approach, the query boils down to performing weighted ancestor queries.
We notice that the weighted ancestor queries required by the algorithm have the following property:
When we query for the weighted ancestor of $P$ with weight $X$, the weight of $P$ is at most $X+\frac{n}{z}$.

Let us justify this claim.
The algorithm queries for the first ancestor of $P$ such that the path from $P$ to $P'$ has weight at most $R-k$ or $L$, where $L$ and $R$ are $\ell$ and $r$ values of the query index $i_0$.
Due to \cref{lem:no-big-phrases}, we can assume that the phrase containing $i_0$ has length at most $\frac{n}{z}$.
It follows that $L,R \le \frac{n}{z}$ and the property we stated holds.

Assume that the weights assigned to the tree $T_{\Bad}$ never result in a child having the same weight as its parent.
Under this assumption, (and noticing that all weights are integers) we have that each child has weight at least 1 more than its parent.
It follows that the answer for our query is at unweighted hight at most $\frac{n}{z}$ above $P$.
We can find this answer by preprocessing $T_{\Bad}$ for constant time unweighted level ancestor queries~\cite{DBLP:journals/tcs/BenderF04}, and binary searching the range $[0..\frac{n}{z}]$ for the first parent with weight below a certain threshold in $O(\log \frac{n}{z})$ time. 

We can enforce the assumption that a child never has the same weight as its parent by contracting edges between same weight parent-child paris.
\end{proof}

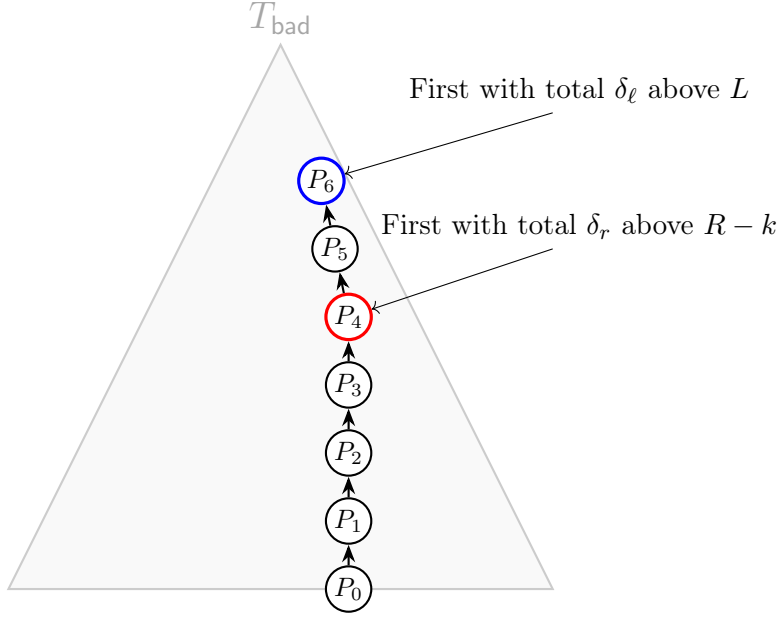
\begin{figure}
\centering
\begin{tikzpicture}[
    scale=0.9,
    tree node/.style={circle, draw=black!60, fill=white, minimum size=6mm, inner sep=0pt, font=\small},
    active node/.style={circle, draw=black, fill=white, thick, minimum size=6mm, inner sep=0pt, font=\small},
    marked layer/.style={black, dashed, thick},
    shortcut layer/.style={ dotted, thick},
    step line/.style={black, thick, ->, >=Stealth},
    jump line/.style={black, thick, ->, >=Stealth, dashed}
]

    \draw[fill=gray!5, draw=gray!40, thick] (0,9) -- (-4,1) -- (4,1) -- cycle;
    \node[gray!70] at (0,9.4) {\Large $T_{\Bad}$};


    \node[active node] (v0) at (1, 1) {$P_0$};
    
    \node[active node] (v1) at (1, 2) {$P_1$};
    \node[active node] (v2) at (1, 3) {$P_2$}; 
    
    \node[active node] (v3) at (1, 4) {$P_3$}; 
    
    \node[active node,draw=red,very thick] (v4) at (1, 5) {$P_4$};
    \draw[->] (4,6)--(v4) node[at start,above,xshift=10]{First with total $\delta_r$ above $R-k$};

    \node[active node] (v5) at (0.8, 6) {$P_5$};
    \node[active node,draw=blue,very thick] (v_target) at (0.6, 7.0) {$P_6$};
    \draw[->] (4,8)--(v_target) node[at start,above,xshift=10]{First with total $\delta_{\ell}$ above $L$};

    \draw[step line] (v0) -- (v1);
    \draw[step line] (v1) -- (v2);
    
    \draw[step line] (v2) -- (v3);
    
    \draw[step line] (v3) -- (v4);

    \draw[step line] (v4) -- (v5);
    \draw[step line] (v5) -- (v_target);

\end{tikzpicture}
\caption{A demonstration of the stable part of the epoch (\cref{lem:stable}).
The algorithm uses ancestor queries to find the lowest ancestor of $P_0$ such that the sum of $\delta_{ell}$ values on the path $P_0$ to that ancestor is above $L$.
Similarly, the algorithm finds the lowest ancestor such that the sub om $\delta_r$ values is at least $k-R$.
The algorithm specifically takes $D$ the lowest of these two located ancestors.
According to \cref{lem:dl-dr-is-what-we-need}, the sequence of $sJ$ steps starting at $i_0$ traverse through the phrases in the path from $P_0$ to $D$, and the following step leads to an index $j$ with $r$-value below $k$.
The auxiliary information stored in $D$ and in $P_0$ allows us to find $j$ in constant time.}\label{fig:stable}
\end{figure}
We are finally ready to prove \cref{lem:stable}
\begin{proof}[Proof of \cref{lem:stable} (See \cref{fig:stable})]
We construct $T_{\Bad}$ in $O(z \log^2 \frac{n}{z})$.
We can find the bad parent of each $P\in \LZend$ in $O(\log \frac{n}{z})$ time by first constructing the data structures of \cref{lem:ds-sji} and \cref{lem:find-con-phrase} in $O(z \log^2 \frac{n}{z})$.
Then, for every $P=S[a..b]$, the bad parent of $P$ is the phrase containing $sJ(a+\floor{\frac{2}{3}|P|})$, which can be found in $O(\log\frac{n}{z})$ time.

Given $T_{\Bad}$, we construct the data structures of \cref{lem:algorithm-jump-bad-quickly} and of \cref{lem:algorithm-dl-dr}.

Upon query for $i_0$, we use \cref{lem:find-con-phrase} to find the phrase $P$ containing $i_0$ in $O(\log \frac{n}{z})$ time.
Knowing $P$, we can compute $L=\ell_{i_0}$, $R=r_{i_0}$ and the integer $k\in \K$ such that $r_{i_0} \in [k..1.5k)$.

We use the data structure of \cref{lem:algorithm-dl-dr} to find $D= \max (0, \min (D_L, D_R+1))$ in $O(\log \frac{n}{z})$ time, alongside the corresponding ancestors.
According to \cref{lem:dl-dr-is-what-we-need}, we have $r_{sJ^x(i_0)}\ge k$ for every $x\in[0..D]$ and $j= sJ^{D+1}$ has $r_j <k$.
If $D=0$, we simply return $j = sJ(i_0)$ using \cref{lem:ds-sji}.
Otherwise, we have $D = \min(D_L, D_R+1)$ and therefore $\Delta_L(D) \le L$ and $\Delta_R(D-1)\le R-k$.
It follows from \cref{lem:connecting_everything} that $i_0$ is $D$-bad.
We can therefore use the data structure of \cref{lem:algorithm-jump-bad-quickly} to obtain $j=sJ^{D+1}(i_0)$ in $O(1)$ time, and return it as a valid output.

The construction time and the space of the data structure consist of the construction and space complexities of the data structures of \cref{lem:ds-sji,lem:find-con-phrase,lem:algorithm-dl-dr,lem:algorithm-jump-bad-quickly}, which are all bounded by $O(z)$ space and $O(z \log^2 \frac{n}{z})$ time.
The query time is dominated by one query to each of the data structures of \cref{lem:algorithm-dl-dr,lem:algorithm-jump-bad-quickly}, which is dominated by $O(\log \frac{n}{z})$, as required.

\end{proof}

\section{Batched pre Computation}\label{sec:computepre}
In this section, we prove \cref{lem:computepre}.

We use the following data structure that has been designed (in a slightly more restricted variant) by Farach and Thorup~\cite{DBLP:journals/algorithmica/FarachT98}. For completeness, we provide a complete proof in \Cref{sec:shift-merge}.

\begin{restatable}{lemma}{shiftstructure}
\label{lem:split-merge-shift}
    There is a data structure storing a set $X$ of balanced search trees where each tree in $X$ has elements in $[1..n]$.
    The data structure supports the following updates:
    \begin{enumerate}
        \item Initialize $X$ as a single empty balanced search tree.
        \item Apply an insertion/deletion operation to a tree in $X$.
        \item $\Split(T,a)$ : Remove $T_1$ from $X$, and add $T_1$ over the set $T\cap [1..a]$ and $T_2$ over the set $T\cap [a+1..n]$ to $X$.
        \item $\Merge(T_1,T_2)$: Remove $T_1$ and $T_2$ from $X$, and add a tree $T$ to $X$ over the set $T_1 \cup T_2$.
        \item $\Shift(T,\delta)$: Add $\delta$ to every item in the tree $T$. 
    \end{enumerate}
    The initialization takes $O(1)$ time.
    Each update is implemented in amortized time  $O(\log n  \log N)$ where $N$ is an upper bound on the size of each tree in $X$.
\end{restatable}

\begin{figure}[htbp]
    \centering
    \tikzset{every picture/.style={line width=0.75pt}} 
\usetikzlibrary{patterns}

\begin{tikzpicture}[x=1pt,y=1pt,yscale=-1,xscale=1]
\pgfmathsetmacro{\blockheight}{25}
\pgfmathsetmacro{\stringheight}{20}
\pgfmathsetmacro{\height}{20}

\pgfmathsetmacro{\width}{10}
\pgfmathsetmacro{\x}{0} 

\draw (0,0) rectangle (320,\blockheight);
\draw(250,5) rectangle node{$P_t$} (320,\blockheight-5);

\draw(70,5) rectangle node{$P_{t'}$} (140,\blockheight-5);

\draw[dashed] (320,-15) -- (320,15+\blockheight);

\draw[{Circle}-{|}] (10,-5) -- (30,-5);

\draw[{Circle}-{|}] (35,-5) -- (60,-5);

\draw[{Circle}-{|}] (90,-5) -- (130,-5);

\draw[{Circle}-{|}] (140,-5) -- (165,-5);

\draw[{Circle}-{|}] (190,-5) -- (220,-5);

\draw[{Circle}-{|}] (235,-5) -- (280,-5);

\draw[{Circle}-{|}] (290,-5) -- (310,-5);

\draw[{Circle}-{|}] (35,-10) -- (60,-10);

\draw[{Circle}-{|}] (70,-10) -- (110,-10);

\draw[{Circle}-{|}] (150,-10) -- (200,-10);

\draw[{Circle}-{|}] (255,-10) -- (270,-10);

\end{tikzpicture}
    \label{fig:top_tikz}
    
    \vspace{0.8cm} 
    \tikzset{every picture/.style={line width=0.75pt}} 
\usetikzlibrary{patterns}

\begin{tikzpicture}[x=1pt,y=1pt,yscale=-1,xscale=1]
\pgfmathsetmacro{\blockheight}{25}
\pgfmathsetmacro{\stringheight}{20}
\pgfmathsetmacro{\height}{20}

\pgfmathsetmacro{\width}{10}
\pgfmathsetmacro{\x}{0} 

\draw (0,0) rectangle (250,\blockheight);
\draw(250,50) rectangle node{$P_t$} (320,40+\blockheight);

\draw (0,0) rectangle (250,\blockheight);
\draw[dotted] (70,50) rectangle (140,40+\blockheight);

\draw[{Circle}-{|},dotted] (75,40) -- (90,40);
\draw[{Circle}-{|},dotted] (110,40) -- (130,40);
\draw[->](285,40+\blockheight) to[out=90,in=90,looseness=0.4] (105,40+\blockheight);

\draw(70,5)[dotted] rectangle node{$P_{t'}$} (140,\blockheight-5);

\draw[dashed] (250,-15) -- (250,15+\blockheight);

\draw[{Circle}-{|}] (10,-5) -- (30,-5);

\draw[{Circle}-{|}] (35,-5) -- (60,-5);

\draw[{Circle}-{|}] (90,-5) -- (130,-5);

\draw[{Circle}-{|}] (140,-5) -- (165,-5);

\draw[{Circle}-{|}] (190,-5) -- (220,-5);

\draw[{Circle}-{|}] (235,-5) -- (280,-5);

\draw[{Circle}-{|}] (290,40) -- (310,40);

\draw[{Circle}-{|}] (35,-10) -- (60,-10);

\draw[{Circle}-{|}] (70,-10) -- (110,-10);

\draw[{Circle}-{|}] (150,-10) -- (200,-10);

\draw[{Circle}-{|}] (255,40) -- (270,40);

\end{tikzpicture}
    \caption{An illustration of the algorithm of \cref{lem:computepre}.
    Every vertex $v$ in the tree is represented as an interval with a dot end representing $\val(v)$ and a line end representing $E(v)$.
    At iteration $t$, we take all the intervals with dots contained in $P_t$ and shift them to the source of $P_t$ (displayed as $P_{t'}$).
    Then, every remaining vertex with $E$ value that exceeds into $P_t$ is removed and reported (in the figure, there is only one such vertex).
    }
    \label{fig:pre-batch}
\end{figure}
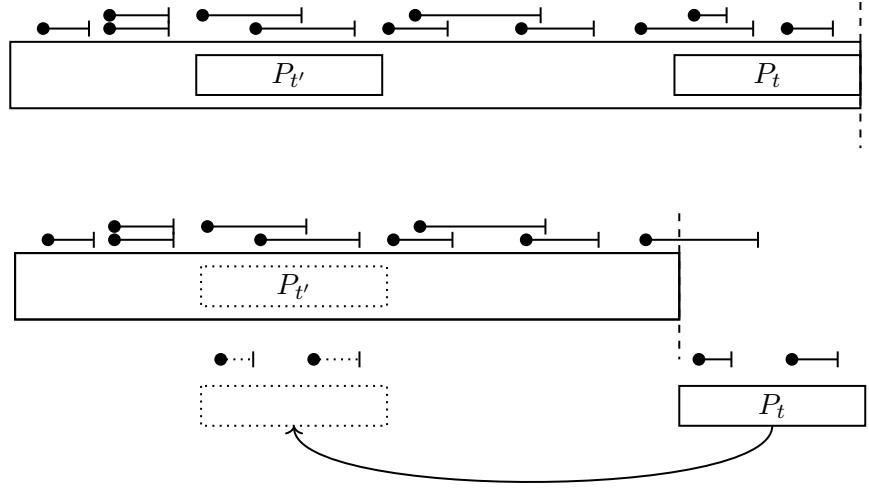

We first present an algorithm with $O(z\log^2n)$ running time.
We will later show how to replace the $\log n$ factors with $\log \frac{n}{z}$.
Denote $I = (b_1,e_1),(b_2,e_2),\ldots , (b_{|I|},e_{|I|})$.
Denote $B = \{ b_i \mid i \in [|I|] \}$.
Let $Z = \{ a_i \mid i\in [z] \}$ be the starting indices of the phrases of the Lempel-Ziv End partition $\LZend$ of $S$.
We define a node for every query in $I$ as follows.
For every $(b_i,e_i)\in I$, we define a node $v_i$ with initial value $\val(v_i) =b_i$.
The node $v_i$ contains as auxiliary information $E(v_i)= e_i$ and $I(v_i)=i$.
We initialize a tree data of \cref{lem:split-merge-shift} and add to $X$ a balanced tree $T$ containing all the vertices $\{v_i \mid i\in [|I|]\}$. 
The tree is sorted in increasing order of $\val(v)$.

For every node $v\in T$, we store as auxiliary information the value $M(v)$ which is the maximal $E(v)$ value of a node in the subtree rooted at $v$. 

We maintain the auxiliary information as the tree is manipulated by the operations of the data structure.
Throughout the algorithm, when a node $v$ is in a tree $T'$ that is applied a $\Shift(T,\delta)$ operation, we consider $E(v)$ to be shifted by the same amount.
However, $I(v)$ does not change throughout the running time of the algorithm.

\paragraph{The algorithm (See \cref{fig:pre-batch}) }The algorithm initializes an array $A$ of size $|I|$ intended to store the answers to each query.
Initially, $A[i] = \nil$ for every $i\in [|I|]$.
We will show that when the algorithm terminates, $A[i]=\pre(b_i,e_i)$.
Throughout the algorithm, we will remove nodes from $T$.
When we remove a node $v$ from $T$, we assign $A[I(v)] = \val(v)$.

The algorithm consists of $z$ iterations indexed in decreasing order from $z$ to $1$.
We maintain the invariant that at the start of every iteration, $X$ contains a single balanced search tree.
Initially, $T= T_z$ is the tree described above.
For every $t\in [1..z]$, in decreasing order, the algorithm applies the following.
\begin{enumerate}
    \item Let $T$ be the single tree in $X$ at the start of iteration $t$. Let $P_t = S[a_t..b_t]$ be the $t$-th phrase in $\LZend$ with source $S[a'_t..b'_t]$, and let $\delta_t = b_t - b'_t$. 
    If $P_t$ does not have a source, let $\delta_t = 0$.
    \item If $\delta_t \neq 0$:
    \begin{enumerate}
        \item \textbf{Split:} split $T$ into two trees $T^1= T \cap [1..a_t-1]$ and $T^2= T \cap [a_t..b_t]$ using a $\Split(T,a_{t-1})$ operation.
        \item \textbf{Shift:} Apply $ \Shift(T^2, -\delta_t)$ to obtain $T^2_\delta$.
        \item \textbf{Merge:} merge $T^1$ and $T^2_\delta$ into $T$ using $\Merge(T^1,T^2_{\delta})$  
    \end{enumerate}
    \item \textbf{Deletion:} as long as there is a node $v\in T$ with $E(v)\ge a_t$, find $v$ and remove it from $T$. 
\end{enumerate}

After applying the subroutine for $t=1$, output $A$.

In each iteration, the algorithm applies $O(1)$ operations of type $\Split$, $\Shift$, and $\Merge$ to the data structure.
Additionally, the subroutine deletes nodes, with each node deleted at most once throughout the running time of the algorithm.
Navigating to nodes with maximal $E$ value can be done in $O(\log n)$ using the auxiliary $M$ information. 
Since we apply $|I|$ deletions and $z$ operations of types $\Split,\Shift,\Merge$, the total running time is $O((z+|I|)\log^2n)$.

We proceed to show that the algorithm correctly computes all $\pre$ values.
For every $t\in [z]$, denote as $T_t$ the single tree $T$ stored in $X$ at the start of the $t$-th iteration of the algorithm (where the $z$-indexed operation is the first one and the $1$-indexed iteration is the last one)

We start by proving the following invariant.
\begin{lemma}\label{lem:compute-pre-all-below-b}
    Let $P_t = S[a_t..b_t]$ be the $t$-th phrase in $\LZend$.
    At the start of the $t$'th operation, every node in $T_t$ has $1\le \val(v) \le E(v)\le b_t$. 
\end{lemma}
\begin{proof}
    The claim is correct for $t=z$, as $b_z = n$ and in the initial tree $T_z$ every node $v$ has $[\val(v)..E(v)] \in I$.
    Since $I$ contains only intervals in $[1..n]$, the claim follows.
    For $t<z$, the $(t+1)$'th iteration concludes by deleting all vertices $v$ with $E(v) \ge a_{t+1} = b_{t}+1$.
    It is therefore clearly holds that all vertices at the start of the $t$'th iteration have $E(v)\le b_t$.

    It remains to show that we never apply a shift that results in a node $v'$ with $\val(v')<1$.
    Consider a node $v$ that is shifted in iteration $t$, and denote as $v'$ the resulting node.
    As we have shown, it holds that $E(v) \le b_t$.
    It follows form the fact that $v$ was shifted that it is placed in $T^2$ when the $t$-th iteration splits $T$ into $T^1$ and $T^2$.
    This can only happen if $\val(v) \ge a_t$.
    In conclusion, we have shown $a_t \le \val(v) \le E(v) \le b_t$.
    Then, $\val(v)$ and $E(v)$ are reduced by $\delta_t = b_t - b'_t$ where $S[a'_t..b'_t]$ is the source of $P=S[a_t..b_t]$.
    the vertex $v'$ is the version of $v$ after the shift.
    It follows from the above that $\val(v') = \val(v) - \delta_t\ge a_t - (b_t - b'_t) = b'_t - (b_t-a_t)$.
    Since $S[b'_t - (b_t-a_t) ..b'_t]$ is an occurrence of $P_t$ in $S$, we have that $val(v') \ge b'_t - (b_t-a_t) \ge 1$ as it is a proper index in $S$. 
\end{proof}

We now prove the following invariant, intuitively stating that when we change the value of a node throughout the algorithm, the newly created node represents an equivalent query.

\begin{lemma}\label{lem:compute-pre-equivelent-query}
    Let $v$ be a vertex in $T$ at any time throughout the algorithm.
    It holds that $\pre(b_{I(v)},e_{I(v)}) = \pre(\val(v),E(v))$.
\end{lemma}
\begin{proof} 
    At the start of the algorithm, the claim immediately follows, as in the initial tree for every $i\in [|I|]$ we have $B(v_i)  =b_i$ and $E(v_i) = e_i$.

    Consider an iteration $t \in[2..z]$ in which the $\val$ and $E$ values of a vertex $v$ with $I(v) = i$ change.
    Let $v'$ be the modified $v$ with the new $\val$ and $E$ values.
    Denote $b = var(v)$, $e=E(b)$, $b'= var(v')$ and $e'= E(v')$.
    Let $P_t = S[a_t..b_t]$ be the $t$'th phrase in $\LZend$ and let $S[a'_t..b'_t]$ be the source of $P_t$.
    Notice that all vertices that have their value change in the $t$-th iteration have both their $\val$ and $E$ decreased by $\delta_t = b_t - b'_t$.
    So in particular, we have $b' = b -\delta_t$ and $e' = e-\delta_t$.
    
    Due to \cref{lem:compute-pre-all-below-b}, we have that $e \le b_t$.
    Since the value of $v$ was changed by the subroutine, it must be the case that $v$ is in the tree $T^2$ created in the $t$'th iteration, which contains only vertices with $\val$ value at least $a_t$.
    In conclusion, we have shown that $a_t \le b \le e\le b_t$.
    This means that $b$ and $e$ are both contained in the same phrase $P_t$.
    Therefore, by the definition of $\pre$, we have $\pre(b,e) = \pre(b - \delta_t,e-\delta_t) = \pre(b',e')$.
    By the induction hypothesis, we have $\pre(b,e) = \pre(b_i,e_i)$.
    It follows that $\pre(b',e') = \pre(b_i,e_i)$, as required.
\end{proof}

The last invariant we need is that when the algorithm removes a node, the value assigned in $A$ is the correct $\pre$ value for the query represented by this node.

\begin{lemma}\label{lem:compute-pre-answer-query}
    When a node $v$ is removed from $T$, it holds that $\pre(\val(v),E(v))= \val(v)$.
\end{lemma}
\begin{proof}
    Let $t$ be the iteration in which $v$ was removed.
    Let $P_t= S[a_t..b_t]$ be the $t$'th phrase in $\LZend$.
    If $P_t$ does not have a source, we have $P=[a_t]$.
    In this case, the fact that $v$ is removed implies $E(v) \ge (a_t)$.
    From \cref{lem:compute-pre-all-below-b} we have that $E(v)<b_t=a_t$, which implies $E(v) = a_t$.
    We also have from \cref{lem:compute-pre-all-below-b} that $\val(v) \le E(v)$.

    If $\val(v) = E(v)$, then we have that $[\val(v)..E(v)] = [a_t..a_t]$ is completely contained in the phrase $P_t$ which has no source.
    
    Otherwise, we have $\val(v) < E(v) =a_t$, which means that $\val(v)$ is not in the phrase $P_t$, while $E(v)$ is in the phrase $P_t$.  
    It follows that $\val(v)$ and $E(v)$ are in different phrases.
    
    We have shown that in both cases, $\val(v)$ and $E(v)$ are not contained in the same phrase which has a source. 
    Therefore, $\pre(\val(v),E(v))=\val(v)$ from the definition of $\pre$.
\end{proof}

Let us conclude the correctness of the algorithm.
In every iteration, we either delete vertices or change their value.
When the algorithm concludes, all vertices are deleted since the first phrase in $\LZend$ must be $S[1..1]$.
Therefore, at iteration $1$ we will remove all vertices $v$ with $\val(v) \ge 1$, which is satisfied for all remaining vertices by \cref{lem:compute-pre-all-below-b}.

When the vertex $v$ with $I(v)=i$ is deleted, we set $A[i] = \val(v)$.
By \cref{lem:compute-pre-answer-query}, it holds that $\val(v) = \pre(\val(v),E(v))$.
By \cref{lem:compute-pre-equivelent-query}, it holds that $\val(v) = \pre(b_i,e_i)$, as required.

\paragraph{Substituting $\log n$ factors with $\log \frac{n}{z}$.}
We now describe how to implement the above approach, replacing the $\log n$ factors with $\log \frac{n}{z}$ factors.
We uniformly partition the domain $[1..n]$ into $z$ intervals $U_i = (i \cdot \frac{n}{z}..(i+1)\frac{n}{z}]$, each of size $\frac{n}{z}$.
Instead of maintaining one tree containing all queries, we maintain a separate tree $T_i$ maintaining queries $(b,e)$ with $b\in U_i$.
The queries are sorted according to the primary order $b$, and the secondary order $e$.

When we implement a split, the split interval corresponds to some phrase $P_i = S[a_i..b_i]$.
Due to \cref{lem:no-big-phrases}, we have that $|P_i| \le \frac{n}{z}$, which means that $[a_i..b_i]$ spans at most two intervals $U_x$,$U_{x+1}$.
We partition the range $[a_i.,b_i]$ into at most two sub-ranges, and remove each of them from the corresponding $U_i$.
Similarly, when we merge the shifted version of $[a_i..b_i]$ back intro a tree, we first find in constant time the at most two intervals $U_x,U_{x+1}$ intersecting the (shifted) range, split $[a_i..b_i]$ accordingly and merge each piece into the corresponding $U_x$,$U_{x+1}$.
Finding the $U_x$ intersecting $[a_i..b_i]$ in both of these cases can be done in constant time.

We would maintain the invariant that every tree is of size $O((\frac{n}{z})^2)$, so each search tree operation costs $O(\log \frac{n}{z})$.
Since $T_i$ contains only queries $(b,e)$ with $b\in U_i$, the number of distinct $b$ values in $T_i$ is bounded by $\frac{n}{z}$.
Notice that the distance between $b$ and $e$ never changes throughout the lifetime of a node.
Since initially we have $e-b \le \frac{n}{z}$ for every query, we have that this inequality holds for every query throughout the running time of the algorithm.
Therefore, for every $b$ there are at most $\frac{n}{z}$ possible values of $e$, and the possible distinct $(b,e)$ pairs that can be in a tree $T_i$ is bounded by $(\frac{n}{z})^2$.
In $T_i$, every node with values $(b,e)$ stores, as auxiliary information, a linked list of all $i$ such that the query $q_i$ in the initial input currently corresponds to this node.
Initially, each node is attached with a single $i$, but we may combine nodes with identical $b$ and $e$ values throughout the running time of the algorithm, connecting their corresponding lists.

Whenever the size of a tree $T_i$ exceeds $2 (\frac{n}{z})^2$, we iterate $T_i$ and merge every duplication we find in $O(1)$ time.
Since $|T_i| \ge (\frac{n}{z})^2$ and we have less than $(\frac{n}{z})^2$ distinct values of $T_i$ that may belong in $T_i$, we are guaranteed to find at least $(\frac{n}{z})^2$ duplications in this process.
It follows that we find a duplication in amortized constant time.
After applying the above, the size of $T_i$ is again bounded by $(\frac{n}{z})^2$.

Since every merge increase the total number of queries across all trees by 1, we will have at most $z$ merges throughout the running time of the algorithm, each found and implemented in $O(1)$ time.

Another change that the partition into intervals introduces to the algorithm is in the tests for the maximal value $e$ across all queries in our data structures.
We would not like to check all $T_i$'s individually, and even storing a heap over all of $T_i$'s maximal elements is too costly (This will introduce an extra cost of $O(\log z)$).
Instead, recall that every $(b,e)$ pair has $e\in [b..b+\frac{n}{z}]$.
We also recall that at every point throughout the running time of the algorithm, we have some prefix of the domain still containing queries, so $U_1,U_2,\ldots U_x$ are active, and every $U_y$ with $y>x$ is empty.
Together, these two properties imply that the maximal $e$ values is either in $T_x$ or in $T_{x-1}$, so it is sufficient to query both $T_i$ and $T_{i-1}$ for their maximal value in $O(1)$ time.

\section{Data Structure for Finding a Bad Parent}\label{sec:kbad-parent-ds}
In this section, we prove \cref{lem:kbad-parent-ds}.
We will describe a data structure with running time $O(z (\log^2 \frac{n}{z} + (\log \log n)^2))$.
Then, we will discuss how to replace the $\log\log n$ factor with $\log\log \frac{n}{z}$.
In this section, $k$-bad parent corresponds to the definition presented in \cref{sec:naive} (following \cref{lem:bad-jumps-same-phrase}), i.e., the $k$-bad parent of $P$ is the parent of $P$ in $T_k$.

Our data structure consist of two main component.
The primary ingredient is a mechanism that allows us, given a query phrase $P$ and $k\in \K$, to retrieve a set of $O(1)$ candidates for being the $k$-bad parent of $P$ in $O(1)$ time.
The second, complementary component is a verification algorithm, allowing us to check if a given phrase is the $k$-bad parent of $P$ in constant time.

We start by presenting the verification algorithm.
The main part of the verification algorithm is the following characterization of the $k$-bad parent.

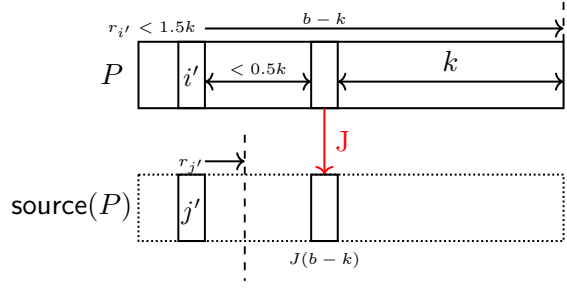
\begin{figure}
    \centering
    \tikzset{every picture/.style={line width=0.75pt}} 
\usetikzlibrary{patterns}

\begin{tikzpicture}[x=1pt,y=1pt,yscale=-1,xscale=1]
\pgfmathsetmacro{\blockheight}{25}
\pgfmathsetmacro{\stringheight}{20}
\pgfmathsetmacro{\height}{20}

\pgfmathsetmacro{\width}{10}
\pgfmathsetmacro{\x}{0} 

\draw (20,0) rectangle (180,\blockheight);

\draw[densely dotted](20,50) rectangle (180, \blockheight+50);

\draw(35,0) rectangle (45,\blockheight);
\node[anchor=north,font=\small] at(40,5) {$i'$};

\draw(85,0) rectangle (95,\blockheight);
\node[anchor=north,font=\tiny] at(90,-15) {$b-k$};

\draw[->,red] (90,\blockheight) -- (90, 50) node[midway,right]{J};

\draw(85,50) rectangle (95,50+\blockheight);
\node[anchor=north,font=\tiny] at(90,50+\blockheight) {$J(b-k)$};

\draw(35,50) rectangle (45,50+\blockheight);
\node[anchor=north,font=\small] at(40,55) {$j'$};

\node[anchor=north,font=\tiny] at(25,-12) {$r_{i'}<1.5k$};
\draw[->] (45,-5) -- (180,-5);

\draw[<->] (95,15) -- (180,15) node[midway,above]{$k$};

\draw[dashed, line width = 0.7] (60,35) -- (60,65+\blockheight);

\draw[dashed, line width = 0.7] (180 ,-15) -- (180 ,0);

\draw[<->] (45,15) -- (85,15);
\node[anchor=north,font=\tiny] at(65,4) {$< 0.5k$};

\node[anchor=north,font=\tiny] at(40,40) {$r_{j'}$};
\draw[->] (45,45) -- (60,45);

\node at (10,12) {$P$};

\node at (-5,62) {$\textsf{source}(P)$};

\end{tikzpicture}
    \caption{An illustration of the first part in the proof of \cref{lem:bad-k-parent-properties}.
    Assuming to the contrary that the phrase containing $J(b-k)$ is not bad, and that $i'<b-k$ is $k$-bad, we have that there is a phrase boundary separating $j'= J(i')$ and $J(b-k)$. 
    Therefore, the $r$-value of $i'$ is less than $k$, a contradiction.
    }\label{fig:FindBadParent1}
\end{figure}

\begin{lemma}\label{lem:bad-k-parent-properties}
    Let $P=S[a..b]$ be a phrase and let $k\in \K$.
    Let $P'=S[a'..b']$ be the $k$-bad parent of $P$ and let $i^* = \max(a,b-1.5k+1,a'+\delta_P)$.
    The following properties hold.
    \begin{enumerate}
        \item $P'$ contains the index $J(b-k)$.
        \item $r_{i^*} \ge k$ and $r_{J(i^*)} \ge k$.
    \end{enumerate}
\end{lemma}
\begin{proof}
    We start by proving the first statement (See \cref{fig:FindBadParent1}).
    Assume to the contrary that $P'=S[a'..b']$ exists and $J(b-k)\notin [a'..b']$.
    Notice that $J(b-k) = b-k -\delta_P$.
    Let $i'\in [a..b]$ be a bad jump with $r_{i'}\in [k..1.5k)$.
    In other words, $i'= b-\hat k$ for some $\hat k \in [k..1.5k)$.
    Denote $ j' = J(i')=b- \hat k -\delta_P$.
    Since $i'$ is a bad jump with $r_{i'} \in [k..1.5k)$, we have that $j'$ is in $P'$(\cref{lem:bad-jumps-same-phrase}).
    Since $ j'$ and $J(b-k)$ are not in the same phrase, there is a phrase boundary between them.
    This guarantees that $r_{j'} \le J(b-k) -  j' = b- k -\delta - (b'-\hat k-\delta) = \hat k - k < k$.
    A contradiction to $i'$ being a bad jump.

    We now prove that the second statement (See \cref{fig:FindBadParent2}).
    Notice that $i^* = \max(a,b-1.5k+1,a'+\delta_P )$ is exactly the leftmost index satisfying:
    \begin{enumerate}
        \item $i^* \in [a..b]$,
        \item $J(i^*) \in [a'..b']$, and
        \item $r_{i^*} \in [k..1.5k)$.
    \end{enumerate}
    Let $i'$ be a bad jump in $P$ with $r_{i'}\in [k..1.5k)$.
    We show that $i'$ also satisfies the above three conditions.
    By definition, $i'$ is in $[a..b]$ and $r_{i'} \in [k..1.5k]$.
    Since it is a bad jump, we also have $j' = J(i') \in [a'..b']$ due to \cref{lem:stable-bad-jumps-same-phrase}.

    This leads to $i^* \le i'$, and since they are in the same phrase, it leads to $r_{i^*} \ge r_{i'}\ge k$.
    Since $i'$ is a $k$-bad jump, we have that $j'= J(i')$ is in $P'$.
    From the definition of $i^*$, we have that $i^*$ is also in $P'$, and that $i^*\le i'$.
    Since $J(i^*) = j^*$ and $j'$ are also in the same phrase, we have $r_{j^*}\ge k$, as required.
\end{proof}

\begin{figure}
    \centering
    \tikzset{every picture/.style={line width=0.75pt}} 
\usetikzlibrary{patterns}

\begin{tikzpicture}[x=1pt,y=1pt,yscale=-1,xscale=1]
\pgfmathsetmacro{\blockheight}{25}
\pgfmathsetmacro{\stringheight}{20}
\pgfmathsetmacro{\height}{20}

\pgfmathsetmacro{\width}{10}
\pgfmathsetmacro{\x}{0} 

\draw (10,0) rectangle (180,\blockheight);

\draw(40,50) rectangle (160, \blockheight+50);

\node at (45,-10) {$a' + \delta_P$};
\draw(40,0) rectangle (50,\blockheight);

\draw[->,red] (45,\blockheight) -- (45, 50) node[midway,right]{$J$};

\draw(10,0) rectangle (20,\blockheight);
\node at (15,15) {$a$};

\draw(70,0) rectangle (80,\blockheight);

\draw(40,50) rectangle (50,50+\blockheight);
\node[anchor=north,font=\small] at(45,55) {$a'$};

\draw[<->] (80,-5) -- (180,-5)  node[midway,above]{$1.5k-1$};

\draw[dashed, line width = 0.7] (180 ,-15) -- (180 ,0);

\node at (0,12) {$P$};

\node at (30,62) {$P'$};

\end{tikzpicture}
    \caption{An illustration of the second part in the proof of \cref{lem:bad-k-parent-properties}.
   The three indices from the definition of $i^*$ are presented.
   Clearly, an index $i$ smaller than any of these three cannot be a $k$-bad index of $P$: If $i< a$, then $i$ is not in $P$. If $i<a'+\delta_P$, then $J(i)$ is not in $P'$ (which contains all $J$ values of $k$-bad indices). If $i< b-1.5k$ (and $i$ is in $P$) then $r_i>1.5k$.
   It follows that there are no $k$-bad indices to the left of $i^*$.
   Further, $i^*$ has $r$-values below $1.5k$, so if it is not $k$-bad it must be either due to its $r$-value being to small or due to the $r$-value of $J(i^*)$ being too small.
   Any other candidate for being $k$-bad to the right of $i^*$ would have even smaller such values, so if $i^*$ is not $k$-bad, no index is $k$-bad in $P$.
    }\label{fig:FindBadParent2}
\end{figure}
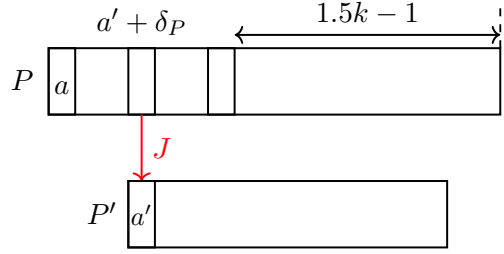

\cref{lem:bad-k-parent-properties} immediately implies an $O(1)$ verification algorithm, given below.
\begin{corollary}\label{cor:test-bad-parent}
    Given a phrase $P = S[a..b]$ in $\LZend$, $\delta_P$, another phrase $P'=S[a'..b']$ and some $k\in \K$, we can check in $O(1)$ time if $P'$ is the $k$-bad parent of $P$.
\end{corollary}
\begin{proof}
We straightforwardly check if $J(b-k) = b-k-\delta_P$ is in $[a'..b']$ and if $\max(a,b-1.5k+1,a'+\delta_P)$ is a bad jump from $P$ to $P'$.
\end{proof}

We proceed to show that a set of $O(1)$ candidates can be retrieved efficiently.
As a means to that end, we show how to use \cref{lem:bad-k-parent-properties} to find the $k$-bad parent of a given $P$ in $O(\log \frac{n}{z})$ time.
This is not fast enough for our final goal - which is supporting such queries in constant time.
However, it will be useful in the construction of the data structure.

\begin{lemma}\label{lem:find-kbad-slow}
There is a data structure that given a phrase $P$ and $k\in \K$, outputs the $k$-bad phrase of $P$ or reports that it does not exist in $O(\log \frac{n}{z})$ time.   
The data structure can be constructed in $O(z\log \frac{n}{z})$ time and uses $O(z)$ space. 
\end{lemma}
\begin{proof}
We construct the data structure of \cref{lem:find-con-phrase} in $O(n \log \frac{n}{z})$ time and $O(z)$ space.
Now, given a query phrase $P=[a..b]$ and $k\in \K$, we can find $J(b-k)$ (If $b-k$ is not in $P$, then $P$ clearly do not have a $k$-bad parent) in $O(1)$ time. 
We can then use the data structure of \cref{lem:find-con-phrase} to find the phrase $P'$ containing $J(b-k)$ in $O(\log \frac{n}{z})$ time.

Due to \cref{lem:bad-k-parent-properties}, we know that if $P$ has a $k$-bad parent, it must be $P'$.
We check if $P'$ is the bad parent of $P$ in $O(1)$ time using \cref{cor:test-bad-parent}.
\end{proof}

We are now ready to describe the main part of the data structure for finding the $k$-bad parent of a given $P$.
We define the notion of critical bad values.
\begin{definition}[Critical Bad Values]
    For a phrase $P$ in $\LZend$ and $k\in \K$ such that $|P| \ge 2k$. 
    Let $k'= 1.5k$ be the successor of $k$ in $\K$. 
    We say that $k$ is a critical bad value for $P$ if $P$ has a $k$-bad parent $P'$, and $P'$ is not the $k'$-bad parent of $P$.
\end{definition}

Intuitively, one can consider the process of fixing a phrase $P$ and increasing $k$ from $\min (\K)$ to $\max(\K)$.
The critical bad values of $P$ are the values of $k$ where the $k$-bad parent changes in this process (excluding very large $k$ values with $k>|P|/2$).
We prove several useful properties of the $k$-critical bad values.
We start by showing that for $k$-critical values, the $k$-bad parents of $P$ is of size roughly $k$. 

\begin{lemma}\label{lem:length-cricial-parent}
    Let $P=S[a..b]$ be a phrase in $\LZend$ with source $S[a'..b']$.
    Let $k\in \K$ be a critical value of $P$.
    Let $P_k$ be the $k$-bad parent of $P$.
    It holds that $|P_k| \in [k..1.5k)$
\end{lemma}
\begin{proof}
    It should be clear that $|P_k| \ge k$, as $P_k$ contains an index $j$ with $r_j \ge k$.
    Assume to the contrary that $|P_k|\ge 1.5k$.
    Since $P_k=S[a_k..b_k]$ is a $k$-bad parent, there is an index $i\in [a..b]$ with $j= J(i) \in [a_k..b_k]$ and $r_i,r_j \in[k..1.5k)$.
    Let us assume that $i$ is maximal among the indices satisfying these conditions.
    This means that either $r_i = k$ or $r_j=k$, as otherwise $i+1$ satisfies these conditions as well.
    Notice that in both cases, it must hold that $r_j=k$, since $r_j \in [k..r_i]$.

    Now consider the index $i'=i - 0.5k$ and let $k'=1.5k$.
    We will show that $i'$ has $r_{i'},r_{j'} \in [k'..1.5k')$, which implies that $P_k$ is the $k'$-bad parent of $P$.
    This is a contradiction to $k$ being a critical bad value.

    Notice that $r_{i'} = b- r_i - 0.5k \ge b- 2k$.
    Since $|P| \ge 2k$, we have that $i'$ is in $P$.
    We have that $r_{i'} =r_i + 0.5k \in [1.5k.. 2k) \subseteq [k'..1.5k')$.
    Additionally, since $i'\in P$ we have that $j'=j - 0.5k= b_k - r_j - 0.5k = b_k - 1.5k$.
    Since we assume $|P_k|\ge 1.5k$, it holds that $j'$ is in $P_k$ with $r_{j'} = 1.5k = k'$.
\end{proof}

Next, we show that the set of values $k$ such that a phrase $P'$ is the $k$-bad parent of $P$ form a consecutive interval.

\begin{lemma}\label{lem:bad-sandwitch}
    Let $P$ be a phrase and let $k_1<k_2<k_3\in \K$.
    If $P'$ is the $k_1$-bad parent of $P$ and the $k_3$-bad parent of $P$, then $P'$ is also the $k_2$-bad parent of $P$.
\end{lemma}
\begin{proof}
    We will prove that under the specified assumptions, $P'$ is also the $1.5k_1$ bad parent of $P$.
    This can be applied inductively until $k_2 = 1.5k_1 $ to obtain the claim.

    Let $P = S[a..b]$ and let $P'=S[\hat a..\hat b]$.
    Let the source of $P$ be $S[a'..b']$.
    For every index $i\in [a..b]$ such that $J(i) \in [\hat a..\hat b]$, the difference $r_{J(i)} - r_i$ is the same.
    Denote this difference as $\delta$.
    Since $P'$ is both the $k_1$ and the $k_3$-bad parent of $P$, we have from \cref{lem:bad-k-parent-properties} that $b' - k_1 , b'-k_3 \in [\hat a .. \hat b]$.
    Since $P'$ is the $k_1$-bad parent, there is an index $i_1\in [a..b]$ with $j_1= J(i_1) \in [\hat a..\hat b]$ such that $r_{i_1} ,r_{j_1} \in [k_1..1.5k_1)$.
    Similarly, since $P'$ is the $k_3$-bad parent we have some $i_3\in [a..b]$ and $j_3=J(i_3)\in [\hat a..\hat b]$  with $r_{i_3},r_{j_3}\in [k_3..1.5k_3]$
    In particular, since $r_{j_1} = r_{i_1} - \delta$, we must have $\delta < 0.5k_1< 0.5k_2$.
    Now, consider the index $i_2 = b-k_2-\delta$.
    
    We claim that $i_2 \in [i_3..i_1]$.
    It follows from $i_3 = b-r_{i_3} \le b- k_3 \le b - 1.5k_2 \le b - k_2 - \delta = i_2$, and from $i_1 = b-r_1 \ge b-1.5k_1 \ge b-k_2 \ge b-k_2 - \delta = i_2$.
    Therefore, $i_2 \in [a..b]$ is an index within $P$, and due to the same reasoning, we have that $j_2 = J(i_2) \in [\hat a..\hat b]$.
    Now notice that $r_{i_2} = k_2 + \delta \in [k_2..1.5k_2)$ and that $r_{j_2} = r_{i_2}-\delta = k_2$.
    We have shown that $i_2$ is an index in $P$ with $j_2$ in $P'$ such that $r_{i_2},r_{j_2}\in [k_2..1.5k_2]$, which means that $P'$ is the $k_2$ bad parent of $P$.
\end{proof}

We are ready to present our data structure for efficiently finding the $k$-bad parent.
For every phrase $P=S[a..b]\in\LZend$ we store the following information.
\begin{enumerate}
    \item $k(P) \subseteq \K$: the at most two unique values $k\in \K$ such that $|P| \in [k..2k)$
    \item $P'(P) \subseteq \LZend$: for every $k\in k(P)$, the $k$-bad phrase of $P$, if it exists.
    \item $C(P)$: a bit vector of size $\log_{1.5}|P| \in O(\log n)$. If $k$ is a critical bad value of $P$, then $C(P)[k] = 1$, otherwise $C(P)[k]=0$.
    \item $K(P)\in \K$: the unique power of $1.5$ satisfying $|P| \in [K(P)..1.5 K(P))$.
    \item $h(P) \in [1..n]$: the minimal integer multiple of $K(P)$ contained in $P$.
    Formally, $h_P = \min(\{ i \cdot K(P) \mid i\in \LZend \} \cap [a..b])$
\end{enumerate}

Additionally, we build the hash table $H$ in which for every $P$, we have $H(h(P)) = P$ ($P$ is represented as its index in $\LZend$).
The construction of $H$ requires $O(z(\log \log z)^{2})$ time with \cref{lem:hashing}.
The stored information for each phrase can be computed in $O(\log^2 \frac{n}{z})$ time for every phrase as follows.
The values $k(P)$, $K(P)$ and $h(P)$ can be obtained in $O(1)$ time via simple arithmetic operations.
$P'(P)$ can be obtained in $O(\log \frac{n}{k})$ time using the data structure of \cref{lem:find-kbad-slow}.
To construct $C(P)$, we iterate all values in $\K$.
For each such value, we check if it is a critical bad value directly according to the definition.
To do that, we need to be able to retrieve the $k$-bad parent of $P$ for a given $k\in \K$, which we do in $O(\log \frac{n}{k})$ time using \cref{lem:find-kbad-slow}.
The total running time for finding $C(P)$ is $O(|\K| \cdot \log \frac{n}{z})=O(\log^2\frac{n}{z})$. 

Let us prove that the above information is sufficient to retrieve the $k$-bad parent of $P$ in $O(1)$ time.

\begin{lemma}\label{lem:find-kbad-constant}
    Given a phrase $P=S[a..b]$ of $\LZend$ with its source $S[a'..b']$ and the values $k(P), P'(P),C(P)$.
    Given the hash $H$, we can find for a given $k$ the $k$-bad parent of $P$ (or report that it does not exist) in $O(1)$ time.
\end{lemma}
\begin{proof}
We describe an algorithm for finding the $k$-bad parent of $P$ or reporting that it does not exist.
The algorithm constructs a set of candidate phrases $\mathcal{P}$ for being the $k$-bad parent of $P$.
Then, the algorithm tests each candidate for being the $k$-bad parent of $P$ using \cref{cor:test-bad-parent}.

The set of candidates is constructed as follows.
First, we add the phrases in $P'(P)$ to $\mathcal{P}$.
Then, we find the smallest critical bad value $k'$ of $P$ that is at least $k$ using $C(P)$.
If there is indeed such value $k'$, we iterate each integer $x$ that is an integer multiple of $k'$ contained in $[b'-2.5k'..b'+0.5k']$.
For each such $x$, we add $P_x = H(x)$ to $\mathcal{P}$, if such $P_x$ it exists (i.e. if $x$ is a key in $H$).

This concludes the construction of $\mathcal{P}$.
Clearly, $|\mathcal{P}|\in O(1)$.
This follows from $|P'(P)| \le 2$ and from the fact that there are at most $4$ integer multiples of $k'$ in $[b'-2.5k'..b'+0.5k']$.
It should be clear that each element in $\mathcal{P}$ is retrieved in $O(1)$ time.
In order to obtain $k'$, we need to find the minimal bit that is at least $k$ and is equal to $1$ in $C(B)$.
This can be implemented by first computing the bitwise AND of $C(B)$ and its negation, which leaves
only the rightmost bit set to $1$. Then, we can lookup the position of the unique bit set to $1$ by
using a preprocessed perfect hash table that can be constructed in $O(\log n\cdot (\log\log n)^{2})=O(z\log(n/z))$ time
with \cref{lem:hashing}.
Then, we verify each of the candidates in $O(1)$ time using \cref{cor:test-bad-parent}.
It should be clear that the total running time is $O(1)$.

To prove the correctness of the algorithm, we claim that if there is a $k$-bad parent $P'$ for $P$, we indeed include $P$ in $\mathcal{P}$.
Assume that there is a $k$-bad parent $P'$.
Let $k'\ge k$ be the maximal value in $\K$ such that $P'$ is a $k'$-bad parent of $P$.
If $k'\in k(P)$, then $P'\in C(P)\subseteq \mathcal{P}$, as required.
If $k'\notin k(P)$, it follows that $|P| \ge 2k'$.
We claim that in this case, $k'$ is exactly the minimal critical bad value of $P$ that is at least $k$.

\begin{claim}\label{clm:next-critical-is-max-bad}
    Let $P'$ be the $k$-bad parent of $P$ and let $k'$ be the maximal $k$ such that $P'$ is the $k'$-bad parent of $P$.
    Let $c$ be the minimal critical bad value of $P$ that is at least $k$.
    If $|P| \ge 2k'$, then $k'=c$. 
\end{claim}
\begin{proof}
First, we prove that $c\le k'$ by showing that $P'$ is a $c$-bad parent of $P$.
Assume to the contrary that $P'$ is not the $c$-bad parent of $P$.
Let $c'$ be the minimal value in $\K$ such that $P'$ is not the $c'$-bad parent of $P$.
Notice that $c'\in [1.5k..k']\cap \K$.
For this values, we have that $P'$ is the $\frac{c'}{1.5}$-bad parent of $P$ but not the $c'$-bad parent of $P$. 
We also have $2\frac{c'}{1.5} \le 2c' \le 2k'\le 2|P|$, which together indicates that $c'$ is a critical bad value of $P$ that is at least $k$ and is smaller than $k'$ , a contradiction.

Now, let us show that for any $c \ge k'$.
Assume to the contrary that $k'> c$.
We have already shown that $P'$ is the $c$-bad parent of $P$.
Since $k',c\in \K$, we have $k' \ge 1.5c$.
From $P'$ being both the $k'$-bad parent and the $c$-bad parent of $P$, \cref{lem:bad-sandwitch} implies that $P'$ is also the $1.5c$-bad parent of $P$.
In particular, $P'$ is both the $c$ and the $1.5c$ bad parent of $P$ with $2c \le 2k'\le |P|$, a contradiction to $c$ being a critical bad value.
\end{proof}

It follows from \cref{clm:next-critical-is-max-bad} that $P'$ is indeed the $k'$-bad parent of $P$ for the value $k'$ we find using $C(P)$.
According to \cref{lem:length-cricial-parent}, we have that $K(P) = k'$, and according to \cref{lem:bad-k-parent-properties} we have that $P'$ contains the index $b'-k'$.
Therefore, since $|P'| \le 1.5k'$ we have that $P'$ is contained in $[b'- 2.5k'..b'+0.5k']$, and in particular $h(P') \in [b'-2.5k'..b'+0.5k']$.
Since $h(P')$ is a multiple of $k'$, and we check every multiple of $k'$ in this range, we have found $H(h(P')) = P'$ and added it to $\mathcal{P}$, as required. 
\end{proof}

\paragraph{Substituting $(\log \log z)^{2}$ with $\log \frac n z$.}
We notice that the $(\log \log z)^{2}$ factor arises from computing the deterministic hash of~\cite{DBLP:conf/icalp/Ruzic08} over $z$ elements from domain of size $n$.
We can instead partition the domain $n$ into $z$ uniform intervals $U_1,U_2,\ldots U_z$ of size $\frac{n}{z}$.
Then, we construct a separate hash $H_i$ for every $U_i$, only mapping the values in $U_i$ to their respective elements in $[z]$.
When constructing the hash of $U_i = (i\cdot \frac{n}{z}.. (i+1) \cdot \frac{n}{z}]$, we use relative values instead of absolute values (i.e., $x \in U_i$ is encoded as $x- i\cdot \frac{n}{z}$).
Therefore, the hash is over a domain of size $\frac{n}{z}$.
Now, when we wish to find $H(x)$, we first find $U_i$ containing $x$ in constant time, and query $H_i[x']$ with $x'-i\cdot \frac{n}{z}$ in constant time.
The time for constructing $H_i$ is $z_i (\log\log z_{i})^{2} = O(z_{i} \log \frac{n}{z})$, where $z_i$ is the number of hashed values within $U_i$.
Since $\sum_{i=1}^z(z_i)= z$, the total running time for constructing all hashes is $O(z \log \frac{n}{z})$ as required.
The query time remains $O(1)$.

\bibliographystyle{alpha}
\bibliography{bib}

\appendix

\section{Shift-Split-Merge Structure}
\label{sec:shift-merge}

This section is dedicated to proving the following lemma. As mentioned earlier, it has been already proven
by Farach and Thorup~\cite{DBLP:journals/algorithmica/FarachT98}, and we only present our proof for completeness.
Also, we make it more clear that the time complexity depends on $n$, the maximum number of distinct elements in a tree.
We briefly comment that Iacono and {\"{O}}zkan~\cite{DBLP:conf/icalp/IaconoO10} claimed a faster data structure, but
a complete description of the shift operation isn't present in the published version, and it is not immediately clear
if their complexity could be made dependant on the number of distinct elements.

\shiftstructure

\newcommand{\distinct}{\mathsf{distinct}}

We implement each balanced search tree with e.g. red-black trees~\cite{DBLP:conf/focs/GuibasS78}. Recall that a red-black tree $T$
has the property that, given a node $u\in T$, we can split $T$ into a tree containing all elements strictly smaller or equal to $u$ and a tree
containing all elements larger than $u$ in $O(\log N)$ time. Similarly, given two trees $T_{1}$ and $T_{2}$ such that
every element stored in $T_{1}$ is smaller or equal to every element stored in $T_{2}$, we can obtain a tree storing all elements
in $T_{1}$ and $T_{2}$ in $O(\log N)$ time. We call the former operation splitting and the latter joining (to distinguish from
merging, which is defined as in the statement of the lemma). Of course, because the tree is balanced, we are able to find
the (strict) successor/predecessor of any element in $O(\log N)$ time, and insert/delete new elements in the same time complexity.

Each node $v$ of the tree
is decorated with a shift $\delta_{v}$, and the invariant is that the, a node $u$ physically storing an element $x$ actually corresponds
to $x$ increased by the sum of $\delta_{v}$, over every $v$ that is an ancestor of $u$. It i straightforward to maintain the invariant
during a rotation. Initialize, insertion/deletion, split, and shift are immediate to implement, with the first working in worst-case
constant time, and the remaining operations working in $O(\log N)$ worst-case time. The non-trivial step is implementing a merge.

For a tree $T$, let $\distinct(T)$ denote the number of distinct elements in $T$. We first observe that, given two trees $T_{1}$ and $T_{2}$,
we can compute their merge in $O(\min\{\distinct(T_{1}),\distinct(T_{2})\}\log N)$ time. This is done as follows.
We first retrieve the smallest elements of $T_{1}$ and $T_{2}$, denoted $x_{1}$ and $y_{1}$, respectively.
By symmetry, let us assume that $x_{1} \leq y_{1}$. We split $T_{1}$ at $y_{1}$ to obtain $T'_{1}$ and $T'_{1}$
and we similarly split $T_{2}$ at $y_{1}$ to obtain $T'_{2}$ (containing only elements equal to $y_{1}$) and $T''_{2}$.
We join $T'_{1}$, $T'_{2}$, and the result of repeating the procedure on $T''_{1}$ and $T''_{2}$.
The time complexity is $O(\min\{\distinct(T_{1}),\distinct(T_{2})\}\log N)$, and more precisely it is upper bounded by the interleave
between the distinct elements of $T_{1}$ and $T_{2}$ times $O(\log N)$. However, we will establish that the amortized complexity
is actually much better.

To analyze the amortized complexity of a merge, we use the following potential function, where to avoid clutter
each unit of potential suffices to pay for $O(\log N)$ time.
Consider a tree $T$, and let its elements (ignoring the duplicates) be $x_{1} < x_{2} < \ldots < x_{k}$. We define
potential of $T$ to be $\sum_{i} 1+\log (x_{i+1}-x_{i})$, and then the potential of the whole structure is the sum of the potentials
of its constituent trees.
It is immediate that the amortized cost of initialization is $O(1)$, the amortized complexity of insertion
is $O(\log n \log N)$ as we must account for the increase in the potential, while the amortized complexity of deletion,
split and shift is $O(\log n \log N)$ as the potential cannot increase. To analyze the change in the potential after
a merge operation, we first observe that the additional $1$ in the definition allows us to absorb the cost for the elements
that occur in both trees. Thus, we will assume that the elements stored in both trees are distinct.
To analyze the change in the potential, it is enough to analyze the situation after inserting a range of $k$ elements
$x_{i}<x_{i+1}<\ldots x_{i+k-1}$ from one tree between two consecutive elements $y_{j} < y_{j+1}$ of the other tree,
which corresponds to a single step of the merging procedure.
Then, the distance of $y_{j}$ to its successor is decreased by at least half while the distance of every other
element to its successor does not increase,
or the distance of $y_{j}$ to its predecessor is decreased by at least half while the distance of every other
element to its predecessor does not increase. 
The potential can be equivalently defined by considering the distance of each element to its successor or the distance
of each element to its predecessor, so we obtain that for at least half steps of the merging procedure the potential
decreases by at least $1$, which allows us to pay for all the steps.

\section{Eliminating Large Phrases}\label{apx:no-big-phrases}

This section is dedicated to proving the following lemma.
\nobigphrases*

We start by showing that a large phrase can be replaced by a small number of smaller phrases.
We then show how to apply this replacement iteratively and efficiently.
\begin{claim}\label{clm:replace-big-with-small}
    Let $\LZend$ be an LZ-End partition with $z$ phrases of a string $S$ such that all the phrases have length at most $\frac{n}{z}$ except for the last phrase $P_i$.
    We can partition $P_i$ into at most $\frac{2z|P_i|}{n}+1$ substrings $Q_1,Q_2,\ldots$ such that for every $Q_j=S[c_1..d_1]$, there is a phrase $P_{i_j}=S[a_{i_j}..b_{i_j}]$ for some $i_j<i$ such that $Q_j = S[b_{i_j} - |Q_j|+1 ..b_{i_j}]$.
    Equivalently, $P_i$ can be replaced with $Q_1,Q_2,\ldots $ in $\LZend$ with respective sources $P_{i_1},P_{i_2},\ldots $, resulting in a valid LZ-End partition of $S$.
\end{claim}
\begin{proof}
For $i\in [z]$, denote $cut(i) = j$ as the index of the unique phrase $P_j$ containing $b_i - \frac{n}{z}$.
We will show how to replace $P_i$ with (at most) three substrings $Q_1,Q_2,Q_3$ such that $|Q_2|,|Q_3|\le \frac{n}{z}$ and $|Q_2|+|Q_3| \ge \frac{n}{z}$.
As a consequence, we have that $|Q_1| \le |P_i| - \frac{n}{z}$.
This routine may partition $P_i$ with just $Q_2$ and $Q_3$ satisfying the above conditions, omitting $Q_1$.
To simplify presentation, we will denote that $|Q_1|=0$ if the phrase $Q_1$ was not created.

If $|Q_1| > \frac{n}{z}$, we will apply the same splitting routine on $Q_1$, replacing it with three new phrases and reducing the length of the leftmost remaining phrase by $\frac{n}{z}$ again.
After $\frac{z|P_i|}{n}$ iterations of the above approach, we must have that $|Q_1| \le \frac{n}{z}$, and we have replaced $P_i$ with a total of $\frac{2z|P_i|}{n}+1$ phrases.

We now describe the construction of $Q_1$,$Q_2$, and $Q_3$.
Let $P_j = S[a_j..b_j]$ be the source of $P_i$. 
Let $j' = cut(j)$, so the phrase $P_{j'}=[a_{j'}..b_{j'}]$ contains $b_{j}-\frac{n}{z}$.
Let $c_1 = b_i - (b_{j}-a_{j'})$ and $c_2 = b_i -(b_{j}-b_{j'})$

We define $Q_{3} = S[c_2..b_i]$ with source $P_{j}$.
Notice that $|Q_3| =b_i - c_2+1=b_j - b_{j'}+1 \le \frac{n}{z}$.
And in particular, $|Q_3| \le |P_i|$, so $Q_3$ is always a substring of $P_i$.

For $Q_2$ and $Q_1$, we have two cases.
If $c_1 \ge a_i$, we define $Q_1=S[a_i..c_1-1]$ with source $P_{j'}$, and $Q_2 = S[c_1..c_2-1]$ with source $P_{j'-1}$.
Otherwise, if $c_1> a_i$ we only define $Q_2= S[a_i..c_2-1]$ with source $P_{j'}$.

It can be easily verified that the sources we define for $Q_1$,$Q_2$, and $Q_3$ indeed satisfy the required equality, as they correspond to the correctly aligned substrings of $P_i = S[b_j- |P_i|+1..b_j]$.
Since the phrase $P_{j'}$ contains the index $b_j -\frac{n}{z}$, we have that $|Q_3| = b_j - b'_j - 1 \le \frac{n}{z}$.
Since $P_{j'}$ is a phrase with $j '<i$, we have that $|P_{j'}| \le \frac{n}{z}$.
Notice that in both cases we have $|Q_2| \le |P_{j'}| \le \frac{n}{z}$, as required.

On the other hand, notice that $|Q_2| + |Q_3| = \min (b_i - c_1, |P_i|)$.
We know that $|P_i| \ge \frac{n}{z}$. 
Since $P_{j'}=S[a_{j'}..b_{j'}]$ contains $b_j - \frac{n}{z}$, we have that $a_{j'} \le b_j - \frac{n}{z}$ which leads to $b_i - c_1 = b_j - a_{j'} \ge \frac{n}{z}$.
We have shown that both components in $\min (b_i - c_1, |P_i|)$ are at least $\frac{n}{z}$, so $|Q_2| + |Q_3| \ge \frac{n}{z}$, as required.
\end{proof}

We are now ready to prove \cref{lem:no-big-phrases}.

\begin{proof}
We process $\LZend$ from left to right.
As long as we encounter phrases smaller than $\frac{n}{z}$, we simply append them to $\LZend'$ (which is initially empty).
When we reach a phrase $P_i$ with $|P_i| \ge \frac{n}{z}$ with $O(\frac{2z|P_i|}{n}+1)$ phrases by applying \cref{clm:replace-big-with-small} on $\LZend', P_i$ (notice that $\LZend ',P_i$ is an LZ-End partition of $S[1..b_i]$ with only the last phrase larger than $\frac{n}{z}$), and add the obtained pieces $Q_1,Q_2,\ldots$ to $\LZend'$ from left to right (instead of adding $P_i$).

In order to apply the decomposition of $P_i$ into smaller phrases described in \cref{clm:replace-big-with-small}, we need to be able to access $cut(j)$, the phrase containing $b_{j}-\frac{n}{z}$ in constant time.
To support this, we compute and store $cut(j)$ for every phrase $P'_i$ that is added to $\LZend'$ as we process $\LZend$ from left to right.
With access to $cut(j)$, a straightforward implementation of \cref{lem:bad-jumps-same-phrase} will takes $O(n_i)$ where $n_i$ is the number of newly created phrases.

When we append $P'_{x}=[a'_x..b'_x]$ to $\LZend'$, we first check if $P'_{j'}$ with $j'= cut(x' -1)$ contains $b'_{x}-\frac{n}{z}$.
If it does, we set $cut(x) = j'$.
Otherwise, it must be the case that $b'_{x} - \frac{n}{z}$ is to the right of $P_{j'}$.
We check $P_{j'+1}, P_{j'+2}, \ldots$ until finally reaching the phrase $P'_c$ containing $b'_{x}-\frac{n}{z}$ and set $cut(x) = c$.
Since every phrase is only eliminated once as $cut(x)$, the overall evaluation of all $cut(x)$ values requires $O(z') = O(z)$ time (we justify $z'= O(z)$ below).

The number of new pieces created is bounded by $\sum_{P}(\frac{2z\cdot |P|}{n}+1)$ where the sum is taken over the phrases we split.
Since we only split phrases with length at least $\frac{n}{z}$ and the sum of the lengths of all phrases is $n$, the number of components in this sum is $O(z)$.
We therefore have that the total number of added phrases is bounded by $ \frac{2zn}{n} + z = 3z$.
It follows that the newly created partition has at most $4z$ phrases.

Since we create every new phrase in constant time, the time complexity is proportional to the number of phrases, which is $O(z)$.
\end{proof}

\section{Data Structure for Substring Extraction}\label{sec:extraction}

In this section, we present an efficient compressed data structure for substring extraction, proving \cref{thm:extraction}.
That is, given $[i..j]$, return $S[i..j]$.
We will present a data structure that uses $O(z)$ space and supports queries in $O(\log^2 \frac{n}{z}+ j-i)$ time.
That is, the best conceivable query time without improving upon the query time for random access presented in \cref{thm:main}.

We will assume that $j-i \le \log^2\frac{n}{z}$, and show an algorithm with running time $O(\log^2 \frac{n}{z})$.
In the general case where $j-i$ is large, we can break $S[i..j]$ into pieces of size $\log^2\frac{n}{z}$ and extract each piece separately for a total running time of $O(j-i)$.

\textbf{Intuition and Overview.}
As a primary tool, we will first present an optimal data structure with $O(z)$ space and $O(j-i)$ query time for the special case where $j$ is the ending index of a phrase $P$ (also given at query time)\footnote{Kreft and Navarro \cite{DBLP:journals/tcs/KreftN13} also introduced an algorithm for this case, but their definition of LZ-End is slightly different from the one we use, which is the common definition in recent works on LZ-End. }.
We call this data structure a \textit{suffix extraction} data structure, and we will make extensive use of this data structure in our general substring extraction data structure.

The existence of the suffix extraction data structure naturally defined an 'advantageous'  setting where general substring extraction can be done efficiently.
For instance, if $S[i..j]$ contains a phrase boundary $b$, we can extract $S[i..b]$ in $O(b-i+\log \frac{n}{z})$ time (the $O(\log \frac{n}{z})$ factor is required to find $b$ via \cref{lem:find-con-phrase}).
This allows us to assume that our input does not contain a phrase boundary

With this assumption in place, another convenient case for extraction arises.
If $r_{i} \in O(j-i)$ we can use suffix extraction to extract $S[i..i+r_i]$ in $O(j-i)$ time.
Due to our assumption that $[i..j]$ has no phrase boundaries, we have $i+r_i > j$ and we therefore have $S[i..j]$ at hand as a prefix of $S[i..i+r_i]$.

Our goal now shifts to reducing a general input $i,j$ into an input with $r_i \in O(j-i)$.
To this end, we will employ the epoch approach of the random access algorithm.
Intuitively, we would like to say that this goal is obtainable by simply running the random access algorithm on $i$, and halting the algorithm prematurely when a sufficiently small $r$-value is obtained.
Recall that the random access algorithm consists of epochs, each reducing the $r$-value to the next exponential level compared to the value seen at the start of the epoch.
Therefore, one may expect that when running a sequence of epochs on index $i$, we will encounter a sequence of indices $i = i_1,i_2,\ldots i_x$ such that if $r_{i_x}\in [k..1.5k)$, then  $r_{i_{x+1}} \in [k/1.5..k)$ (i.e., in the next smallest exponential level).
Combined with the observation that as long as $r_{i_x}>j-i$, it holds that $S[i_x..i_x+j-i] = S[i..j]$, this structure immediately leads to an extraction algorithm.
We can run the random access algorithm on $i$, and instead of halting upon reaching $i'$ with $r_{i'}=0$, we halt upon reaching $r_{i'}\in [j-i..1.5 (j-i)]$
This should occur within $O(\log \frac{n}{z})$ epochs, leading to $O(\log^2\frac{n}{z})$ running time, with an additional $O(j-i)$ running time to apply suffix extraction on the reduced instance.

Unfortunately, the above structure does not necessarily occur in our random access algorithm.
It may be the case that the $r$-value of $i_x$ is arbitrarily smaller than $r_{i_{x-1}}$, leading to an $r$-value smaller than $j-i$.
This means that following the source of the phrase containing $i_x$ will not necessarily lead to an occurrence of $S[i..j]$, making any further step following $i_x$ potentially lose some of the information we are required to extract. 

To bypass this problem, we define a modified variant of the epoch procedure.
Our epoch procedure takes indices $i,j$ as input with $r_i \in [\max(10(j-i),k)..1.5k)$ for some $k\in \K$.
The implied assumption that $r_i >10(j-i)$ can be enforced as we established in the previous discussion.

The epoch procedure has two outputs.
First, it outputs the prefix $S[i..i')$ for some index $i'\in [i..j+1]$ (notice that if $i'=i$, this component of the output is trivial).
Additionally, the epoch outputs an index $i^*$ on which we can recourse, and that makes progress towards a termination condition.
That is, the index $i^*$ has $S[i^*..j^*] = S[i'..j]$ for $j^* = i^*-i'+j$, which means that the remaining suffix to be extracted from $S[i..j]$ occurs at $i^*$.
We also demand that the index $i^*$ has $r_{i^*} \in [j-i..k)$, bringing us closer towards the termination condition of the algorithm.
The running time is $O(i'-i + \log \frac{n}{z})$.

For the special case where $i'=j+1$, we simply have that the epoch procedure extracted $S[i..j]$ in its entirety, and completing the execution of the extraction query.
In this special case, we allow the epoch to consume $O(j-i+ \log^2\frac{n}{z})$ time.
For the formal presentation of the epoch procedure, see \cref{lem:extraction-epoch}.

We implement the epoch by essentially running the epoch of the extraction procedure on $i$.
If it happen to result in a gradual decrease in $r$-value, we can simply return its output.
Otherwise, we carefully analyze the cases that arise from an $r$-value below $j-i$, and exploit the structure that arises in each of these cases to extract some prefix of $S[i..j]$ while locating an alternative $i^*$ with a sufficiently small $r$-value. 

\subsection{Suffix Extraction}.
We start by solving the special case where $j$ is an ending index of a phrase.
We assume that the query specifies the phrase $P$ that ends in index $j$.
We will show a data structure with $O(z)$ space, $O(z \log^2 \frac{n}{z})$ construction time, and $O(j-i)$ query time.

\paragraph{Preprocessing and tree structure}
In preprocessing time, we compute $\delpre(a,b)$ for every phrase $P=S[a..b]$ in $O(z\log^2\frac{n}{z})$ using \cref{lem:computepre}.
We defined for each $P=S[a..b]$ a parent phrase $P'$ as follows.
Let $b'= b - \delpre(a,b)$.
The phrase containing $b'$ is $P'$.
This defines a forest structure over the phrases of $\LZend$, we compute and store this forest $T$.
Finally, we construct an unweighted ancestor data structure for $T$.

\paragraph{Query}
Our query procedure is recursive.
Given a query $S[i..j]$ with the phrase ending at index $j$, we simply return $S[j]$ if $i=j$.

Otherwise, we recurse on the LZ-End phrases strictly contained within $S[i..j]$.
We find the phrases intersecting $S[i..j]$ in time proportional to their number by simply starting with the phrase $P=P_y$ ending in $j$ and traversing $P_{y-1},P_{y-2}\ldots$ until we reach some $P_x=[a_x..b_x]$ such that $a_x\le i$.
For every $p\in[x..y]$ seen in this process, we recursively extract $S[a_p..b_p]$ by recursing on $S[a'_p..b'_p]$ where $a'_p = a_p- \delpre(a_p,b_p)$ and $b'_p = b_p-\delpre(a_p,b_p)$.

The leftmost phrase $P_x= S[a_x..b_x]$ intersecting $S[i..j]$ requires special care 
We use an ancestor query to find the root of the tree containing $P_x$ in the forest $T$.
Notice that this must be a phrase of length $1$, so it is equal to a suffix of $S[a_x..b_x]$ (of length 1) contained within $S[i..j]$.
Further notice that every ancestor of $P_x$ is an occurrence of some suffix $S[b_x-\ell..b_x]$ of $P_x$.
Let $P^1$ be the root of the tree in $T$ containing $P_x$.
We generate the sequence $P^1,P^2,\ldots P^{x'}$ of phrases as follows.
For $k \le x$, if $|P_k| < b_x-i$ we set $P^{k+1}$ to be the child of $P^{k}$ in $T$ towards $P_x$.
Otherwise, if $|P^k|\ge b_x-i$, we set $P^{k-1}=P^{x'}$ as the last element of the sequence and terminate.
In words, this is the prefix of the path from the root to $P_x$ that is truncated before the first phrase with length at least $b_x-i$.
Notice that the lengths $|P^1|,|P^2|,\ldots$ form an increasing sequence.

Denote for every $k\in[1..x']$ the endpoints $P^k=S[a^k..b^k]$.
The algorithm recurses on $S[b^k-|P^{k+1}|..b^k - |P^k|]$.
Finally for $P^{x'}$ the algorithm recurses on $S[b^{x'}-(b_x-i)..b^{x'}-|P^{x'}|]$

\paragraph{Correctness.}
Clearly, $P_{x+1},P_{x+2},\ldots,P_y$ form a suffix of $P[i..j]$, so extracting each of them leads to the extraction of this suffix.

As for the prefix covered by $P_x$, we claim that the recursive calls induced by $P^1,P^2,\ldots P^{x'}$ are for strings that form a partition of $S[i..b_x]$.
More precisely, we claim that for every $k\in [1..x')$ it holds that $S[b_x-|P^{k+1}|..b_x-|P^k|]=S[b^k-|P^{k+1}|..b^k-|P^k|]$, and that $S[i..b_x- |P^{x'}]=S[b^{x'}-(j-b_x)..b^{x'}-|P^{x'}|]$.

This follows immediately from the following stronger property
\begin{lemma}\label{lem:pre-tree-suffix-equality}
Let $P=S[a..b]$ be a phrase in $\LZend$. 
Let $P''$ be an ancestor of $P$ in $T$ and let $P''=S[a''..b'']$ be the parent of $P'$ in $T$.
It holds that $S[b-(b'-a')..b]=S[b''-(b'-a')..b'']$.
\end{lemma}
\begin{proof}
We prove by induction on the distance from $P$ to $P'$.
For the case of distance 0, we have that $P=P'=S[a..b]$ and $P''=S[a''..b'']$ such that $b'' = b-\delpre(a,b)$.
It is immediately implies by the definition of $\delpre$ that $S[a..b]=S[a''..b'']$, as required.

For the inductive vase, we have that $P'$ is at distance $d>0$ from $P$.
Let $\hat P = S[\hat a .. \hat b]$ be the child of $P'$ towards $P$.
From the induction hypothesis, we have that $S[b-(\hat b-\hat a)..b]=S[b'-(\hat b-\hat a)..b']$.
Since $\hat b - \hat a > b'-a'$, we have in particular $S[a'..b'] = S[b'-(b'-a')..b']=S[b-(b'-a')..b]$.
By definition, $P''=S[a''..b'']$ such that $b'' = b'-\delpre(a',b')$.
It is immediately implies by the definition of $\delpre$ that $S[b''-(b'-a')..b'']=S[a'..b']=S[b-(b'-a')..b]$, as required.
\end{proof}

\cref{lem:pre-tree-suffix-equality} shows that the equality $S[b_x-|P^{k+1}|..b_x-|P^k|]=S[b^k-|P^{k+1}|..b^k-|P^k|]$ indeed holds.
Further notice that the substring $S[b^k-|P^{k+1}|..b^k-|P^k|]$ ends in a phrase boundary, as $b^k-|P^k|+1$ is the first index of $P^k$, so $.b^k-|P^k|$ is the last index of the phrase preceding $P^k$.

In conclusion, we have shown that the recursive calls we made are for strings that together form a partition of $S[i..j]$, and that each of them is smaller than $S[i..j]$.
It follows that the algorithm returns the correct answer, and that it terminates.

\paragraph{Complexity.}
During the running time of the query, we are either in a termination case of the recursion - and return $S[j]$ in constant time, or we make $R$ recursive calls, and work $O(R)$ time in total to make these calls for some integer $R>1$.

This process can be described as a tree, where the children of each recursive call are the recursive calls made by it.
As every recursive call partition the substring it is responsible for, this tree has $j-i$ leaves, and each leaf correspond to $O(1)$ time of the algorithm.
Every external node of the tree with degree $R$ corresponds to a recursive calls which is executed in $O(R)$ time.

It follows that the total running time is bounded by the number of edges in a tree with $j-i$ leafs such that all internal degrees are at least 2, which is $O(j-i)$ as required.

We conclude with the following.
\begin{lemma}\label{lem:suffix-extraction}
    Given an LZ-End factorization $\LZend$ with $z$ phrases of a string $S$ with length $n$, we can construct in $O(z)$ space and $O(z\log^2\frac{n}{z})$ time a data structure that given $i,j$, and a phrase $P$ that ends in index $j$, returns $S[i..j]$ in $O(j-i)$ time.
\end{lemma}

\subsection{The General Case}

As in the random access algorithm of \cref{thm:main}, our algorithm is partition into epochs, where each epoch may extract some prefix $S[i..i')$ of $S[i..j]$, and provide an index $i^*$ that is an occurrence of $S[i'..j]$, and $i^*$ has a significantly lower $r$-value than $i$.
An epoch is described as the following procedure.
\begin{lemma}\label{lem:extraction-epoch}
For a string $S$ with LZ-End factorization $\LZend$ of size $z$, there is a data structure supporting the following query.

\textbf{Input:} $i\le j\in [n]$ such that $r_i \in [k..1.5k)$ for $k\in \K$, and $r_i > 10(j-i)$.

\textbf{Output:} For some $i'\in [i..j]$, return $S[i..i')$ and an index $i^*\in [n]$ such that $S[i^*..i^*-i'+j] = S[i'..j]$ and $r_{i^*} \in [j-i'.. k)$.

The data structure can be constructed from $\LZend$ in $O(z\log^2(\frac{n}{z}))$.
The query time is $O(j-i + \log^2\frac{n}{z})$ if the output is $S[i..j]$.
Otherwise, the query time is $O(i'-i+\log \frac{n}{z})$.

\end{lemma}

Given \cref{lem:extraction-epoch,lem:extraction-epoch-naive}, we can prove \cref{thm:extraction}.

\begin{proof}[Proof of \cref{thm:extraction}]
For the preprocessing, we construct the data structure of \cref{lem:extraction-epoch}.
For a query $i,j$, we repeatedly apply the following.

First, we would like to enforce $r_i > j-i$.
To this end, we find the starting index $i'$ of the phrase containing $j$ using \cref{lem:find-con-phrase}.
If $i' \le i$, we already have $r_i \ge j-i$.
Otherwise, the prefix $S[i..i']$ can be retrieved in $O(i'-i)$ using \cref{lem:suffix-extraction}, and the remaining task is to retrieve $S[i'..j]$, where $r_{i'}\ge j-i'$. 
To avoid clutter, let us simply assume that $r_{i} \ge j-i$ (implicitly using $i$ in reference to $i'$ if the above reduction was required).

If $r_i \le 10(j-i)$ we can use \cref{lem:suffix-extraction} to retrieve $S[i..i+r_i]$ in $O(r_i)=O(j-i)$ time.
Otherwise, we have that $r_i > 10(j-i)$ and therefore we can apply \cref{lem:extraction-epoch}.
If \cref{lem:extraction-epoch} outputs $S[i..j]$, we are done.
Otherwise, let $k\in \K$ such that $r_i \in [k..1.5k)$.
The output of \cref{lem:extraction-epoch} is $S[i..i')$ and an index $i^*$ such that $S[i'..j]=S[i^*..i^*-i+j]$, and $r_{i^*}<k$.
Our task is therefore reduced to extracting $S[i^*..i^*-i+j]$.
We obtain this reduced instance in total $O(i'-i + \log \frac{n}{z})$ time.

When iteratively applying the above, every call for \cref{lem:extraction-epoch} results in $i'$ with  $r_{i'}\in [k..1.5k)$ for a $k\in K$ strictly smaller than the one seen in the previous call.
Therefore, we will have a total of at most $O(\log \frac{n}{z})$ calls before we finally reach $r_i \in [j-i..10(j-i)]$.
At this point, we can extract $S[i..j]$ as previously described.

Since every call to \cref{lem:extraction-epoch} retrieves a disjoint fragment of the initial $S[i..j]$, the total time contributed by the $O(i'-i)$ component of all calls to \cref{lem:extraction-epoch} is $O(j-i)$.
Since we make a total of $O(\log \frac{n}{z})$ calls, the total time contributed by the $O(\log \frac{n}{z})$ components of the calls to \cref{lem:extraction-epoch} is $O(\log^2\frac{n}{z})$.
In total, the running time is $O(j-i +\log ^2 \frac{n}{z})$, as required. 

It may also be the case that the last call we make for \cref{lem:extraction-epoch} returned the remaining suffix. 
In this case, the running time of this call has an $O(\log^2 \frac{n}{z})$ factor instead of $O(\log \frac{n}{z})$.
Since we immediately terminate when this happens, the total running time would not exceed $(j-i + \log^2 \frac{n}{z})$ 
\end{proof}

For the rest of this section, we prove \cref{lem:extraction-epoch}.
As in the random access algorithm, we split the epoch into a naive part and a stable part.
We will reuse a lot of concepts and terminology from the random access algorithm.
Let us prove some helpful properties of the $\pre$ function.

\begin{lemma}\label{lem:pre-phraeboundary-in-front}
Let $i\in [a..b]$ with $r_i \ge b-i$, and let $i'= i - \delpre(a,b)$. 
There is a phrase boundary in $[i'+b-i..i'+r_i]$.
\end{lemma}
\begin{proof}
Let us define the $\pre$-depth of a pair $(a,b)$, denoted as $d(a,b)$.
If $\pre(a,b)=a$ then the $d(a,b)=0$.
Otherwise, the $d(a,b) = d(J(a),J(b))+1$.
In words, the $\pre$-depth of $(a,b)$ is the number of recursive calls of $\pre()$ before the value of $\pre(a,b)$ is reached.

We prove the claim by induction on the $\pre$-depth of $(a,b)$.
If the $d(a,b)=0$, then $i=i'$.
If $r_{i}<b-a$ then $a$ and $b$ are not in the same phrase, and therefore $\delpre(a,b) = 0$, we have $i'=i$ and the claim holds since $i'+r_i = i+r_i$ is a phrase boundary.

If  $d=d(a,b)>0$, we have that $a$ and $b$ are contained in the same phrase $P$, and $\pre(a,b) = \pre(a^*,b^*)$ such that the $\pre$-depth of $(a^*,b^*)=(J(a),J(b))$ is $d-1$.
Let $i^* = J(i) = i - \delta_P \in [a^*..b^*]$. 
It holds due the induction hypothesis that if $r_{i^*} \ge b^*-i^*$, then there is a phrase boundary in $[i'..b^*-i^*..i'+r_{i^*}]$.
If it indeed hold that $r_{i^*}\ge b^*-i^*$, the induction hypothesis implies that there is a phrase boundary in $[i'..b^*-i^*..i'+r_{i^*}]$.
The $r$-monotonicity of $J()$ implies $r_{i^*} < r_{i}$, and it should be clear that $b-i= b^*-i^*$.
Together, the three imply that there is a phrase boundary in $[i'..b-i..i'+r_{i}]$

In the other case, if $r_{i^*} < b^*-i^*$, there is a phrase boundary in $[a^*..b^*]$, so $\pre(a,b)=a^*$ and $i^* = i'$.
Since $i'$ is obtained by following the source link of the phrase ending in $i+{r_i}$, we have that there is a phrase ending in $i'+r_{i}$, which is in the required interval $[i'+b-i..i'+r_{i}]$ if $r_i \ge b-i$.

\end{proof}

The following two are immediate corollaries of \cref{lem:pre-phraeboundary-in-front}

\begin{corollary}\label{cor:phraseboundary-M}
   Let $P\in M$ be a marked phrase.
   Let $I=[a..b]$ be an interval in the halved canonical partition of $P$, and let $i\le j$ be two indices in $[a..b]$ and let $M(i)=i'$ and $M(j)=j'$.
   If $r_{i'}<j-i$, then the phrase containing $j'$ is of length at most $r_i$. 
\end{corollary}
\begin{proof}
Since $i,j$ are contained in the same phrase, we have $r_i = j-i+r_j$.
Clearly, we have $r_j \ge b-j$ since $j+r_j$ is the right boundary of $P$, while $[a..b]$ is an interval contained in $P$. 
By definition, $j'= j - \delpre(a,b)$, so it follows from \cref{lem:pre-phraeboundary-in-front} that there is a phrase boundary in $[j'..j'+r_{j}]$.
Since $i' = i - \delpre(a,b)$, we also have $i-j=i'-j'$.

Due to our assumption that $r_{i'}<j-i$, we have a phrase boundary in $[i'..i'+(j-i)) = [i'..j')$.
We have shown that there is a phrase boundary at most $(j-i)$ indices to the left of $j'$, and a phrase boundary at most $r_{j}$ indices to the right of $j'$.
Therefore, the length of the phrase containing $j'$ is at most $r_{j} + j-i = r_i$, as required.
\end{proof}

\begin{corollary}\label{cor:phraseboundary-pre}
   Let $P= S[a..c]\in \LZend$ be phrase and let $L(P)=S[a..b]$.
   $i\le j$ be two indices in $[a..b]$ and let $sJ(i)=i'$ and $sJ(j)=j'$.
   If $r_{i'}<j-i$, then the phrase containing $j'$ is of length at most $r_i$. 
\end{corollary}
\begin{proof}
Since $i,j$ are contained in the same phrase, we have $r_i = j-i+r_j$.
Clearly, we have $r_j \ge b-j$ since $j+r_j$ is the right boundary of $P$, while $S[a..b]$ is prefix of $P$. 
By definition, $j'= j - \delpre(a,b)$, so it follows from \cref{lem:pre-phraeboundary-in-front} that there is a phrase boundary in $[j'..j'+r_{j}]$.
Since $i' = i - \delpre(a,b)$, we also have $i-j=i'-j'$.

Due to our assumption that $r_{i'}<j-i$, we have a phrase boundary in $[i'..i'+(j-i)) = [i'..j')$.
We have shown that there is a phrase boundary at most $(j-i)$ indices to the left of $j'$, and a phrase boundary at most $r_{j}$ indices to the right of $j'$.
Therefore, the length of the phrase containing $j'$ is at most $r_{j} + j-i = r_i$, as required.
\end{proof}

We now present an adjusted version of the naive part of the epoch from the random access algorithm (\cref{lem:naive}). 
\begin{lemma}\label{lem:extraction-epoch-naive}
Given an LZ-End factorization $\LZend$ with $z$ phrases of a string $S$ with length $n$, we can construct in $O(z \log^2 \frac{n}{z})$ time a data structure taking $O(z)$ space that supports the following query.

Given $i,j\in [n]^2$ with $r_{i} \in [\max(10(j-i),k)..1.5k)$ for some $k\in \K$, return an occurrence $S[i'..j']$ of $S[i..j]$ (specified by its endpoints) satisfying one of the following.
\begin{enumerate}
    \item $r_{i'} \in [j-i..k)$,
    \item There is a phrase boundary in $S[i'..j']$, and $\ell_{j'}  + r_{j'} + 1< 1.5k$, or
    \item $S[i'..j']$ is contained some a marked phrase and $r_{i'} < 1.5k$.
\end{enumerate}

The running time of the query is $O(\log \frac{n}{z})$.
\end{lemma}
\begin{proof}
In preprocessing time, we construct the data structure of \cref{lem:naive}.

In query time,we run the query algorithm of \cref{lem:naive} with the following modification.
Recall that the query of \cref{lem:naive} applies steps of the form $J(i)$, $J^{\log \frac{n}{z}}(i)$ until the first time an index $i'$ is reached that either has $r_i < k$ or is contained within a marked phrase.

We apply exactly the same steps, but we stop when we either reach $r_{i'}<\max(k,j-i)$ or then we reach an $i'$ contained in a marked phrase.
Since our stopping condition is strictly stronger than that of \cref{lem:naive}, the running time until one of our conditions is met is bounded by the query time of \cref{lem:naive}, which is $O(\log \frac{n}{z})$ (notice that checking that these conditions are satisfied is done in constant time, as we have access to $r_{i'}$ for every index $i'$ visited throughout the query).

Let $i'$ be the first index seen throughout the query satisfying $r_{i'}<\max(k,j-i)$ or $i'$ is in a marked phrase.
Firstly, we wish to ensure that $i'$ is an occurrence of $S[i..j]$.
Let $i''$ be the last index seen before $i'$ in the execution of the query.
It must hold that $r_{i''} > j-i$, and since $r$-values are monotonic, it holds that the $r$-values of all seen indices preceding $i''$ is also at least $j-i$.
This implies that $S[i''..i'' + j-i] = S[i..j]$.
If $i' = J(i'')$, it follows from the definition of $J$ that $S[i'..i'+j-i] = S[i'' ..i'' +j-i] = S[i..j]$.

Otherwise, if $i'= J^{\log \frac{n}{z}}(i'')$, it is not necessarily the case that $S[i'..i'+ j-i] = S[i..j]$.
However, if we iteratively apply $J(i''),J^2(i''),\ldots$ until we finally reach an index $i^*$ with $r$-value less than $\max(k,j-i)$, then $i^*$ will be an occurrence of $S[i..j]$ due to the same argument as before.
Notice that this can be implemented in $O(\log \frac{n}{z})$ time, since we are guaranteed to meet this condition at $J^{\log \frac{n}{z}}(i'')$ (and all of those steps correspond to an upward path in $T_{\Bad}$, since $i'$ was in the $(\log \frac{n}{z})$-bad interval of its containing phrase).

So far, we have shown a procedure that finds an occurrence $S[i'..j']$ of $S[i..j]$ with either $r_{i'} < \max(j-i,k)$ or $i'$ in a marked phrase. 
Let us show that this $S[i'..j']$ is a valid output, satisfying one of the three conditions specified in the statement of the lemma.

If $r_{i'} \in [j-i..k)$, then it is a valid output (satisfying the first condition in the statement of the lemma).
If $r_{i'}<j-i$, then $[i'..j']$ contains a phrase boundary.
Furthermore, recall that $r_{i'}\le r_i < 1.5k$.
Therefore, there is a phrase boundary at most $1.5k$ indices to the right of $i'$, and since $j'$ is trapped between the phrase boundary in $[i'..j']$ and the phrase boundary in $[i'..i'+1.5k)$, the phrase containing $j'$ is of size less than $1.5k$, and we have that $S[i'..j']$ satisfies the second condition statement of the lemma.

If $i'$ is contained in a marked phrase, we have two cases.
If $[i'..j']$ is fully contained in the marked phrase, we have that the last condition of the lemma is satisfied.
Otherwise, we have a phrase boundary in $[i'..j']$, and due to the same argument as before, we have that the phrase containing $j'$ is of size less than $1.5k$.
Therefore, $S[i'..j']$ satisfies the second condition of the lemma in this case.

The running time consists of the $O(\log \frac{n}{z})$ running time of the query of \cref{lem:naive}, plus possibly $O(\log \frac{n}{z})$ time to compute $J^{\log \frac{n}{z}}(i'')$ (or some lower number of $J$ applications) step-by-step.
The total running time is $O(\log \frac{n}{z})$, as required.
\end{proof}

Similarly to the random access algorithm, we have that the naive part, specified above as \cref{lem:extraction-epoch-naive}, may already return a satisfactory output as $S[i'..j'] = S[i..j]$ with $r_{i}\in [j-i..k)$.
We prove the following useful properties of $sJ()$.

\begin{lemma}\label{lem:stable-propreties-extraction}
    Let $P=S[a..b]$ be a phrase with bad parent $P'=S[a'..b']$.
    Let $i \in [a..b]$ such that $i'= sJ(i)$ is well-defined.
    It holds that $|P'| \le |P|$.
    Additionally, if $i'\notin [a'..b']$, it holds that $r_{i'} < 2/3|P|$, and if $i' \in [a'..b']$, then $\ell_{i'}< \ell_i$.
\end{lemma}
\begin{proof}
Denote $a_p = sJ(a)$ and $b_p = sJ(a+\floor{\frac{2}{3}|P|}-1)$.
Recall that $a'$ is the rightmost phrase boundary in $[a_p..b_p]$.
Also recall that applying $sJ$ never increases the $r$-value, so $r_{b_p} \le r_{a+\floor{\frac{2}{3}|P|}-1} = \ceil{|P|/3}$.
We have shown that there is a phrase boundary in $[b_p..b_p+\ceil{|P|/3}] = [a_p + \floor{\frac{2}{3}|P|}-1 .. a_p + |P|]$, and that $a'\in [a_p..b_p]$.
Since $b'$ is the first phrase boundary to the right of $a'$, this necessarily means $b'< a_p+|P|$ which in turn leads to $b'< a'+|P|$.
It follows that $|P'| \le |P|$.
Furthermore, notice that $a'>a_p$, as $a_p$ and $b_p$ are not in the same phrase (by the definition of $\pre$).
Therefore, we have $\ell_{i'}<\ell_i$ if $i' \in [a'..b']$. 

If $i'$ is not in $P'$, than the $r$-value of $i'$ is decreased by at least $\ceil{|P|/3}$ compared to $i$ which had $r_i \le |P|$.
Therefore, $r_{i'} \le r_i - |P|/3 \le \frac{2}{3}|P|$.
\end{proof}

We proceed to show that in each of the second and the third cases for the output of \cref{lem:extraction-epoch-naive}, we can apply an additional procedure to obtain a valid output for \cref{lem:extraction-epoch}.
We start by providing the following subroutine.
Intuitively, we show that if $[i..j]$ contains a phrase boundary, and $j$ is contained in phrase of length less than $1.5k$, we can efficiently return a valid output for the epoch.

\begin{lemma}\label{lem:extraction-stable}
There is a data structure with $O(z)$ space for the following query.
Given $i,j$ satisfying:
\begin{enumerate}
\item There is a phrase boundary in $S[i..j]$, and 
\item  $\ell_j+r_{j}-1< 1.5k$ for $k\in \K$, and $k>2(j-i)$,
\end{enumerate}
return for some $i'\in [i..j+1]$ the string $S[i..i')$ and $i^*,j^*$ such that $S[i'..j] = S[i^*..j^*]$ and $r_{i^*} < k$.

The query time is $O(i'- i + \log \frac{n}{z})$
\end{lemma}
\begin{proof}
At preprocessing, we apply exactly the preprocessing of \cref{lem:stable} to obtain the tree $T_{\Bad}$ such that the parent of each $P$ is its bad parent $P'$, and the $\pre()$ values of $L(P)$ for every phrase $P \in \LZend$.

Given a query $i,j$ with $k\in \K$ satisfying the required properties, we  first find the phrase $P = S[a..b]$ containing $j$ using \cref{lem:find-con-phrase}, and extract $S[i..a-1]$ using \cref{lem:suffix-extraction} in $O(a-i)$ time.

Notice that $|P| = \ell_j + r_j -1< 1.5k$.
If $|P| \le 3(j-i)$, we can extract $P= S[a..b]$ in $O(j-i)$ time using \cref{lem:suffix-extraction} and obtain $S[a..j]$ as a prefix of $S[a..b]$, completing the extraction of all $S[i..j]$ (formally, we also return an arbitrary phrase boundary as $i^*$ and $j^* = i^*-1$. This is valid since $i'= j+1$ and $S[i'..j]=\varepsilon$).

Otherwise, we have $r_j = b-j = (b-a)- (j-a) \ge 3(j-i) - (j-i) = 2(j-i)$, and $\ell_j = j-a \le j-i \le 2r_j$.
It follows from \cref{obs:sJ-defined-nice} that $sJ(j)$ is defined.
We repeatedly apply $sJ()$ to obtain the sequence $j_0 = j = sJ^0(j), j_1 = sJ(j),j_2=sJ^2(j),\ldots j_x = sJ^X(j)$ until we reach some $j_X$ that is not in the bad suffix of the phrase containing it.
For every $x\in [0..X]$, denote as $P_x$ the phrase containing $j_x$.
Notice that we have access to the bad parent $P_{x+1}$ of $P_x$ via the tree $T_{\Bad}$, so each of those steps is carried in $O(1)$ time.
As we compute the sequence sequence $j_0,j_1,\ldots$, we will show that the following invariants are preserved.
\begin{enumerate}
    \item $\ell_{x+1}<\ell_x$ and $|P_{x+1}|\le|P_x|$ for every $x \in [0..X-1]$.
    \item $S[j_{x} - \ell_{j_x} .. j_{x}] = S[j_{x+1}-\ell_{j_x} .. j_{x+1}]$ for every $x \in [0..X-1]$.
    \item After reaching $j_x$, we have already extracted $S[i..j-\ell_{j_x})$ for every $x \in [0..X-1]$.  
\end{enumerate}

Initially, the invariant is satisfied because $j_0 = j$, and $j-\ell_j = a-1$, and we already extracted $S[i..a-1]$.
It follows directly from \cref{lem:stable-propreties-extraction} that $|P_{x+1}| \le |P_{x}|$ and $\ell_{x+1} < \ell_{x}$ for every $x\in [0..X-1]$. 
The second invariant follows from the definition of $sJ(j_x)$.
For every $x\in [0..X-1]$, the value $sJ(j_x)$ is defined which implies that $j_x$ is in $L(P_X)$, the leftmost 2/3 of $P_x$.
Denote $L(P_x)=S[a_x..b_x]$, and notice that $a_x = j_x - \ell_{j_x}$.

The function $sJ(j_x) =j_x - \delpre(P_x)$ maps $j_x$ to an occurrence of $L(P_x)$ such that $j_{x+1}$ is aligned with $j_x$ within this occurrence.
It directly follows that  $S[j_{x} - \ell_{j_x} .. j_{x}] = S[j_{x+1}-\ell_{j_x} .. j_{x+1}]$.
Notice that the second invariant combined with the first yield $S[j_{x+1}-\ell_{x} ..j_{x+1}] = S[j-\ell_{x}..j]$, since the $\ell$ values are decreasing.

It remains to show how we maintain the extraction invariant.
When we traverse from $j_x$ to $j_{x+1}$, we notice that $S[j_{x+1}-\ell_{j_x}..j_{x+1}-\ell_{j_{x+1}}]$ is a suffix of the phrase ending right before $P_{x+1}$.
We  extract $S[j_{x+1}-\ell_{j_x}..j_{x+1}-\ell_{j_{x+1}}] = S[j-\ell_{j_x}..j-\ell_{j_{x+1}}]$ in $O(\ell_{x} - \ell_{x+1})$ time using \cref{lem:suffix-extraction}.
This concludes the maintenance of all three invariants.

When we finally reach $P_X$, we have that $j_X$ is not contained within the bad suffix of $P_X$, and that $|P_X| \le P_{X-1} \le\ldots \le P_0 = |P| <1.5k$.
If $j_X$ is not in $L(P_X)$, then $r_{j_X}<|P_X|/3 < 1.5k/3 = 0.5k$, and it also holds that $r_{j_X- \ell_{j_X}} < 0.5k + \ell_{j_X} \le 0.5k +\ell_0 \le 0.5k + (j-i)\le k$.
It is therefore valid to output $j_X - \ell_X$ as $i^*$ and $j_X$ as $j^*$.

Otherwise, we have that $j_X$ is to the left of the bad suffix, and therefore $j_{X+1} = sJ(j_X)$ is well defined.
Let $i_X = j_X- \min(\ell_{j_X},\ell_{j_{X+1}})$.
It holds that $i_X$ is also not in the bad suffix of $P_X$ (as it is still within $P_X= S[j_X-\ell_{j_X}..j_X+r_{j_X}]$, and to the left of $j_X$).
It follows that $sJ(i_X) = i'$ is well defined, and due to \cref{lem:stable-propreties-extraction} we have $r_{i'} \le \frac{2}{3}|P_x| < k$.
In particular it holds that $i' =j_{X+1} - \min(\ell_X,\ell_{X+1})$. 
If $i' = j_{X+1} - \ell_{j_X}$, we have already extracted $S[i..j-\ell_X)$, and we have $S[i'..j_{X+1}] = S[j-\ell_X..j]$ and $r_{i'}<k$, so $i'$ and $j_{X+1}$ are valid outputs for $i^*$ and $j^*$, respectively.

Otherwise, if $i' = j_{X+1} - \ell_{j_{X+1}}$, we can extract $S[j-\ell_{j_X}..j-\ell_{j_{X+1}}]$ as in the previous cases, and then $i'$ and $j_{X+1}$ are valid outputs. 

\paragraph{Query Time.}
Initially, we spend $O(\log \frac{n}{z})$ time to find the phrase $P$ containing $j$ and extract the prefix of $S[i..j]$ to the left of $P$.
When processing $j_x$, we retrieve a non empty substring of $S[i..i')$ disjoint from all previously extracted substrings in time proportional to the length of the substring, so the total time spent on extractions is $O(i'-i)$ where $S[i..i')$ is the prefix we end up extracting.
We may also spend $O(\log \frac{n}{z})$ time when the last $j_X$ is reached to find the phrase containing $j_{X+1}$, so we can compute $\ell_{X+1}$, $i_X$, and $i'$ (for the case in which this is required).
The total running time is $O(i'- i + \log \frac{n}{z})$, as required.
\end{proof}

We now show another case where we can efficiently find a valid output for the epoch.
Namely, we show that if we have an occurrence of $S[i..j]$ such that $i\in L(P)$, i.e. in the leftmost 2/3 indices of the phrase $P$ containing $i$, then we can efficiently find a valid output for the epoch.
More precisely: we either find a valid output for the epoch in $O(i'-i + \log \frac{n}{z})$, or we completely extract all of $S[i..j]$ in time $O(\log^2 \frac{n}{z})$.
The latter case correspond to the looser bound on the running time allowed in the statement of \cref{lem:extraction-epoch} for the case in which the epoch extracts all of $S[i..j]$.

\begin{lemma}\label{lem:stable-extraction-inleft}
    Given an LZ-End factorization $\LZend$ with $z$ phrases of a string $S$ with length $n$, we can construct in $O(z\log^2\frac{n}{z})$ time a data structure with $O(z)$ space supporting the following query.
    
    Given $i,j$ such that $i$ and $j$ are in the same phrase $P$, $j-i\in[0..\log^2\frac{n}{z})$,  $i$ is in $L(P)$, and $r_i \in [\max(k,10(j-i))..1.5k)$ for some $k\in \K$, output one of the following.
    \begin{enumerate}
        \item $S[i..j]$
        \item $S[i..i')$ and indices $i^*,j^*$ such that $S[i^*..j^*] = S[i'..j]$ and $r_{i^*} <k$. 
    \end{enumerate}
    If the first is returned, the running time is $(j-i + \log^2\frac{n}{z})$.
    If the latter is returned, the running time is $O(i'-i + \log \frac{n}{z})$.
\end{lemma}
\begin{proof}
    As preprocessing, we construct the data structure of \cref{lem:stable}.
    Additionally, for every phrase $P=S[a..b] \in \LZend$, we compute $\delpre(a,a+\floor{\frac{2}{3}|P|} + \log^2\frac{n}{z})$. 
    This can be achieved in $O(z \log^2 \frac{n}{z})$ time using \cref{lem:computepre}.

    At query time, we apply the query of \cref{lem:stable} on $i$ with $r_i \in [\max(k,10(j-i))..1.5k)$.
    Notice that since $[\max(k,10(j-i))..1.5k)$ is not empty, it must be the case that $1.5k > 10(j-i)$ which implies $k>j-i$. 
    Let us recall the output of the query of \cref{lem:stable}.
    Consider the sequence $i_0,i_1, \ldots$ of indices with $i_0 = i$ and $i_x = sJ(i_x)$, which is terminated when we reach $i_X$ with $r_{i_X}< k$.
    The query retrieves this $i_X$ and also $i_{X-1}$.
    Denote $j_{X-1} = i_{X-1}-i+j$ and $j_X = i_X - i+j$.
    Since $r_{i_{X-1}} \ge k > (j-i)$, we have that $S[i_{X-1}..j_{X-1}] = S[i..j]$.
    Furthermore, since $i_X = sJ(i_{X-1})$ is well defined, we have that $i_{X-1}$ is in $L(P')$ where $P'$ is the phrase containing $i_{X-1}$, and that $i_{X}$ is within an occurrence of $L(P')$.
    Denote $L(P) = S[a..b]$ and let $S[a'..b']$ be the occurrence of $S[a..b]$ containing $i_X$.
    We consider three cases depending on the $r$-value of $i_X$ and on whether or not $j_X \in L(P')$.

\paragraph{Case 1: $r_{i_X} \in [j-i..k)$}
    In this case, it holds that $S[i_X..i_X +j-i]=S[i_{X-1} .. i_{X-1} + j-i]= S[i..j]$.
    We can therefore return $i_X$ as a valid $i^*$ output and extract nothing.
\paragraph{Case 2.A:  $r_{i_X} < j-i$, and $j_{X-1}<b$}
    In this case, $j_{X-1}$ is also contained in $L(P')=S[a'..b']$, then it holds that $S[i_X..j_X] = S[i_{X-1}..j_{X-1}] = S[i..j]$.
    Since $r_{i_X}<j-i$, there is a phrase boundary $\hat a = i_X+r_{i_X}$ such that $\hat a \in S[i_X..j_X)$.
    It follows from \cref{cor:phraseboundary-pre} that the phrase containing $j_X$ is of length at most $r_i$ which is less than $1.5k$, and we can apply \cref{lem:extraction-stable}.
This yields some prefix $S[i_X..i'_X) = S[i..i')$ and indices $i^*,j^*$ such that $S[i^*..j^*] = S[i'_X..j_X]= S[i'..j]$ in $O(i'-i + \log\frac{n}{z})$ time, which is a valid output.

\paragraph{Case 2.B: $r_{i_X}<j-i$ and $j_{X-1}>b$}
Recall our assumption that $j- i \le \log^2\frac{n}{z}$.
Notice that $j_{X-1} - i_{X-1} = j-i$, so we have $j_{X-1} - i_{X-1}\le \log^2\frac{n}{z}$.
It follows from $i_{x-1}\le b$ that we have that $j_{X-1} < b+ \log^2\frac{n}{z}$.
Recall that $\pre(a,b)$ is defined by a sequence $(a,b)= (a_0,b_0),(a_1,b_1),\ldots ,(a_d,b_d)$ such that if $a_y$ and $b_y$ are not in the same phrase, then $(a_y,b_y) = (a_d,b_d)$ and otherwise $a_{y+1},b_{y+1} = (J(a_y),J(b_y))$.
It should be clear that $\pre(a,b+\log^2\frac{n}{z})$ will be defined by a prefix of this sequence, where it would only be a proper prefix if for some $y \in [0..d)$, it holds that $b_y$ and $b_y+\log^2 \frac{n}{z}$ are in different phrases.
If it happens to be the same sequence, then $S[i_X..j_X]=S[i_{X-1}..j_{X-1}]=S[i..j]$ and the same argument as before follows.
Otherwise, we have some $i'_X = i_{X-1} - \delpre(a,b+\log^2\frac{n}{z})$ such that $S[i'_X..j'_X]=S[i..j]$, and there is a phrase boundary at most $\log^2\frac{n}{z}$ to the right of $j'_X$. 
If this phrase boundary is to the left of $j'_X$, we are in the same situation as in Case 2.A.
Otherwise, we can retrieve $S[i..j]$ in $O(j-i + \log^2\frac{n}{z})$ time using \cref{lem:suffix-extraction}.
\end{proof}

We are finally ready to put it all together and prove \cref{lem:extraction-epoch}.
\begin{proof}[Proof of \cref{lem:extraction-epoch}]
In preprocessing time, we construct the data structures of \cref{lem:extraction-epoch-naive,lem:extraction-stable}.

Given a query $i,j$ with $r_i \in [\max(10(j-i),k)..1.5k)$, we apply a query of \cref{lem:extraction-epoch-naive} to obtain an occurrence $S[i'..j']$ of $S[i..j]$ satisfying one of the three conditions specified in \cref{lem:extraction-epoch-naive}.
We proceed as follows based on the satisfied condition.

\paragraph{Case 1: $r_{i'} \in [j-i..k)$.}
In this case, we can simply output $i'$.
\paragraph{Case 2:  There is a phrase boundary in $S[i'..j']$, and $\ell_{j'}+r_{j'}+1 < 1.5k$.}
In this case, we apply \cref{lem:extraction-stable} to obtain $S[i'..i'')$ for some $i'' \in [i'..j'+1]$ and an index $i^*$ with $S[i^*..j^*] = S[i'' .. j']$ in $O(i'' - i' + \log \frac{n}{z})$ time.
Notice that this is a valid call, since $r_i \in [10(j-i)..1.5k)$ and therefore $1.5k > 10(j-i)$ which in particular implies $k > 2(j-i)$.
The extracted prefix and the index $i^*$ are a valid output, and the running time is as required.

\paragraph{Case 3:  $S[i'..j']$ is contained in a marked phrase and $r_{i'} < 1.5k$.}
Let $S[x..y]$ be the canonical interval containing $i'$, and let $M(i') = \hat i$.
Notice that $S[\hat i..\hat i+r_{\hat i}] = S[i..i+r_{\hat i}]$.
We consider three sub-cases.
\begin{enumerate}
    \item \textbf{Case 3.A : $r_{\hat i} \in [k..1.5k)$.} In particular, we have that $S[\hat i .. \hat j] = S[i..j]$ for $\hat j = \hat i+j-i$ (since $k > j-i$). 
    It follows from $r_{\hat i}\ge k$ and from \cref{clm:m-reduce-l-or-r} that $\ell_j <k$, which leads to $\hat i \in L(\hat P)$ where $\hat P$ is the phrase containing $\hat i$. 
    We can apply \cref{lem:stable-extraction-inleft} to $\hat i,\hat j$ to obtain $S[\hat i..\hat j] = S[i..j]$ in $O(j-i+ \log^2 \frac{n}{z})$ time or to obtain $S[i..i')$ and $i^*,j^*$ that are valid outputs in $O(i'-i + \log \frac{n}{z})$ time.
    \item \textbf{Case 3.B : $r_{\hat i} \in [j-i..k)$.} In this case we can simply return $\hat i$ and extract nothing.
    \item \textbf{Case 3.C.I : $r_{\hat i} <j-i$, and $j'\in[x..y]$}. In this case, we  have $S[\hat i ..\hat j]=S[i..j]$ and $M(j) = \hat j$.
    This follows from the fact that $M()$ maps $i'$ to an occurrence $S[\hat x..\hat y]$ of $S[x..y]$, which contains an occurrence of $S[i..j]$.
    Since $M(i') = \hat i$ and since $r_{i'}<1.5k$, it follows from \cref{cor:phraseboundary-M} that the phrase containing $\hat j$ has size less than $1.5k$ and we can apply \cref{lem:extraction-stable} to obtain a valid output. 
    \item \textbf{Case 3.C.II : $r_{\hat i} <j-i$, and $j'\notin[x..y]$}.
    Here we find the rightmost phrase boundary $a'$ in $[x'..y']$, and check what is $r_{a'}$.
    If $a'+ r_{a'}-\hat i <j-i$, we extract $S[\hat i..a' + r_{a'}] = S[i..i+r_{a'} + a'-\hat i]$ in time linear to the length of this prefix using \cref{lem:suffix-extraction}. 
    After that, we proceed to extract the remaining suffix using the next interval $[x^*..y^*]$ that contains non-extracted indices. 
    If $a'+ r_{a'} - \hat i > j-i$, we are in the same case as 3.C.I as the phrase containing $a'$ must have length less than $1.5k$ due \cref{cor:phraseboundary-M}.
\end{enumerate}

In each of the cases above, we either return a valid output directly, or delegate to \cref{lem:extraction-stable,lem:stable-extraction-inleft} to return a valid output in $O(\log\frac{n}{z})$ plus time proportional to the length of an extracted prefix (or possibly $O(\log^2\frac{n}{z})$ if the entire string is extracted).

The only exception is Case 3.C.II where we extract some prefix in time proportional to its length and then proceed to extract the remaining suffix.
In this case, notice that we will remain in Case 3.C.II, and proceed to extract non-empty prefixes in time proportional to their length, or eventually land in another one of the subbases of Case 3.
As long as we remain in Case 3.C.II, the running time can be charged on the lengths of the extracted prefixes.
\end{proof}

\end{document}